\documentclass[british,a4paper,11pt]{article}
\usepackage{babel,amsmath,amssymb,amsthm}
\usepackage{bm}
\usepackage[sans]{dsfont}
\usepackage[mathscr]{euscript}
\usepackage{times}
\usepackage{cite}
\usepackage{xcolor}
\usepackage{mathtools}
\usepackage[all]{xy}
\usepackage{float}
\usepackage{hyperref}
\usepackage{a4wide}
\theoremstyle{plain}
\newtheorem{theorem}{Theorem}
\newtheorem{proposition}[theorem]{Proposition}

\theoremstyle{remark}
\newtheorem{remark}{Remark}

\theoremstyle{definition}
\newtheorem{definition}{Definition}
\theoremstyle{definition}
\newtheorem{assumption}{Assumption}

\newcommand{\Var}{\operatorname{Var}}

\newcommand{\esssup}{\operatorname{ess\,\sup}}

\newcommand{\Tr}{\operatorname{Tr}}
\newcommand{\rmd}{\mathrm{d}}
\newcommand{\rme}{\mathrm{e}}
\newcommand{\rmi}{\mathrm{i}}
\newcommand{\RE}{\mathrm{\,Re\,}}
\newcommand{\IM}{\mathrm{\,Im\,}}

\newcommand{\Rbb}{\mathbb{R}}
\newcommand{\Cbb}{\mathbb{C}}

\newcommand{\Pbb}{\mathbb{P}}
\newcommand{\Nbb}{\mathbb{N}}

\newcommand{\Qbb}{\mathbb{Q}}
\newcommand{\Ebb}{\operatorname{\mathbb{E}}}

\newcommand{\openone}{\mathds{1}}

\newcommand{\ind}{\mathtt{1}}
\newcommand{\id}{\mathrm{Id}}

\newcommand{\norm}[1]{\left\Vert#1\right\Vert}
\newcommand{\abs}[1]{\left\vert#1\right\vert}

\newcommand{\Acal}{\mathcal{A}}

 \newcommand{\Dcal}{\mathcal{D}}
 \newcommand{\Ecal}{\mathcal{E}}

\newcommand{\Gcal}{\mathcal{G}}

\newcommand{\Ical}{\mathcal{I}}
\newcommand{\Jcal}{\mathcal{J}}

\newcommand{\Lcal}{\mathcal{L}}

 \newcommand{\Ocal}{\mathcal{O}}
  \newcommand{\Pcal}{\mathcal{P}}
 
\newcommand{\Scal}{\mathcal{S}}
\newcommand{\Tcal}{\mathcal{T}}

\newcommand{\Ycal}{\mathcal{Y}}

\newcommand{\Bscr}{\mathscr{B}}

 \newcommand{\Escr}{\mathscr{E}}
\newcommand{\Fscr}{\mathscr{F}}

\newcommand{\Hscr}{\mathscr{H}}

\newcommand{\Sscr}{\mathscr{S}}
\newcommand{\Tscr}{\mathscr{T}}
\newcommand{\Uscr}{\mathscr{U}}

 \newcommand{\Eo}{\mathsf{E}}
\newcommand{\Fo}{\mathsf{F}}

\begin{document}
\title{Quantum jump trajectories, hybrid systems, non-Hermitian evolutions, quantum/classical walks}

\author{Alberto Barchielli \footnote{also
Istituto Nazionale di Alta Matematica (INDAM-GNAMPA)} \\ Istituto Nazionale di Fisica Nucleare (INFN), Sezione di Milano, Italy }

\maketitle

\begin{abstract}
Stochastic differential equations have been used to give the evolution of quantum open system and to formulate a theory of quantum measurements in continuous time; often the name of quantum trajectory theory is used. Here we develop the general theory of stochastic master equations of jump type, when also memory is involved. The quantum system together to the counting processes involved in the stochastic master equation constitute an hybrid quantum/classical system. The construction of the physical probability law is given; it turns out to depend on the full dynamics. The role of conditional states and mean states is discussed.
By introducing the notion of ``typical trajectory", we show how to recursively construct the solution of the non-linear stochastic master equation (the conditional state). Moreover, by the notion of ``exclusive probability densities" we can describe all the probabilities related to the jumps, in particular, the waiting times of the jumps and their probability distributions. 

This general formulation and the idea of hybrid system allow to unify and generalize different fields: evolutions under non-Hermitian Hamiltonians, piecewise dynamics, various types of  quantum/classical random walks. Non-Hermitian effective Hamiltonians naturally appear in the time evolution in between two jumps; they give in a consistent way the quantum evolution and the probability distribution of the associated survival time. We show how we can have a large variety of probability densities for the survival time when we are at an ``exceptional point" or near to it. Models given by a unitary  dynamics interrupted at random times by a quantum channel have been introduced in solid state physics and statistical mechanics to describe dissipative evolutions; these models have been also used to give non-Markovian evolutions for quantum open systems. Here we show how all these models can have a dynamical formulation by means of suitable jump stochastic master equations and how to use the exclusive probability densities in this context. Finally, we restrict our approach to the case in which the dynamics of the full hybrid system is Markovian, but not the evolution of the quantum component alone. In this way we include other classes of models which were proposed as examples of evolutions with memory, such as the ``Lindblad rate equation". By restricting the possible values of the classical component to the vertices of a graph, we include also the so called ``continuous time open quantum walks". Again, in the formulation by means of stochastic master equations, also the probability distributions of the waiting times of the walk are included. A two-dimensional example is discussed to show the variety of possible behaviours.

\end{abstract}

\vspace{2pc}

\tableofcontents

\section{Introduction} 
Quantum trajectory formalism was developed originally to give a formulation to the theory of measurements in continuous time and to quantum filtering theory \cite{Bar86,Bel88,Bel89,BarB91,Car93,Hol01,BarP96,Bar06,BarG09,WisM10,Jac14,ZolG97,Caves+18,KumM03,Bar93,Car08,Maa24,CavJ23,Potts+24,Jak25}. As the observed output can be seen as a classical signal, quantum trajectories became also part of the theory of hybrid quantum/classical systems \cite{DGS00,Dio14,Dio23,ManRT23,Pomar+23,Opp+23,DamWer23,Bar23,LOW24,BarW24,Bar24,Sergi+23},  and of quantum feedback and control \cite{Potts+22,Till24,AlbGen23,WisM10,BarG12,Gough20,Potts+25,TicNA13,TicAlt12}. Quantum jump trajectories and the associated \emph{stochastic Schr\"odinger equations} and \emph{stochastic master equations} (SME) had applications in quantum optics (typically counting of emitted photons) \cite{Dav69,Dav76,SriD81,SriD82,Bar90,BarB91,Car93,BarP96,Car08,ZolG97,WisM10,Bar97,Potts+24,Potts+25,Bar06} and in the general theory of quantum open systems \cite{KW20,DM-G22,CLPSm25,SLCSPVac24}; these equations have been used also to ``unravel" quantum master equations \cite{WisM10,Car08} and they were  connected to the ``quantum Monte Carlo wave function method" \cite{BarP96,MCD93,RadLB24}.
Even if restricted to the jump case, quantum trajectories can be used to unify various subjects in quantum open system theory, quantum information, quantum statistical mechanics\ldots The aim of this article is to give a general formulation of jump type SMEs, suitable for applications in different subjects, and to show how to unify some research fields.

In Sec.\ \ref{sec:sigmat}, after some mathematical preliminaries,  we introduce the jump type SMEs, based on general point processes \cite{DaVJ08} (the types of jumps are not only discrete, but they can have even continuous labels); also feedback and external noise are included. An important point is the construction of the physical probability, which depends also on the dynamics of the quantum system; all the probabilities are connected to ``positive operator valued measures" and the standard structure of a quantum theory is respected. The notions of \emph{conditional state} (the solution of the SME) and of \emph{mean state} (the mean on the random jumps and any past noise) are discussed. Connections with continuous monitoring of optical systems and with quantum/classical hybrid systems are given. 

In  Sec.\ \ref{sec:ttrajs} the notion of \emph{typical trajectory} is introduced; by exploiting this notion it is possible to construct in a recursive way the solutions of the SMEs. Then, it is shown how the quantum dynamics influences the probabilities of the jumps. Moreover, the theory naturally introduces the probability law of the ``waiting time of the next jump" 
and this allows for connections with the subject of ``time measurements" in a quantum theory; we have a natural framework for treating waiting times, hitting times, survival times\ldots \cite{HBZT25,RGMor25,LMMR25}. Finally, we introduce the \emph{exclusive probability densities} \cite{SriD81}, which play a key role in the construction of all the probabilities of times of jumps, type of jumps,  \ldots; they give the ``full counting statistics" of the system.

In the last years there was an intense activity in the field of evolutions generated by non-Hermitian Hamiltonians \cite{SerZ13,AGU20,Gao+15,SCr15,MMCN19,NAJM19,NonS25,NonVs25,Ho+26,DMD26,Longhi25,DBernardinD21}. In the theory of jump trajectories, the non-Hermitian Hamiltonians appear naturally in between two jumps. Section \ref{sec:nonH}  is devoted to show how this field can be included in the trajectory approach. The main point is that now also the problem of the survival of the non-Hermitian system can be attached in a consistent way. In the literature it was underlined the relevance of the ``exceptional points" in the spectrum of a non-Hermitian Hamiltonian; here we compute the distribution of the survival times in these exceptional points and near them. These probability distributions are determined by the non-Hermitian Hamiltonian and show a variety of structures.

Unitary (or dissipative) dynamics interrupted at random times by non-unitary transformations is a technique developed mainly in solid state physics to get very general dissipative evolutions \cite{GoHa06,ScLu07,DDG22,NaGu23}. It has been also a recipe to obtain classes of non-Markovian quantum dynamics (quantum semi-Markov processes, quantum renewal processes) \cite{Bud04,Vacc+11,Vacc13,Vacc20,Vacc+21}. In Section \ref{sec:inters}, we show how this subject can be included in the theory of jump trajectories and SMEs. A big difference with respect to the models of the previous section is that the probabilities of the times of the jumps are predetermined and are not modified by the quantum dynamics. In our general approach we can generalize the cases treated in the literature by introducing more types of jumps, even labeled by a continuous index. This allows to describe random interactions with different types of reservoirs and even with measuring apparatuses. While the distribution of the jump times is predetermined, the probabilities of the result of the measurement or of the type of interaction can depend on the quantum dynamics. We also show how to use typical trajectories and exclusive probability densities in this context. Starting from a model with infinitely many memory revivals in the mean state \cite{Vacc+11}, we show how these revivals manifest themselves in the classical probabilities introduced in our hybrid approach.

In Sec.\ \ref{sec:walks} we present a class of Markovian hybrid systems related to some classes of quantum walks; let us stress that it is the full hybrid system to be Markovian, not the quantum component taken alone. This class of models includes some subjects already developed in the literature, such as the ``Lindblad rate equation" (or ``non-Markovian generalized Lindblad-type master equation") \cite{Bud06,BrGM06,Breuer07} and the ``continuous time open quantum walks" \cite{Pell14,CLRB17,Bri18,BBPP19,Kang19,Loeb23,Loe24a,Loe24b} . Also in this context we discuss the role of the notions of typical trajectories and of waiting time of the ``next jump". Starting from examples of non-Markovian evolutions for qbits \cite{Breuer07,BrGM06}, we develop models with a large variety of behaviours of the waiting times.

A few final comments  are given in Sec.\ \ref{sec:end}.

\section{Stochastic master equations}\label{sec:sigmat}

The dynamical stochastic differential equations (SDEs) involved in the quantum trajectory approach are based on two probability measures \cite{BarH95,BarPZ98}: a \emph{reference probability} $\Qbb$ and a \emph{physical probability} $\Pbb$. The physical probability will be constructed in Sec.\ \ref{sec:nprob}, in terms of $\Qbb$ and of the quantum/classical dynamics.
We start by introducing the reference probability and the counting processes needed in the SDEs giving the quantum dynamics. The key evolution equations (the linear stochastic master equation and the non-linear one) are given in Secs.\ \ref{sec:linSME} and \ref{sec:nlinSME}.

\subsection{Probability structures, Hilbert space, and relevant operators}

We have a filtered probability space satisfying the \emph{usual hypothesis}, typically assumed in books on stochastic calculus  \cite[Def.\ 1.1, pp.\ 3--4]{Met82}, \cite[p.\ 3, Sec.\ I.5]{Prott04}, \cite[p.\ 45]{IWat89}.

\begin{assumption}\label{Ass:processes}
Let $\big(\Omega, \Fscr, \Qbb\big)$ 
be a \emph{complete probability space} with a \emph{filtration} of $\sigma$-algebras $\left\{\Fscr_t, t\geq 0\right\}$, 
such that \ (a) \ $\Fscr= \bigvee_{t\geq 0} \Fscr_t$, \ (b) \ $\Fscr_t= \bigcap_{r> t}\Fscr_r$, \ (c) \ $\Fo\in\Fscr, \ \Qbb(\Fo)=0 \ \Rightarrow \ \Fo\in \Fscr_0$; we set also $\Fscr_{t_-}=  \bigvee_{0\leq r<t} \Fscr_r$. 
\end{assumption}

In a probability space, the $\sigma$-algebra $\Fscr$ denotes the set of all possible events, while the algebra $\Fscr_t$ contains all the events up to time $t$, and $\Fscr_{t_-}$ all the events before time $t$. We shall denote by $\Ebb_\Qbb[X]$ 
the mean value of the random variable $X$ with respect to the probability measure $\Qbb$. As a real random variable is a measurable function from $\Omega$ to $\Rbb$, we can write $X: \omega \in \Omega \to X(\omega)\in \Rbb$; then, the mean value can be written as $\Ebb_\Qbb[X]=\int_\Omega X(\omega) \Qbb(\rmd \omega)$. In this filtered probability space we shall introduce several stochastic processes, in particular \emph{regular right continuous} (RRC) \emph{processes} \cite[Def.\ 1.5]{Met82}, i.e. adapted and with right continuous paths with left limits; if $X(t)$ is RRC, $X(t_-)$ denotes the left limit at $t$. We shall need also \emph{regular left continuous} (RLC) \emph{processes} (adapted, with left continuous paths with right limits); an RLC process $X(\cdot)$ is predictable, i.e. $X(t)$ is $\Fscr_{t_-}$-measurable \cite[Sec.\ 3]{Met82}.

Then, we introduce the counting measure and the compensated counting measure \cite{DaVJ08}.
\begin{assumption}\label{ass:U} $\Uscr$ is a
locally compact Hausdorff space with a topology with a countable basis; $\Bscr(\Uscr)$ is the Borel $\sigma$-algebra of $\Uscr$. Then, in $\big(\Omega, \Fscr, \left\{\Fscr_t \right\}_{t\geq 0}, \Qbb\big)$   we introduce an adapted random point measure $\Pi(\rmd u, \rmd t)$ on $ \Uscr \times \Rbb_+$.
We ask the compensator of $\Pi$ (or stochastic intensity, or dual predictable projection) to be of the form $\nu_t(\rmd u) \rmd t$ where $\nu_t$ is a random, \emph{finite} Radon measure
\footnote{A Radon measure is finite on all compact sets; it is outer regular on all Borel sets and inner regular on open sets.}
on $\big(\Uscr, \Bscr(\Uscr) \big)$:
\begin{equation}\label{nufin}
\nu_t(\Uscr)<+\infty, \qquad \forall t\geq 0.
\end{equation}
\end{assumption}

By
\[\widetilde \Pi(\rmd u,\rmd t):= \Pi(\rmd u,\rmd t)- \nu_t(\rmd u)\rmd t
\]
we denote the \emph{$\Qbb$-compensated point measure}. Also other expressions appear in the literature; for instance, \emph{white random measure} is used in \cite[p.\ 219]{Met82}, and \emph{martingale measure} in \cite[Sec.\ 3.5]{LipS86}. 

\begin{remark} Apart from the regularity assumptions, the key restriction is Eq.\ \eqref{nufin}. By this, we have the existence of the counting processes   $\Pi(A, [0,t])$, with $A$ Borel set in $\Uscr$. In particular, it exists the process which gives the sum of all types of counts described by $\Pi(\rmd u, \rmd t)$:
\begin{equation}\label{def:N}
N(t):= \Pi(\Uscr, [0,t]).
\end{equation}
Under $\Qbb$, $N(t)$ is a counting  process of stochastic intensity $\nu_t(\Uscr)$; when $\nu_t(\Uscr)$ is not random $N(t)$ is a Poisson process.
\end{remark}

About the quantum component, we assume that it is described in a separable Hilbert space $\Hscr$; we denote by $\Bscr(\Hscr)$ the space of the bounded operators on $\Hscr$, by $\Tscr(\Hscr)$ the space of trace class operators and by $\Sscr(\Hscr)\subset \Tscr(\Hscr)$ the subset of statistical operators. Then, we introduce the (random) operators involved in the quantum dynamics.

\begin{assumption}[{\cite[Assumptions 2.0.A, 2.0.B, 2.3.A, 2.3.B, 3.1, 3.3]{BarH95}}]\label{ass:HJL} Let 
$ H(t),  \; J_j(u,t), \; L_{k}(t)$, $j=1,\ldots,d_1$, $k=1,\ldots,d_2$, be $\Bscr(\Hscr)$-valued RLC processes (continuity in the strong operator topology; eventually, $d_1=+\infty$ and/or $d_2=+\infty$). We also assume the functions $(u,t,\omega)\to J_j(u,t;\omega)$ to be strongly measurable. Then, we define the operators
\begin{equation}\label{def:R(t)}
R(t):=\sum_{j=1}^{d_1}\int_\Uscr \nu_t(\rmd u)  J_j(u,t)^\dagger J_j(u,t) \in \Bscr(\Hscr), \qquad R_0(t):= \sum_{k=1}^{d_2}L_k(t)^\dagger L_k(t)\in \Bscr(\Hscr);
\end{equation}
we ask the integral and the sums to be strongly convergent ($\forall t\in \Rbb^+$, $\forall \omega \in \Omega$). We also require, $\forall t\in \Rbb^+$,
\[
\esssup_{\omega\in\Omega}\int_0^t \norm{R(s;\omega)} \rmd s <+\infty, \qquad \esssup_{\omega\in\Omega}\int_0^t \norm{R_0(s;\omega)}\rmd s<+\infty,
\]
\[
\esssup_{\omega\in\Omega}\int_0^t \norm{H(s;\omega)}\rmd s<+\infty, \qquad \sup_{j,\; 0\leq s\leq t} \sup_{u\in\Uscr, \, \omega\in\Omega}\norm{J_j(u,t;\omega} < +\infty.
\]
\end{assumption}
Then, we define the following operators on $\Tscr(\Hscr)$:
\begin{subequations}\label{linearSME+}
\begin{equation}\label{Jcalu}
\Jcal(u,t)[\rho]:=\sum_{j=1}^{d_1}J_j(u,t) \rho J_j(u,t)^\dagger,\qquad \Gcal(t)[\rho]:= \int_\Uscr \nu_t(\rmd u)\Jcal(u,t)[\rho]
\end{equation}
\begin{equation}\label{Lcalt}
\Lcal_0(t)[\rho]:=
-\rmi[H(t),\rho]+\sum_{k=1}^{d_2} L_{k}(t) \rho L_{k}(t)^\dagger -\frac 12 \left\{R_0(t), \rho\right\},
\end{equation}
\begin{equation}\label{Lcalx}
\Lcal(t)[\rho]:=\Lcal_0(t)[\rho] +\Gcal(t)[ \rho ]-\frac 12 \left\{R(t),\,\rho\right\}.
\end{equation}
\end{subequations}
The curly brackets denote the anti-commutator: $\{A,B\}=AB+BA$; $\Lcal_0(t)$ and $\Lcal(t)$ are  generators in GKSL form \cite{GKS76,L76}, but in general they can be time dependent and random. Given an operator $\Acal$ acting on $\Tscr(\Hscr)$, its adjoint $\Acal^*$ acts on $\Bscr(\Hscr)$; then, we have
\begin{equation}\label{GtoR}
R(t)=\Gcal(t)^*[\openone];
\end{equation}
$\openone\in \Bscr(\Hscr)$ is the identity operator.

Now, it is possible to introduce the stochastic equations which determine the dynamics of the quantum component.

\subsection{The linear SME}\label{sec:linSME}
The first SDE in $\Tscr(\Hscr)$ is
the \emph{linear stochastic master equation} (SME) \cite[(10)--(12)]{Bar24}, \cite[Prop.\ 3.4]{BarH95}: 
\begin{equation}\label{linearSME}
\rmd  \sigma(t)= \Lcal(t)[\sigma(t_-)]  \rmd t
+ \int_{\Uscr} \Bigl( \Jcal(u,t)[ \sigma(t_-) ]-\sigma(t_-) \Bigr)\widetilde \Pi(\rmd u,\rmd t), \qquad \sigma(0)=\rho_0 \in \Sscr(\Hscr);
\end{equation}
eventually this equation is understood ``in weak sense", as in \cite[(3.13)]{BarH95}.
By the assumption \eqref{nufin} and the definitions \eqref{linearSME+}, 
the equation above can be written in the equivalent form
\begin{equation}\label{linearSME++}
\rmd  \sigma(t)= \Lcal_0(t)[\sigma(t_-)]  \rmd t-\frac 12 \left\{R(t)-\nu_t(\Uscr)\openone,\sigma(t_-)\right\}\rmd t
+ \int_{\Uscr} \Bigl( \Jcal(u,t)[ \sigma(t_-) ]-\sigma(t_-) \Bigr) \Pi(\rmd u,\rmd t).
\end{equation}
Here and in the following we take  $\hbar=1$.

\begin{proposition}[{\cite[Proposition 3.4]{BarH95}}] \label{theor:sigmaprop}
The SDE \eqref{linearSME}, with the initial condition $\rho_0$, has a unique solution $\sigma(t)\in \Tscr(\Hscr)$, \ \ $\sigma(t)\geq 0$.
By defining
\begin{eqnarray}\label{def:pt}
p(t):= \Tr \{\sigma(t)\}, \qquad
I(u,t):=\Tr\left\{\Jcal(u,t)\big[\rho(t_-)\big]\right\},
\\ \label{defhatsigma}
\rho(t;\omega):=\begin{cases}\sigma(t;\omega)/p(t,\omega) & \text{if } \  p(t,\omega)\neq 0,
\\ \sigma_v \ {\rm (fixed\ statistical\ operator)} & \text{if } \ p(t,\omega)= 0,\end{cases}
\end{eqnarray}
we have also 
\begin{equation}\label{eq:pt}
p(t)\geq 0, \qquad \Ebb_\Qbb[p(t)]=1,  \qquad
\rmd p(t)=p(t_-)\int_\Uscr \left(I(u,t) -1\right) \widetilde\Pi(\rmd u, \rmd t).
\end{equation}
\end{proposition}

The hypotheses given in Assumption \ref{ass:HJL} are taken as formulated in \cite{BarH95}, so that the proposition above comes directly from the results of that work. As now the diffusive component is not present in the SME and the restriction \eqref{nufin} of a finite measure holds, it is plausible that Assumption \ref{ass:HJL} can be relaxed, but we do not develop this point.

By definition $\rho(t;\omega)\in\Sscr(\Hscr)$, i.e.\ $\rho(t)$ is a random quantum state. The random quantity $I(u,t)$ is defined   so as to be $\Fscr_{t_-}$-measurable (\emph{predictable processes} \cite[Sec.\ 3]{Met82}).

\subsection{The physical probability}\label{sec:nprob}

We use $p(t)$ \eqref{def:pt} as a probability density and we define the new probability $\Pbb$ on the $\sigma$-algebra $\Fscr=\bigvee_{t\geq 0} \Fscr_t$ by
\begin{equation}\label{physprob}
\forall t\geq 0, \quad \forall \Fo\in \Fscr_t, \qquad \Pbb (\Fo) := \Ebb_\Qbb\big[p(t)\ind_\Fo\big];
\end{equation}
for a similar construction of a probability measure see \cite[Sec.\ 30.2]{Met82}. The probability $\Pbb$ represents the physical probability. For the construction of the physical probability the key point is Eq.\ \eqref{eq:pt}.

\begin{remark}
By \eqref{eq:pt}, $p(t)$ is a \emph{martingale} and this implies the consistency of the probabilities constructed in different time intervals: for $0<s<t$, $\Eo\in \Fscr_s$, one has $\int_\Eo p(t;\omega) \Qbb(\rmd\omega)=\int_\Eo p(s;\omega) \Qbb(\rmd\omega)$. This allows to extend $\Pbb$ to the whole $\Fscr$.
\end{remark}

\begin{remark}\label{rem:Gir} By the so called Girsanov transformation and generalizations \cite[Sec.\ III.8]{Prott04}, we have that, under the probability $\Pbb$, the random measure
$\Pi(\rmd u,\rmd t)$  becomes a point process with compensator $I(u,t)\nu_t(\rmd u)\rmd t$  \cite[Prop.\ 2.5, Remarks 2.6, 3.5]{BarH95}. We also introduce the notation
\begin{equation}\label{hatPi}
\widehat\Pi(\rmd u,\rmd t)=\Pi(\rmd u,\rmd t)- I(u,t)\nu_t(\rmd u)\rmd t.
\end{equation}
According to \cite[Sec.\ 31]{Met82}, $\widehat\Pi(\rmd u,\rmd t)$ is a \emph{white random measure}, under $\Pbb$.
\end{remark}

By defining
\begin{equation}\label{lambda}
\lambda(t):=\int_\Uscr I(u,t)\nu_t(\rmd u) \equiv \Tr\left\{R(t) \rho(t_-)\right\}<+\infty,
\end{equation}
we have that, under $\Pbb$, $N(t)$ \eqref{def:N} is a counting process of stochastic intensity $\lambda(t)$. 

\begin{remark}\label{rem:pastL+I}
The stochastic intensities $I(u,t)$ \eqref{def:pt} and $\lambda(t)$ \eqref{lambda} depend on the past through $\rho(t_-)$ \eqref{defhatsigma}. More\-over, they could depend on the past also through the operators $J_j(u,t)$, $R(t)$, which could be random and could depend on the past in case of feedback (see Assumption \ref{ass:HJL}). 
\end{remark}

\subsubsection{The mean state and the conditional state}\label{sec:mean+cond}

The random statistical operator $\rho (t)$ \eqref{defhatsigma} is the state of the quantum system when the whole past, represented by $\Fscr_t$, is known; the terms \emph{conditional state} or \emph{a posteriori state} are in use \cite{Bel88,Bel89,BarB91,WisM10,Jac14}. 

Instead, when the past is not known or not taken into account, the state to be attributed to the quantum state is the \emph{mean state} $\eta(t)$, or \emph{a priori state}, defined by
\begin{equation}\label{meanst}
\eta(t):=\Ebb_\Pbb\big[\rho (t)\big]=\Ebb_\Qbb[\sigma(t)];
\end{equation}
the last equality follows from the fact that $\rho (t)$ \eqref{defhatsigma} and $\sigma(t)$ differ only by the normalization factor $p(t)$ \eqref{def:pt}, which represents also the probability density in \eqref{physprob}. We can say that $\{\Pbb(\rmd \omega),\; \rho(t,\omega)\}$ is an \emph{ensemble} of quantum states; then, $\eta(t)$ is the mean state of this ensemble. 

By taking the $\Qbb$-mean of the linear SME \eqref{linearSME} we get 
\begin{equation}\label{masteq}
\frac{\rmd \eta(t)}{\rmd t}=\Ebb_\Qbb\big[\Lcal(t)[\sigma(t)]\big];
\end{equation}
this result follows from the fact that the last term in \eqref{linearSME} is a $\Qbb$-martingale (and it has zero mean) \cite[Proposition 3.2]{BarH95}. 

When the operators appearing in $\Lcal(t)$ depend on some extra source of randomness, the mean in the right hand side of \eqref{masteq} cannot be explicitly taken and, in general, there is not a closed equation for the mean state $\eta(t)$. When $\Lcal(t)$ is not random, Eq.\ \eqref{masteq} reduces to the time dependent quantum master equation $\dot \eta(t)=\Lcal(t)[\eta(t)]$, with the usual ``Lindblad structure". The role of randomness in $\Lcal(t)$ is indeed to introduce  memory effects into the dynamics \cite{BarG12,BarPP12}; in some non-Markovian spacial cases it is possible to find closed equations for $\eta(t)$ by using master equations with memory kernels \cite{Vacc20}.

\subsection{The non-linear SME}\label{sec:nlinSME}
Starting from the SDEs \eqref{linearSME++} and \eqref{eq:pt} and by applying the rules of stochastic calculus for jump processes, we obtain the \emph{non-linear SME} satisfied by the normalized random states $\rho (t)$ defined in \eqref{defhatsigma} \cite[(22)]{Bar24}, \cite[Remark 3.6]{BarH95}:
\begin{multline}\label{nonlinearSME}
\rmd \rho (t)= \Lcal_0(t)\big[\rho(t_-)\big]  \rmd t-\frac 12 \left\{R(t)-\lambda(t)\openone,\,\rho (t_-)\right\} \rmd t
\\ {}+ \int_{\Eo_t} \Bigl( I(u,t)^{-1}\Jcal(u,t)[ \rho(t_-) ]-\rho(t_-) \Bigr) \Pi(\rmd u,\rmd t),
\end{multline}
$\Eo_t$ is the random subset of $\Uscr$, defined by $\Eo_t(\omega)=\{u\in \Uscr: I(u,t;\omega)\neq 0\}$.

Equation \eqref{nonlinearSME} is known as ``quantum filtering equation". Indeed, the notions of conditional and mean states are typical of the approach to hybrid systems based on measurement and filtering \cite{Bel88,Hol01,ZolG97,BarG09,WisM10}. However, the idea of considering the full dynamics and its reduced (or marginal) version appears also in other approaches \cite{Pomar+23}.

By adding some other point processes, if needed, the linear and non-linear SMEs can be ``unravelled" in a purity preserving form, and the linear and non-linear stochastic Schr\"odinger equations are involved \cite[Sec. 2]{BarH95} \cite{WisM10}.

\subsection{The classical component}\label{sec:cl+meas}
The dynamical system described by the SMEs \eqref{linearSME++}, \eqref{nonlinearSME} can be seen as an hybrid quantum/classical system because classical stochastic processes are involved, whose probability law depends on the quantum dynamics itself. These processes are classical, which implies that they can be observed without perturbing their dynamics. However, some of these stochastic processes could represent pure noise, not observed nor controlled by the experimenter.   Various choices of the \emph{classical observables} are possible, depending on the physical situation.

A first choice of observables is to take discrete counting processes. We could take for $\Uscr$ a discrete space or to consider as monitored output the processes
\[
N_k(t)= \Pi(A_k; (0,t]), \qquad A_k\in \Bscr(\Uscr), \qquad A_k\cap A_l=\emptyset \ \ \text{for} \ \ k\neq l.
\]
The typical application in quantum optics is direct detection; $N_k(t)$ can represent the counts at a photon detector and the associated jump operator describes the emission process due to the jump between two discrete levels \cite{Bar97}. The observed output could be also a smoothed version of these counting processes.  For instance, we could  have detection with time delay and smoothing; in this case, the output $X(t)$ is 
\begin{equation}\label{outjump}
X_j(t)=  \sum_{k}\int_{(0,t]} a_{jk}(t-s)\,\rmd  N_k(s) + e_j(t).
\end{equation}
The terms $e_j(t)$ represent additive extra-noise due to the detection apparatus; the quantities $a_{jk}(t-s)$ are non-random and represent the detector response functions. A dependence on $X(s)$, $s<t$, can be introduced in the Hamiltonian $H(t)$ or in other operators appearing in $\Lcal(t)$. In this way we can realize feedback and closed loop control \cite{BarG12,WisM10,Potts+24,Potts+25}. By controlling the quantum system we have also a (partial) control of the properties of the output $X(t)$, for instance the squeezing in the output spectrum. Also the statistical properties of the counting processes $N_k(t)$ can be influenced, for instance by going from super-Poissonian to sub-Poissonian statistics \cite{BarG12,Potts+25,BarPP12}.

Another example of observables based on the counting processes $N_k$ is given in \cite{Potts+24,Potts+25}. They consider linear combinations, such as $\sum_k c_kN_k(t)$, which are used to represent entropies and other thermodynamical quantities. Also feedback is explicitly introduced by taking the Hamiltonian dependent on the last jump. 

Similarly to the case of discrete levels, we could use a continuous counting measure $\Pi$ to describe emission of photons by jumps involving energy bands. In any case, as observed classical component, we could take also the full counting measure $\Pi$, whose trajectories are described in Sec.\ \ref{sec:ttrajs}, where also possible applications are presented. 

The classical component can have also its own dynamics, as in the Markovian case studied in \cite{Bar24}. 
The classical component of the hybrid system is the stochastic process $X(t)$ taking values in the phase space $\Rbb^s$ and satisfying the  stochastic differential equation (restricted to the purely jump case) \cite[(2)]{Bar24}
\begin{equation}\label{Xprocess}
\rmd X_i(t) =c_i \big(X(t_-)\big)\rmd t
+ \int_{\Uscr} g_{i}\big(X(t_-),u\big) \Pi(\rmd u,\rmd t), \qquad i=1,\ldots, s.
\end{equation}
Also $ H(t), \; L_{k}(t),\;J_j(u,t)$ are taken to depend on the past only through $X(t_-)$. The dynamics of the full quantum/classical system can be represented by a semigroup on a suitable hybrid state space and this gives the ``Markovian" character of the dynamics \cite{BarW24,Bar24}. A class of Markovian hybrid evolutions is developed in Sec.\ \ref{sec:walks}.

In this presentation we have introduced a quantum system coupled to classical processes and classical observables; however, also approaches using only coupled quantum systems are possible. The quantum theory of measurements in continuous time can be formulated by representing the classical observables by commuting operators of a Bose field, used as ``ancilla" \cite{Bar86,Bar06,Hol01,KW20}. Similar structures are obtained by starting with measurements at discrete times and, then,  by taking the continuous limit, inside the formalism of ``repeated interactions" or ``collision models" \cite{AttP06,CLGP22,AlbGen23}.

\subsubsection{Connection with instruments and positive operator valued measures} \label{sec:POVM}
The equation \eqref{nonlinearSME} for the conditional states is non linear, but this non-linearity is due only to the normalization in the definition of $\rho (t)$ \eqref{defhatsigma}; underlying this non-linear equation there is the linear SME \eqref{linearSME}. This allows to show that the linear structure of quantum mechanics is respected and that probabilities and conditional states can be obtained from the standard notions of quantum measurement theory: \emph{positive operator valued measures} (POVM) and \emph{instruments} \cite{Hol01,WisM10,Jac14,BarG09}.

The map $\rho_0 \mapsto \sigma(t)\in\Tscr(\Hscr)$ is linear in $\rho_0$ and positive (Proposition \ref{theor:sigmaprop}); so, it can be extended to the whole trace class. By writing $\sigma(t;\rho)$ for the solution of \eqref{linearSME++} with initial condition $\rho\in \Tscr(\Hscr)$, we can define the propagator of the linear SME by
\begin{equation}\label{propaglin}
\Sigma(t):\Tscr(\Hscr)\to \Tscr(\Hscr), \qquad \Sigma(t)[\rho]:=\sigma(t;\rho).
\end{equation}
This map is linear and completely positive (CP); complete positivity follows from the form of the involved operators: Lindblad-type generators and CP jump operators. By taking the mean and the trace of the SME \eqref{linearSME} and by using linearity, we get the trace preserving property
\[
\Ebb_\Qbb\big[\Tr\left\{\Sigma(t)[\rho]\right\}\big]= \Tr\{\rho\}, \qquad \forall\rho\geq 0, \quad \rho\in \Tscr(\Hscr).
\]

As discussed in the first part of Sec.\ \ref{sec:cl+meas} some of the involved stochastic processes are pure noise and only some relevant classical observables are monitored. Then, we can consider the filtration $\Escr_t$, $t \geq 0$, generated by the monitored processes. We have $\Escr_t\subset \Fscr_t$; the $\sigma$-algebra $\Fscr_t$, containing the events up to time $t$, has been introduced in Assumption \ref{Ass:processes}.
For every $t>0$, we can now define an \emph{instrument} $\Ical_t$ on the measurable space $(\Omega, \Escr_t)$. For every $\Eo\in \Escr_t$, we define the map on $\Tscr(\Hscr)$ by
\begin{equation}\label{Icalinstr}
\forall a\in \Bscr(\Hscr), \quad \forall \rho\in \Tscr(\Hscr),  \qquad \Tr\left\{a\Ical_t(\Eo)[\rho]\right\}:=\Ebb_\Qbb\big[ \ind_\Eo\Tr\left\{a \Sigma(t)[\rho]\right\}\big].
\end{equation}
Then, $\Ical_t$ turns out to be an instrument, which means that the following properties hold
\begin{enumerate}
\item $\forall \Eo\in \Escr_t$, $\Ical_t(\Eo)$ is a CP, linear map on $\Tscr(\Hscr)$ and it is trace decreasing, i.e. $\Tr\big\{ \Ical_t(\Eo)[\rho]\big\} \leq \Tr\{\rho\}$, $\forall \rho\geq 0$.
\item (normalization) $\Tr\big\{ \Ical_t(\Omega)[\rho]\big\}=\Tr\{\rho\}$.
\item ($\sigma$-additivity) For every countable family $\{\Eo_i\}$ of disjoint sets in $\Escr_t$
\[
\Ical_t\Big(\bigcup_i \Eo_i\Big)[\rho]=\sum_i \Ical_t(\Eo_i)[\rho].
\]
\end{enumerate}
These properties follow easily from the properties of the map $\Sigma(t)$ and of the indicator function $\ind_\Eo$. Given the instrument and the pre-measurement state $\rho\in \Sscr(\Hscr)$, the probability of the result $\Eo$ is $\Tr\{\Ical_t(\Eo)[\rho]\}$ and the post-measurement state, knowing that the result is in $\Eo$, is given by the mapping $\Ical_t(\Eo)$ and normalization: 
\begin{equation}\label{Econd}
\rho(t|\Eo):=\frac{\Ical_t(\Eo)[\rho]}{\Tr\{\Ical_t(\Eo)[\rho]\}}, \qquad \Eo\in\Escr_t.
\end{equation} 
When the result of the observation is not known or not taken into account the post-measurement state reduces to $\Ical_t(\Omega)[\rho]$. Finally, the operator valued function $\Eo \mapsto \Ical_t(\Eo)^*[\openone]\in \Bscr(\Hscr)$ is the POVM associated with the instrument $\Ical_t$ and represents a generalized quantum observable.

Note that the probabilities given by this family of instruments coincide with the physical probability introduced in Sec.\ \ref{sec:nprob}.  Now the initial state at time zero is $\rho_0$ fixed in \eqref{linearSME}. Then, it is easy to see that the physical probability, defined by \eqref{def:pt} and \eqref{physprob}, and the mean state $\eta(t)$ \eqref{meanst} can be expressed as 
\[
\Pbb(\Eo)=\Tr\big\{ \Ical_t(\Eo)^*[\openone]\rho_0\big\}, \quad \forall \Eo \in \Escr_t,  \qquad \eta(t)=\Ical_t(\Omega)[\rho_0].
\]

When the set $\Eo$ in \eqref{Econd} practically shrinks to a single trajectory of the relevant observables, from the definition of the instruments \eqref{Icalinstr} and the related construction of post-measurement states, we get  the $\Escr$-conditional states
\[
\rho(t|\Escr_t)=\frac{\Ebb_\Qbb \big[\Sigma(t)[\rho_0]\big| \Escr_t\big]}{\Tr\big\{\Ebb_\Qbb \big[\Sigma(t)[\rho_0]\big| \Escr_t\big]\big\}}.
\]
Apart from simple cases, it is not possible to find an explicit expression for the conditional mean defined above and a closed evolution equation for the $\Escr$-conditional states. We can say that we are in a case of ``partial observation". When $\Escr_t=\Fscr_t$, $\forall t\geq 0$, the states $\rho(t|\Escr_t)$ coincides with the conditional states introduced in the definition \eqref{defhatsigma}.

\subsection{Non-uniqueness of the construction}
The physical law of the point process $\Pi(\rmd u,\rmd t)$ is determined by its compensator $I(u,t)\nu_t(\rmd u)\rmd t$ (Remark \ref{rem:Gir}); only the product counts. We allowed the measure $\nu_t$ and the operators involved in the definition of $I(u,t)$ to be random in order to include some extra randomness, but this randomness could be shifted from one ingredient to the other without changing the counting process. So, we can have an equivalent construction by making the replacements
\[
\Pi(\rmd u,\rmd t) \to \Pi'(\rmd u,\rmd t), \qquad \nu_t(\rmd u)\to \nu'_t(\rmd u), \qquad \Jcal(u,t)\to \Jcal'(u,t),
\]
with the constraint 
\begin{equation}\label{invariance}
\Jcal(u,t)\nu_t(\rmd u)=\Jcal'(u,t)\nu'_t(\rmd u);
\end{equation}
no other element is changed, so that $\Lcal(t)$ remains unchanged. Then, under the physical probability, the stochastic intensities \eqref{hatPi} and \eqref{lambda} do not change and $\Pi$ and $\Pi'$ have the same law. Moreover,  no change appears in the non-linear SME \eqref{nonlinearSME}.

As only the product \eqref{invariance} is relevant, we can use this freedom to simplify the choice of the $\Qbb$-probability law of the point measure $\Pi(\rmd u, \rmd t)$ and we take the following assumption, which turns out to be only a mild restriction.
\begin{assumption}\label{ass:Poiss}
We take the measure $\nu_t(\rmd u)$ to be deterministic and independent of time: $\nu_t(A) \to \nu(A)$, \qquad $\forall A\in \Bscr(\Uscr)$. We take also $\nu(\Uscr)>0$.
\end{assumption}
All the other properties contained in Assumption \ref{ass:U} continue to hold.  Under this Assumption, the point measure $\Pi(\rmd u,\rmd t)$ turns out to be, under $\Qbb$,  a time homogeneous, Poisson point process; in particular, the process $N(t)$ is a $\Qbb$-Poisson process of intensity $\nu(\Uscr)$. If we want to change the distribution of the counting process, the needed random contribution are  included into the definition of the jump operators $J(u,t)$ as done in Sec.\ \ref{sec:inters}.

\section{Typical trajectories}\label{sec:ttrajs}
The times of jumps are determined by the process $N(t)$ \eqref{def:N}, the sum of all the counting processes represented by $\Pi$ (Assumption \ref{ass:U}). Under $\Qbb$, $N(t)$ is a Poisson process (Assumption \ref{ass:Poiss}), while, under $\Pbb$, it is a counting process of stochastic intensity $\lambda(t)$ \eqref{lambda}. Heuristically, the expressions
\begin{equation}\label{jumpprob}
\Qbb[\rmd N(t)=1|\Fscr_{t_-}]=\nu(\Uscr)\rmd t, \qquad \Pbb[\rmd N(t)=1|\Fscr_{t_-}]=\lambda(t)\rmd t
\end{equation}
represent the conditional probabilities of a jump in an infinitesimal time interval, given the whole past. Under $\Pbb$,
when $\lambda(t)=0$, no jump is possible at time $t$.

Then, the type of jump is determined by the full point measure $\Pi$. If there is a jump at time $t$, the jump is of type $u\in A$, $A\subset \Uscr$, with probabilities
\begin{subequations}\label{probchoices}
\begin{equation} \Qbb\big[\Delta \Pi(A;t)=1\big|\Delta N(t)=1,\Fscr_{t_-}\big]=\frac{\nu(A)}{\nu(\Uscr)}, 
\end{equation}
\begin{equation}\label{condprob:uinA} \Pbb\big[\Delta \Pi(A;t)=1\big|\Delta N(t)=1,\Fscr_{t_-}\big]=  \int_A \frac{I(u,t)}{\lambda(t)}\,\nu(\rmd u),
\end{equation}
where
\begin{equation}\label{DeltaPi}
\Delta N(t):=N(t_+)-N(t_-), \qquad \Delta \Pi(A,t):=\Pi(A,[t_-,t_+]). 
\end{equation}
\end{subequations}
So, under $\Pbb$, the conditional probabilities of the type of jump depend on the ratio $\frac{I(u,t)}{\lambda(t)}$;  if $I(u,t)=0$, but $\lambda(t)>0$, we can have a jump, but not of type $u$.

Due to the assumptions  of a finite measure $\nu(\Uscr)$ \eqref{nufin} and of the existence of a finite operator $R(t)$ \eqref{def:R(t)}, the probabilities of next jump \eqref{jumpprob} have a smooth behaviour in time. Indeed, these assumptions eliminate the possibility of having ``infinitely many small jumps", a typical possibility in the theory of Lévy processes \cite{Prott04}, which has its counterpart in the theory of measurements in continuous time \cite{Hol86,BarHL93}. Then,
a typical trajectory in $[0,T)$ of the point process $\Pi$ is 
\begin{equation}\label{trajec}
\tau(t;m)\equiv \{(u_j,t_j),\; j=1,2,\ldots,m, \; u_j\in \Uscr,\; 0<t_1<\cdots<t_m< t\}, \quad m\in\Nbb;
\end{equation}
the case $m=0$ means that there is no jump in $[0,t)$. 

\subsection{The solution of the SME}\label{sec:solvSME}
The non-linear SME \eqref{nonlinearSME} says that at a jump $(u_j,t_j)$ the quantum state is transformed according to
\begin{equation}\label{sigmato}
\rho(t_{j-}) \mapsto \rho (t_j)=I(u_j,t_j)^{-1}\Jcal(u_j,t_j)[ \rho(t_{j-}) ], \qquad I(u_j,t_j)=\Tr\left\{ \Jcal(u_j,t_j)[ \rho(t_{j-}) ]\right\},
\end{equation}
where $\Jcal(u,t)$ was introduced in \eqref{Jcalu} and $I(u,t)$ in \eqref{def:pt}. Instead, in between two jumps the quantum state evolves according to the (non linear, CP preserving, trace preserving) equation
\begin{equation}\label{nojev}
\frac{\rmd \rho (t)}{\rmd t}= \Lcal_0(t)\big[\rho(t)\big] -\frac 12 \left\{R(t)-\lambda(t)\openone,\,\rho (t)\right\}.
\end{equation}
In \eqref{nojev}, the presence of $\lambda(t)$ \eqref{lambda} gives the trace preservation property.

In an analogous way we have that the linear SME \eqref{linearSME++} gives
\begin{equation}
\sigma(t_{j-}) \mapsto \sigma (t_j)=\Jcal(u_j,t_j)[ \sigma(t_{j-}) ]
\end{equation}
at a jump, and
\begin{equation}\label{sigmadot}
\frac{\rmd \sigma (t)}{\rmd t}= \Lcal_0(t)\big[\sigma(t)\big] -\frac 12 \left\{R(t)-\nu(\Uscr)\openone,\,\sigma(t)\right\}
\end{equation}
in between two jumps. By \eqref{def:pt} and \eqref{defhatsigma} the quantities $\rho(t)$ and $\sigma(t)$ differ only by the normalization. 
 
The operators $\Lcal_0(t)$ and $R(t)$ are  defined in \eqref{Lcalt} and \eqref{def:R(t)}; they can depend on the trajectory \eqref{trajec} up to time $t$ and eventually on some extra noise. Equations \eqref{nojev} and \eqref{sigmadot} hold in between two jumps; so, $\Lcal_0(t)$ and $R(t)$ must to be computed under the condition of no jump in the considered time interval. To construct the solution of \eqref{Lcalt} and \eqref{def:R(t)} we introduce an explicit notation.

\begin{definition}[{No jump propagator}]\label{def:Stt0} Let us denote by $\Lcal_0(t|t_0)$ and $R(t|t_0)$ the operators $\Lcal_0(t)$ and $R(t)$ given under the condition of no jump in $(t_0,t)$: $N(t)-N(t_0)=0$.
Then, we define the (random) \emph{no jump propagator} $\Scal(t,t_0)$ by
\begin{subequations}\label{def:propag}
\begin{equation}\label{eq:propag}
\frac {\rmd \ }{\rmd t}\,\Scal(t,t_0)[\rho]=\Acal(t|t_0)\circ\Scal(t,t_0)[\rho], 
\end{equation}
\begin{equation}\label{def:Acal} 
\Acal(t|t_0)[\rho]=\Lcal_0(t|t_0)[\rho] -\frac 12 \left\{R(t|t_0), \rho\right\}, \qquad \Scal(t_0,t_0)=\id.
\end{equation}
\end{subequations}
\end{definition}

When there is no jump in $(t_0,t)$, from Eqs.\ \eqref{nojev}, \eqref{sigmadot} we get
\begin{equation*}
\rho (t_-)= \exp\left\{\int_{t_0}^t \lambda(s)\rmd s\right\} \Scal(t,t_0)\big[\rho (t_0)\big],  \qquad\sigma(t_-)= \exp\left\{ \nu(\Uscr)\left(t-t_0\right)\right\} \Scal(t,t_0)\big[ \sigma(t_0)\big].
\end{equation*}
By these results, the solution of the linear SME  \eqref{linearSME}, in the typical trajectory $\tau(t,m)$, can be written in the following form:
\begin{subequations}\label{decompsigma}
\begin{equation}\label{sigmatau0}
\sigma(t_-)= \exp\left\{ \nu(\Uscr)t\right\} \Scal(t,0)[ \rho_0], \qquad \text{for} \ m=0,
\end{equation}
\begin{multline}\label{sigma|tau}
\sigma(t_-)= \exp\left\{ \nu(\Uscr)t\right\} \Scal(t,t_m)\circ \Jcal(u_m,t_m)\circ \Scal(t_m,t_{m-1})\circ
\\{} \circ\cdots \circ \Jcal(u_2,t_2)\circ \Scal(t_2,t_{1})\circ \Jcal(u_1,t_1)\circ \Scal(t_1,0)[ \rho_0], \qquad \text{for} \ m\geq 1.
\end{multline}
\end{subequations}
Analogously, the solution of the non-linear SME \eqref{nonlinearSME} takes the expression
\begin{subequations}\label{decomprho}
\begin{equation}
\rho (t_-)= \exp\left\{\int_{0}^t \lambda(s)\rmd s\right\} \Scal(t,0)\big[ \rho_0\big]=\frac{\Scal(t,0)\big[ \rho_0\big]}{\Tr\left\{\Scal(t,0)\big[ \rho_0\big]\right\}}, \qquad  \text{for} \ m=0,
\end{equation}
\begin{multline}
\rho (t_-)= \exp\left\{\int_{0}^t \lambda(s)\rmd s\right\} \Scal(t,t_m)\circ \frac{\Jcal(u_m,t_m)}{I(u_m,t_m)}\circ \Scal(t_m,t_{m-1})\circ
\\{} \circ\cdots \circ \frac{\Jcal(u_2,t_2)}{I(u_2,t_2)}\circ \Scal(t_2,t_{1})\circ \frac{\Jcal(u_1,t_1)}{I(u_1,t_1)}\circ \Scal(t_1,0)[ \rho_0], \qquad \text{for} \ m\geq 1.
\end{multline}
\end{subequations}

Equations \eqref{decomprho} give the structure of the statistical operator $\rho(t)$, the state of the system given  the whole past. In spite of the simple recursive aspect,  quantum jump trajectories can show very complex behaviours when applied to problems in quantum statistical mechanics \cite{P+23,BP+21,SBDKC25,Fazio26}.

\subsection{No jump probability and waiting times}\label{sec:nojprob}
Let us consider now the probabilities of no jump in a given time interval. Under the reference probability, $N$ is a Poisson process and the no-jump probability  has the simple expression
\begin{equation}\label{Q0}
\Qbb[N(t)-N(t_0)=0|\Fscr_{t_0}]=\Qbb[N(t)-N(t_0)=0]=\exp\left\{-\nu(\Uscr)\left(t-t_0\right)\right\}.
\end{equation}
Under the physical probability $\Pbb$, these probabilities become
\begin{subequations}\label{nojeqs}
\begin{equation}\label{N0}
\Pbb[N(t)=0]=\Ebb_\Qbb\big[\Tr\{\Scal(t,0)[\rho_0]\}\big],
\end{equation}
\begin{equation}\label{N0|}
\Pbb[N(t)-N(t_0)=0|\Fscr_{t_0}]=  \Ebb_\Qbb\big[\Tr\left\{\Scal(t,t_0)[ \rho(t_0)]\right\}\big|\Fscr_{t_0}\}\big].
\end{equation}
\end{subequations}

\begin{proof}[{Computations.}]
Being \eqref{nojev} trace preserving, in between two jumps,  by \eqref{def:pt},  \eqref{def:propag}, \eqref{sigmatau0} we get
\begin{equation}\label{lambdaScal} 
\Tr\left\{\Scal(t,t_0)\big[ \rho(t_0)\big]\right\}=\exp\left\{-\int_{t_0}^t \lambda(s)\rmd s\right\}, \qquad p(t_-)= \exp\left\{\int_{0}^t\left( \nu(\Uscr)-\lambda(s)\right)\rmd s\right\}.
\end{equation}
Then, the no jump probability is given by
\begin{multline*}
\Pbb[N(t)=0]= \Ebb_\Qbb\big[p(t)\ind_{\{N(t)=0\}}\big]=  \Ebb_\Qbb\Big[\exp\Big\{\int_{0}^t\left( \nu(\Uscr)-\lambda(s)\right)\rmd s\Big\} \ind_{\{N(t)=0\}}\Big]
\\ {}=\frac 1 {\Qbb[N(t)=0]}\int_{\{N(t)=0\}}\Qbb(\rmd \omega)\,\Tr\{\Scal(t,0)[\rho_0]\}
=\Ebb_\Qbb\big[\Tr\{\Scal(t,0)[\rho_0]\}\big].
\end{multline*}
We introduce now a notation for the event ``no jump in $(t_0,t)$": $\Eo(t,t_0):=\{N(t)-N(t_0)=0\}$. Similarly to the previous computations we get:  $\forall \Fo\in \Fscr_{t_0}$,
\[
\frac{\Pbb[\Eo(t,t_0)\cap \Fo]}{\Pbb[\Fo]}=\frac{\Ebb_\Qbb[p(t)\ind_\Fo \ind_{\Eo(t,t_0)}]}{\Pbb[\Fo]}= \frac{\Ebb_\Qbb[p(t_0)\ind_\Fo\ind_{\Eo(t,t_0)}\Tr\left\{\Scal(t,t_0)[ \rho(t_0)]\right\}]} {\Qbb[\Eo(t,t_0)]\Ebb_\Qbb[\ind_\Fo p_{t_0}]}.
\]
When $\Fo$ shrinks to a single trajectory in $(0,t_0)$ we get the final result in \eqref{N0|}.
\end{proof}

Let us take $t=0$ as initial time or as time of a jump; we are interested in the waiting time $T$ of the next jump. By its meaning we have the identification
\begin{subequations}\label{Twtj}
\begin{equation}\label{T=N(t)}
\Pbb[T\leq t]=1-\Pbb[N(t)=0].
\end{equation}
We can obtain the probability density by derivation, using \eqref{nojeqs}, \eqref{lambdaScal}:
\begin{equation}
p_T(t)=\frac{\rmd \Pbb[T\leq t]}{\rmd t}=-\frac{\rmd \Pbb[N(t)=0]}{\rmd t}=\Ebb_\Qbb\big[\lambda(t)\Tr\{\Scal(t,0)[\rho_0]\}\big].
\end{equation}
\end{subequations}
By \eqref{lambda}, the positivity of the operator $R$ gives the positivity of $\lambda(t)$ and this gives the positivity of the probability density of the waiting time. It must be remarked that it is possible to have  $\lim_{t\to +\infty}\Pbb[N(t)=0]> 0$. In this case $\Pbb[T=+\infty]>0$ and the distribution has also a discrete point at infinity; some examples are given in Secs.\ \ref{sec:nonEP} and \ref{sec:Hdepend}. 

The identification \eqref{T=N(t)} opens the possibility of a systematic approach to the notion of waiting times for quantum systems, working in continuous time \cite{HBZT25,RGMor25,LMMR25,DBernardinD21}. Here we have considered only the waiting times related to the process $N(t)$, but we could consider also the ones related to the point measure $\Pi(\rmd u, \rmd t)$; in this case the events for which one can introduce the waiting time are labelled by elements (or subsets) of the space $\Uscr$.

\subsection{Exclusive probability densities} \label{sec:exclpd}
All the probabilities for the counts can be reduced to the \emph{exclusive probability densities}, quantities which were at the basis of the first consistent formulation of the quantum theory of measurements in continuous time \cite{Dav69,Dav76,SriD81,SriD82} (see also \cite{BarB91,Car93,ZolG97,Bar06,Nori26} and references therein). 

The exclusive probability densities give the probabilities of a count of type $u_1$ around $t_1$, a count of type $u_2$ around $t_2$, \ldots,  a count of type $u_m$ around $t_m$, and no other count in between two counts and up to $t$; they are densities with respect to the measure  \ \ $\nu_{t_1}(\rmd u_1)\cdots \nu_{t_m}(\rmd u_m)$, \ \ $\rmd t_1\cdots \rmd t_m$ \ \ in the set $0<t_1<\cdots <t_m <t$ with $m=1,\ldots$. For the reference probabilities, for which the distribution of the waiting time is exponential, these densities are simply $\exp\{-\nu(\Uscr)t\}$. For the physical probability, they are obtained from the definition of $p(t;\omega)$ \eqref{def:pt} and \eqref{decompsigma}:
\begin{multline}\label{genexclpd}
p_t\big(u_m,t_m;\ldots;u_1,t_1\big)=
\Ebb_{\Qbb}\Big[\Tr \Big\{\Scal(t,t_m)\circ \Jcal(u_m,t_m)\circ \Scal(t_m,t_{m-1})\circ  \\ {}\circ\cdots \circ \Jcal(u_2,t_2)\circ \Scal(t_2,t_{1})\circ \Jcal(u_1,t_1)\circ \Scal(t_1,0)[ \rho_0]\Big\}\Big|\tau(m,t)\Big];
\end{multline}
the trajectory $\tau(m,t)$ \eqref{trajec} must be fixed. 

To compare these expressions with previous formulations we simplify the situation by eliminating any extra randomness. Different situations will be analyzed in Secs.\ \ref{sec:inters} and \ref{sec:walks}.

\begin{remark}\label{1/2Ass} In the remaining of this subsection and in Sec.\ \ref{sec:nonH}, all the operators $ H(t),  \; J_j(u,t), \; L_{k}(t)$ (introduced in Assumption \ref{ass:HJL}) are taken to be non random. Then, also $R(t)$, $R_0(t)$ \eqref{def:R(t)}, $\Jcal(u,t)$, $\Gcal(t)$, $\Lcal_0(t)$, $\Lcal(t)$ \eqref{linearSME+}
are non random. In particular, the condition $N(t)-N(t_0)=0$ cannot have influence on the structure of the various operators and in
Definition \ref{def:Stt0} we have  $\Lcal_0(t|t_0)=\Lcal(t)$, $R(t|t_0)=R(t)$. Now $\Acal(t|t_0) \mapsto \Acal(t)$ and $\Scal(t,t_0)$ \eqref{def:propag} is not dependent on the condition $N(t)-N(t_0)=0$. 
Moreover, by \eqref{Lcalx}  and \eqref{def:Acal}, we can write
\begin{subequations}\label{LGeta}
\begin{equation}\label{A+G(t)}
\Lcal(t)=\Acal(t)+ \Gcal(t).
\end{equation}
Due to the absence of randomness, the $\Qbb$-mean in  \eqref{masteq} has no effect; then, the mean state $\eta(t)$ \eqref{meanst} satisfies the quantum master equation
\begin{equation}
\frac{\rmd \eta(t)}{\rmd t}= \Lcal(t)[\eta(t)].
\end{equation}
\end{subequations}
\end{remark}

Now, due to the absence of extra noises, the no jump probability in $(t_0,t)$ \eqref{nojeqs} becomes
\begin{equation}\label{nojumpprob}
\Pbb[N(t)-N(t_0)=0|\Fscr_{t_0}]=  \Tr\left\{\Scal(t,t_0)[ \rho(t_0)]\right\}, \qquad \Pbb[N(t)-N(t_0)=0]=  \Tr\left\{\Scal(t,t_0)[ \eta(t_0)]\right\}.
\end{equation}
Moreover, in the exclusive probability densities \eqref{genexclpd} the $\Qbb$-mean can be eliminated and we get
\begin{multline}\label{expd-nonoise}
p_t\big(u_m,t_m;\ldots;u_1,t_1\big)
=\Tr\Big\{ \Scal(t,t_m) \circ \Jcal(u_m,t_m)\circ\Scal(t_m,t_{m-1})\circ \\ {}\circ \cdots \circ\Jcal(u_2,t_2)\circ\Scal(t_2,t_1)\circ\Jcal(u_1,t_1)\circ\Scal(t_1,0)[\rho_0]\Big\}.
\end{multline}
With respect to previous formulations \cite{Dav69,Dav76,SriD81,SriD82,Bar90,BarB91,Car93,Bar06,Potts+24} the space $\Uscr$ giving the types of jumps is very general, not only discrete. 

All the probabilities of interest can be computed by starting from the probability densities and the no jump probability. For instance, the probability density for a given sequence of $m$ jumps and no other jump in $(0,t)$ is given by
\begin{multline}\label{expd-uu}
p_t\big(u_m;\ldots;u_1\big)
= \int_{0}^t\rmd t_{m}\int_{0}^{t_m}\rmd t_{m-1}\cdots\int_{0}^{t_2}\rmd t_1\Tr\Big\{ \Scal(t,t_m) \circ \Jcal(u_m,t_m)\circ\Scal(t_m,t_{m-1})\circ \\ {}\circ \cdots \circ\Jcal(u_2,t_2)\circ\Scal(t_2,t_1)\circ\Jcal(u_1,t_1)\circ\Scal(t_1,0)[\rho_0]\Big\}.
\end{multline}
On the other side, the probability densities for the times of the jumps, without considering the types of the jumps, are given by
\begin{multline}\label{timeexpd}
p_t\big(t_m;\ldots;t_1\big)=\int_\Uscr \nu(\rmd u_m) \cdots \int_\Uscr \nu(\rmd u_1)p_t\big(u_m,t_m;\ldots;u_1,t_1\big)
\\ {}= \Tr\Big\{ \Scal(t,t_m) \circ \Gcal(t_m)\circ\Scal(t_m,t_{m-1})\circ  \cdots \circ\Gcal(t_2)\circ\Scal(t_2,t_1)\circ\Gcal(t_1)\circ\Scal(t_1,0)[\rho_0]\Big\};
\end{multline}
$\Gcal(t)$ is given in \eqref{Jcalu}. By integrating \eqref{timeexpd} on the possible times $0< t_1 < t_2<\cdots <t_m<t$ (or by integrating \eqref{expd-uu} on $u_1,\ldots,u_m$), we get the probability of having exactly $m$ jumps in $(0,t)$:
\begin{subequations}\label{DcalN}
\begin{equation}\label{Nm}
\Pbb[N(t)=m]=\int_{0}^t\rmd t_{m}\int_{0}^{t_m}\rmd t_{m-1}\cdots\int_{0}^{t_2}\rmd t_1 p_t\big(t_m;\ldots;t_1\big)=\Tr\{\Dcal_m(t)[\rho_0]\},
\end{equation}
\begin{multline}\label{def:Dm}
\Dcal_m(t):=\int_{0}^t\rmd t_{m}\int_{0}^{t_{m}}\rmd t_{m-1}\cdots\int_{0}^{t_3}\rmd t_{2} \int_{0}^{t_2}\rmd t_1\,\Scal(t,t_m) \circ\Gcal(t_m)\circ\Scal(t_m,t_{m-1})
\\ {} \circ \cdots \circ\Gcal(t_2)\circ\Scal(t_2,t_1)\circ\Gcal(t_1)\circ\Scal(t_1,0), \qquad m\geq 1.
\end{multline}
Similarly, the probability of no count \eqref{nojumpprob} can be written as
\begin{equation}\label{N0D}
\Pbb[N(t)=0]=\Tr\{\Dcal_0(t)[\rho_0]\}, \qquad \Dcal_0(t):= \Scal(t,0).
\end{equation}
\end{subequations}
Now we can check that the exclusive probability densities \eqref{expd-nonoise} are correctly normalized by showing that by summing up all the probabilities \eqref{Nm}, \eqref{N0D} we get one. Indeed, this can be done by taking into account the Dyson expansion of the propagator of the mean dynamics \cite{Vacc14}. 

\begin{remark}\label{rem:Dmeta} By using the notation of time-ordered product, the expansion of the propagator related to the dynamics of $\eta(t)$ \eqref{LGeta} in terms of $\Gcal(t)$ is given by
\begin{equation}\label{vareta}
\eta(t)=\overset{\leftarrow}{\exp}\Bigl\{\int_0^t\Lcal(s)\rmd s\Bigr\}[\rho_0], \qquad \overset{\leftarrow}{\exp}\Bigl\{\int_0^t\Lcal(s)\rmd s\Bigr\}=\sum_{m=0}^{+\infty} \Dcal_m(t), 
\end{equation}
where the operators $\Dcal_m(t)$ are exactly the ones defined in \eqref{def:Dm}, \eqref{N0D}. 
\end{remark}

The dynamics generated by \eqref{A+G(t)} is trace preserving; so, we have
\[
\sum_{m=0}^{+\infty} \Pbb[N(t)=m]=\sum_{m=0}^{+\infty} \Tr\{\Dcal_m(t)[\rho_0]\}=\Tr\{\eta(t)\}=\Tr\{\rho_0\}=1.
\]

In principle, all the probabilities concerning the types and the times of the jumps can be obtained from the exclusive probability densities. Analogously, one can obtain any kind of moments and correlations related to these probabilities \cite{BarG12,Potts+24}. We can say that the exclusive probability densities are at the basis of a ``full counting statistics" \cite{Potts+24,Potts+25,Nori26}.

\section{Evolutions under non-Hermitian Hamiltonians}\label{sec:nonH}  
The study of evolutions under non-Hermitian (n-H) Hamiltonians is an active field of theoretical and experimental research, see \cite{SerZ13,Gao+15,SCr15,MMCN19,NAJM19,AGU20,NonS25,NonVs25,Ho+26,DMD26,Longhi25,DBernardinD21} and references therein. In the jump trajectory theory, the n-H Hamiltonians naturally appear in between jumps.

Here we consider the case no extra randomness, as in Remark \ref{1/2Ass}, and no time dependence in the operators giving the dynamics; moreover, we take $\Lcal_0$ of purely Hamiltonian type, i.e. in \eqref{Lcalt} we take $L_k(t)=0$ and $R_0(t)=0$.
Then, equations  \eqref{def:propag} for the no-jump propagator reduce to 
\begin{equation}\label{Aeff}
\frac {\rmd \ }{\rmd t}\,\Scal(t,t_0)[\rho]=\Acal\circ\Scal(t,t_0)[\rho], \qquad 
\qquad \Acal[\rho]= -\rmi\left(H-\frac\rmi 2\, R\right)  \rho + \rmi \,\rho \left(H+\frac\rmi 2\, R\right).
\end{equation}
We can say that the evolution between two jumps is given by the n-H Hamiltonian $H-\frac\rmi 2\, R$, with $R\geq 0$. 

Viceversa, when we want to include any given effective n-H Hamiltonian $H_{\rm eff}$, we can construct $H$ and $R$ by
\begin{equation}\label{Heff}
H:=\frac 12 \left(H_{\rm eff}+H_{\rm eff}^\dagger\right), \qquad R:= c \openone+\rmi \left(H_{\rm eff}-H_{\rm eff}^\dagger\right),\qquad  \text{with} \quad  c\in\Rbb: \ R\geq 0.
\end{equation}
The choice of $c$ is always possible because we are considering only bounded operators (Assumption \ref{ass:HJL}); more in general, we could ask ${}\:\rmi H_{\rm eff}-\rmi H_{\rm eff}^\dagger$\ to be bounded from below. The positivity of $R$ is essential for the consistency of the theory. In particular, the condition $R\geq 0$  guaranties the decreasing of the survival probability $\Pbb[N(t)=0]$ and the positivity of the probability density \eqref{Twtj} of the waiting time $T$ of the next jump; this point was recognized also in the literature on n-H evolutions \cite[Sec.\ III.A]{SCr15}.

In this section we are interested only to the evolution in between two jumps; so, we can simplify by introducing a single counting process $N(t)$, which means that the type of jump is not observed or that there is a single type of jump ($\Uscr$ reduces to one point). Under the reference probability, $N(t)$ is a Poisson process of intensity $\nu$; moreover, we denote by $\Jcal$ the CP jump operator. Then, the linear SME \eqref{linearSME++} reduces to
\begin{equation}\label{SMEeff}
\rmd  \sigma(t)=\Bigl( \Acal_0[ \sigma(t_-) ]  -(c-\nu) \sigma(t_-)\Bigr)\rmd t
+  \Bigl( \Jcal[ \sigma(t_-) ]-\sigma(t_-) \Bigr) \rmd N(t), 
\end{equation}
\begin{equation}\label{A0}
\Acal_0[\rho]=-\rmi H_{\rm eff} \rho  +\rmi  \rho H_{\rm eff}^\dagger \,.
\end{equation}
Moreover, equations \eqref{Jcalu}, \eqref{GtoR}, \eqref{Aeff}, \eqref{Heff}, \eqref{LGeta} give
\begin{equation}\label{etaeff}
\Gcal=\nu \Jcal, \qquad \Gcal^*[\openone] =R,
\qquad
\Acal[\rho]= \Acal_0[\rho] -c \rho, \qquad \Lcal = \Acal+\Gcal, \qquad \frac{\rmd \eta(t)}{\rmd t}= \Lcal[\eta(t)].
\end{equation}
A possible choice for the jump operator is to take $\Gcal[\rho]=R^{1/2}\rho R^{1/2}$. Once $H_{\rm eff}$ has been fixed, there are infinitely many choices for $c$ and $G$. These choices change the jump operator in   \eqref{SMEeff} and, so, they are physically inequivalent, as they change what happens at a jump. However, these choices do not change the dynamics in between two jumps. Being $\Gcal$ CP, the mean state $\eta(t)$ satisfies an usual quantum master equation. 

\subsection{The non-linear SME and the survival probability}\label{sec:survpr}
Under the physical probability $\Pbb$, $N(t)$ is a counting process of stochastic intensity $\lambda(t)$ \eqref{lambda} and the a posteriori states \eqref{defhatsigma} satisfy the non-linear SME \eqref{nonlinearSME}, which becomes
\begin{equation}\label{nlSMEeff}
\rmd \rho (t)= \Bigl( \Acal_0[ \rho(t_-)]    + \tilde\lambda(t)\rho (t_-) \Bigr)\rmd t
+ \Bigl( I(t)^{-1}\Jcal[ \rho(t_-) ]-\rho(t_-) \Bigr) \rmd N(t),
\end{equation}
with 
\begin{equation}\label{tildelam}
\tilde\lambda(t)=\rmi\Tr\left\{\left(H_{\rm eff}-H_{\rm eff}^\dagger\right)\rho(t_-)\right\}, \qquad I(t)=\Tr\left\{\Jcal\big[\rho(t_-)\big]\right\},\qquad \lambda(t)=I(t)\nu= \tilde \lambda(t)+c.
\end{equation}

By \eqref{Aeff}, \eqref{Heff}, \eqref{A0}, the no-jump propagator \eqref{def:propag} can be written as
\begin{equation}\label{S0+Heff}
\Scal(t,t_0)=\rme^{-c(t-t_0)}\Scal_0(t-t_0), \qquad \Scal_0(t)[\rho]:= \rme^{\Acal_0 t}[\rho]\equiv \rme^{-\rmi H_{\rm eff}t}\rho \rme^{\rmi H_{\rm eff}^\dagger t}.
\end{equation}
Then, the two SMEs \eqref{SMEeff} and \eqref{nlSMEeff} give that, in between two jumps, the evolution is given by 
\begin{equation}\label{betweenJ}
\sigma(t)=\rme^{(\nu-c)(t-t_0)} \Scal_0(t-t_0)[\sigma(t_0)], \qquad \rho (t)=\frac{\sigma(t)}{\Tr\{\sigma(t)\}}=\frac{\Scal_0(t-t_0)[\rho(t_0)]} {\Tr\{\Scal_0(t-t_0)[\rho(t_0)]\}}\,.
\end{equation}
Let us stress that the evolution of the conditional state does not depend on the choice of $c$ and that it is determined only by  $H_{\rm eff}$. Being the effective Hamiltonian non-Hermitian, the evolution of $\rho(t)$ includes the normalization factor. In the theory of non-Hermitian evolutions, such a normalization appeared, for instance, in \cite{SerZ13}. 

Instead, at a jump the state change is given by
\begin{equation}
\sigma(t_0)=\Jcal[\sigma({t_0}_-)],\qquad \rho (t_0)=\frac{\Jcal[\sigma({t_0}_-)]}{I(t_0)} =\frac{\Gcal[\rho ({t_0}_-)]}{\lambda(t_0)} =\frac{\Gcal[\rho ({t_0}_-)]}{c+ \rmi\Tr\left\{\left(H_{\rm eff}-H_{\rm eff}^\dagger\right)\rho ({t_0}_-)\right\}}\,; 
\end{equation}
recall that, under $\Pbb$, when $\lambda(t_0)=0$ there is no jump.

From \eqref{nojumpprob}, the survival probability is
\begin{equation}\label{survprob}
\Pbb[N(t)=0]=\rme^{-ct}\Tr\{\Scal_0(t)[\rho_0].
\end{equation}
Then, the distribution of the waiting time \eqref{Twtj} is given by
\begin{equation}\label{Twtjvar}
\Pbb[T\leq t]=1-\Pbb[N(t)=0],\qquad p_T(t)=\lambda(t)\Pbb[N(t)=0]=\bigl(c+\tilde\lambda(t)\bigr)\Pbb[N(t)=0];
\end{equation}
$\tilde \lambda (t)$ is defined in \eqref{tildelam} and $c$ is needed to guarantee the positivity of this probability density. As we shall see in the examples in Sec.\ \ref{sec:nonEP}, 
it is possible to have $\lim_{t\to + \infty}\Pbb[N(t)=0]>0$; in this case we have $\Pbb[T=+\infty]>0$, and the distribution of  $T$ is a mixture of a continuous one and a discrete one.

Summarizing, the whole behaviour of the system before the first jump (or in between two jumps) depends on the constant $c$ and on the effective Hamiltonian $H_{\rm eff}$, both introduced in \eqref{Heff}.  The no-jump propagator $\Scal_0(t)$ \eqref{S0+Heff}, generated by $H_{\rm eff}$, gives the dynamics of the conditional state (the second equation in \eqref{betweenJ}). Finally, $\Scal_0(t)$ and the constant $c$ give the survival probability \eqref{survprob} and the waiting time distribution \eqref{Twtjvar}.

\subsection{Examples in $\Cbb^2$}\label{sec:C2ex}
In the literature on n-H evolutions a relevant role is played by the \emph{exceptional points} (EPs). The evolution under an effective  n-H Hamiltonian can be studied by starting from its eigenvalues and right/left eigenvectors \cite{AGU20}; the EPs in the parameter space are singular points at which two eigenvectors and eigenvalues collapse. As stated in many articles, the physical behaviour of the system near an EP is well represented by the two-level case; so, we illustrate the theory by some examples involving $\Hscr=\Cbb^2$. We do not consider explicitly the jump operator; it could determine a jump inside the two-level space, but it could be a jump to some other levels; in this last case, at the jump the system disappears. In \cite{NAJM19,NonS25,NonVs25} it is explicitly stated that the n-H dynamics stays in a two dimensional subspace up to a jump to a third fundamental state. The experimental realizations of systems with n-H evolutions are complicated apparatuses involving excitons, polaritons, control of gain and losses\ldots \cite{Gao+15,NAJM19,Ho+26}. 

So, let us consider a n-H dynamics in $\Cbb^2$, before the first jump (or in between two jumps). The generic effective Hamiltonian can be parametrized as
\begin{subequations}\label{parametr}
\begin{equation}\label{C2Heff}
H_{\rm eff}=\begin{pmatrix}\epsilon_1 & \beta \\ \delta &\epsilon_0\end{pmatrix}=\frac{\epsilon_1 +\epsilon_2}2 \,\openone +K, \qquad   \epsilon_j,\; \beta,\; \delta \in \Cbb, 
\end{equation}
\begin{equation}\label{Keff}
K=\begin{pmatrix} \alpha & \beta \\ \delta & -\alpha\end{pmatrix}=\frac {\beta +\delta}2\,\sigma_x +\frac {\rmi(\beta-\delta)}2\, \sigma_y + \alpha \sigma_z,  \qquad \alpha:=\frac{\epsilon_1 -\epsilon_0}2;
\end{equation}
\end{subequations}
$\sigma_x$, $\sigma_y$, $\sigma_z$ are the usual Pauli matrices. Then, the $R$-operator and the positivity condition \eqref{Heff} become
\begin{subequations}\label{R+k+tilde}
\begin{equation}
R =\tilde c \,\openone +\rmi\left(K-K^\dagger\right)
= \tilde c\,\openone+\begin{pmatrix} -2\IM\alpha &\rmi\left(\beta-\overline \delta \right) \\ \rmi\left( \delta -\overline \beta\right)&2\IM\alpha\end{pmatrix} , \qquad \tilde c := c-\IM\left(\epsilon_1+\epsilon_0\right),
\end{equation}
\begin{equation} \label{pos2x2}
R\geq 0 \qquad \Leftrightarrow \qquad 
\tilde c\geq k:=\sqrt{\abs{\beta-\overline \delta}^2 + (2\IM \alpha)^2}\geq 0.
\end{equation}
\end{subequations}
A relevant point of the $\Cbb^2$-case is that the square of the $K$-operator \eqref{Keff} is proportional to the identity:
\begin{equation}\label{K2=}
 K^2=\varkappa\openone,\qquad \varkappa:=\alpha^2+\beta\delta.
\end{equation}
By using the quantity $\varkappa$, the bound $k$ in \eqref{pos2x2} can be rewritten as
\begin{equation}\label{kappa2}
 \qquad k^2=2\abs \alpha^2+\abs \beta ^2 +\abs\delta^2-2\RE\varkappa =2\left( \abs \alpha^2+\abs \beta \abs\delta-\RE\varkappa \right) +\left(\abs\beta-\abs\delta\right)^2.
\end{equation}
From the parametrization \eqref{Keff} of $K$ and the expressions of $k$ \eqref{pos2x2} and \eqref{kappa2}, we have 
\begin{equation}
K= K^\dagger \qquad \Leftrightarrow \qquad k=0 \qquad \Leftrightarrow \qquad \abs\beta=\abs\delta, \quad \varkappa=\abs\alpha^2+\abs\beta^2.
\end{equation}

Now, the no jump probability \eqref{survprob} takes the expression
\begin{equation}\label{C2survprob}
\Pbb[N(t)=0]= \rme^{-\tilde c\, t} \Tr \{ \zeta(t)\}, \qquad  \zeta(t):=\rme^{-\rmi K t}\rho_0\rme^{\rmi K^\dagger t};
\end{equation}
the associated probability density of the waiting time of the first jump is given by \eqref{Twtjvar}. Let us remark that the survival probability \eqref{C2survprob} has been obtained consistently with the usual formulation of a quantum theory based on POVMs and ``instruments", see Secs.\ \ref{sec:POVM} and \ref{sec:nojprob}. Then, the conditional state before the first jump \eqref{betweenJ} has the expression
\begin{equation}\label{C2rhot}
\rho (t)=\frac{\zeta(t)}{\Tr\{ \zeta(t)\}};
\end{equation}
note that this evolution is non linear in the initial state because of the normalization: it is the state attributed to the system knowing that no jump occurred up to time $t$. Also the name of post-selected state is used \cite{MMCN19,NonS25,NonVs25}.

According to the choices of the parameters in \eqref{C2Heff} there are different behaviours of the quantum system.
\begin{itemize}
\item $R=0$: equivalent to $\beta=\overline \delta$, \quad $\IM \epsilon_1=\IM \epsilon_0$, \quad $\tilde c=0$. In this case there are no jumps, because $\lambda(t)=0$, and the evolution is unitary, because $K=K^\dagger$.
\item $R\propto \openone$, \ $R\neq 0$: equivalent to \ $\beta=\overline \delta$, \quad $\IM \epsilon_1=\IM \epsilon_0$, \quad $\tilde c>0$. We have a unitary evolution with Hamiltonian $K=K^\dagger$ interspersed by quantum channels; the jumps are at a rate $\lambda=c-\IM(\epsilon_1+\epsilon_0)>0$. We shall treat this case in a more general set up in Sec.\ \ref{sec:inters}.
\item $K\neq K^\dagger$, \ $K^2=0$: equivalent to $\alpha^2+\beta\delta=0$ \ and \ $\abs \beta+ \abs\delta \neq 0$. This is the case of an exceptional point (EP), discussed in Sec.\ \ref{sec:EP}.

\item $K^2\neq 0$,  $K\neq K^\dagger$: equivalent to  $\alpha^2 +\beta\delta\neq 0$ and $k>0$. This is the typical case of a n-H evolution, when two left eigenvector (and two right eigenvectors) can be found, see Sec.\ \ref{sec:nonEP}.
\end{itemize}

In the literature there are various proposals of effective Hamiltonians in the reduced case of $\Hscr=\Cbb^2$, and in some cases also the related experiments are performed. By using the parametrization given in \eqref{parametr}, the proposals of \cite{Gao+15,SCr15,MMCN19,NAJM19,Ho+26,DMD26} can be unified by taking
\begin{subequations}\label{refex}
\begin{equation}
\epsilon_j=E_j-\rmi \Gamma_j, \quad E_j,\; \Gamma_j\in \Rbb, \qquad \delta=\overline \beta \,\rme^{\rmi \theta}, \quad \theta\in \Rbb, \qquad \alpha=\frac{E_1-E_2}2-\rmi\, \frac{\Gamma_1-\Gamma_2}2.
\end{equation}
The parameters $E_j$ represent the energy levels and the $\Gamma_j$ represent gain and losses. Then, in \eqref{R+k+tilde}, \eqref{K2=}, we have
\begin{equation}\begin{split}
\frac{\epsilon_1 +\epsilon_0}2= \frac 12 \left[ E_1+E_0-\rmi \left(\Gamma_1+\Gamma_0\right)\right], \qquad\qquad &\tilde c= c-\frac{\Gamma_1+\Gamma_0}2, 
\\ 
k=\sqrt{4\abs{\beta}^2\left(\sin\theta/2\right)^2+ (\Gamma_1-\Gamma_0)^2},
\qquad\qquad
&\varkappa=\alpha^2+\abs\beta^2\rme^{\rmi \theta}.
\end{split}\end{equation}
\end{subequations}

Here below we study the behaviour of a n-H evolution in $\Cbb^2$ with a particular attention to the distribution of the waiting time of the first jump, an argument not usually considered in the literature.

\subsubsection{The behaviour in an exceptional point}\label{sec:EP}
The EP point in the parameter space is when we have
\begin{equation}\label{EPcond}
\varkappa\equiv \alpha^2+\beta \delta=0, \qquad \abs \beta +\abs \delta \neq 0 \qquad \Leftrightarrow \qquad K\neq 0, \qquad K^2=0;
\end{equation}
the equivalence follows immediately from the expression of $K$ \eqref{Keff} and the property \eqref{K2=}. In this case we have not two distinct eigenvectors; indeed, we have 
\begin{equation}\label{K0form}
K=({\abs\beta+\abs\delta}){|\phi_0\rangle \langle \phi_1|},  \qquad \langle \phi_j|\phi_i\rangle=\delta_{ij},
\end{equation}
\begin{subequations}\label{phi01}
\begin{gather}\label{1phi01}
\text{for} \ \beta\neq 0 : \qquad \phi_0=\frac 1 {\sqrt{\abs \beta+ \abs \delta}}\begin{pmatrix} \sqrt{\abs \beta}\\ \rme^{\rmi \left(\arg \delta-\arg \alpha\right)}\sqrt{\abs \delta}\end{pmatrix}, \qquad \phi_1=\frac 1 {\sqrt{\abs \beta+ \abs \delta}}\begin{pmatrix}\rme^{-\rmi \arg \alpha} \sqrt{\abs \delta}\\ \rme^{-\rmi \arg\beta}\sqrt{\abs \beta}\end{pmatrix} ,
\\ \label{2phi01}
\text{for} \ \beta= 0 : \qquad \qquad \phi_0=\begin{pmatrix} 0 \\ \rme^{\rmi \arg \delta}\end{pmatrix}, \qquad \phi_1=\begin{pmatrix} 1 \\ 0\end{pmatrix} .
\end{gather}
\end{subequations}
\begin{proof}[{Computations}] From the vanishing of $\varkappa$ we get
\[
\abs \alpha^2=\abs \beta \abs \delta, \qquad \rme^{2\rmi \arg \alpha}=-\rme^{\rmi (\arg \beta +\arg \delta)}, \qquad \rme^{\rmi (\arg \alpha- \arg \delta)}=-\rme^{-\rmi (\arg \alpha -\arg \beta)}.
\]
When $\beta=0$, the quantities $\arg \alpha$ and $\arg\beta$ are not defined; similarly, $\arg \alpha$ and $\arg\delta$ are not defined when $\delta=0$.

Then, we search for two vectors such that \ $K\phi_0=0$ \ and \ $K\phi_1\propto \phi_0$; it is easy to verify that the vectors given in \eqref{phi01} satisfy this requirement, they are orthogonal and give the expression \eqref{K0form} for $K$. However, the expression $|\phi_0\rangle\langle \phi_1|$ is independent from a common phase of the two vectors, and one can check that it is not possible to choose this phase in a way independent from any value of the parameters. The choice \eqref{1phi01}  works well when all the parameters are not vanishing; it works also for $\delta=0$, because in this case the dependence on $\arg \delta$ and $\arg \alpha$ disappears and the two vectors become
\[
\text{for} \ \delta= 0 : \qquad \phi_0=\begin{pmatrix} 1\\ 0\end{pmatrix}, \qquad \phi_1= \begin{pmatrix}0\\ \rme^{-\rmi \arg\beta}\end{pmatrix} .
\]
In the case of $\beta=0$, the choice \eqref{1phi01} looses its meaning because the dependence on $\arg\alpha $ survives; in this case, one can check that the choice \eqref{2phi01} works well.
\end{proof}

Being $K$ a nilpotent operator \eqref{EPcond}, the associated evolution operator $\rme^{-\rmi Kt}$ turns out to be linear in $t$ and we have
\begin{equation}\label{EPevolution}
\rme^{-\rmi K t}=\openone -\rmi K t= \openone-\rmi t \left(\abs \beta + \abs \delta\right)|\phi_0\rangle \langle \phi_1|.
\end{equation}
Moreover, by \eqref{kappa2} and \eqref{EPcond}, the positivity condition \eqref{pos2x2} becomes
\begin{equation}\label{posEP}
R\geq 0 \qquad \Leftrightarrow \qquad \tilde c \equiv c-\IM\left(\epsilon_1+\epsilon_0\right)\geq k\equiv \abs \beta + \abs\delta>0. 
\end{equation}

Let us consider a generic initial state
\[
\rho_0= \sum_{i=1}^2 |\psi_i\rangle r_i\langle \psi_i|, \qquad \langle \psi_i|\psi_j\rangle =\delta_{ij}, \qquad 0\leq r_i\leq 1, \qquad r_2= 1-r_1 .
\]
By \eqref{C2survprob}, \eqref{EPevolution} we get
\[ 
\zeta(t)= \sum_{i=1}^2 |\psi_i(t)\rangle r_i\langle \psi_i(t)|, \qquad |\psi_i(t)\rangle =\left(\openone -\rmi K t\right)| \psi_i\rangle  =|\psi_i\rangle -\rmi t \left({\abs\beta+\abs\delta}\right)|\phi_0\rangle \langle \phi_1|\psi_i\rangle.
\]
Then, the conditional state up to the first jump \eqref{C2rhot}  takes the form 
\begin{subequations}
\begin{equation}
\rho(t)=\frac{\zeta(t)}{\sum_jr_j\norm{\psi_j(t)}^2},\qquad  \norm{\psi_i(t)}^2 =\abs{\langle \phi_1|\psi_i\rangle}^2+ \abs{\langle \phi_0|\psi_i\rangle -\rmi t \left(\abs \beta +\abs \delta\right)\langle\phi_1|\psi_i\rangle}^2,
\end{equation}
and the survival probability \eqref{C2survprob} becomes
\begin{equation}
\Pbb[N(t)=0]= \rme^{-\tilde c t} \sum_{i=1}^2 r_i\norm{ \psi_i(t)}^2.
\end{equation}
\end{subequations}
Finally, by direct computation of the derivative in \eqref{Twtj}, we get the probability density of the waiting time of the first jump:
\begin{subequations}\label{densK0}
\begin{equation}
p_T(t)=-\frac{\rmd \Pbb[N(t)=0]}{\rmd t}=\sum_{i=1}^2 r_ip_i(t),
\end{equation}
\begin{equation}\label{k0dens}
p_i(t)=\rme^{-\tilde c t} \Bigl\{ \left(\tilde c-\abs \beta -\abs \delta\right)\norm{ \psi_i(t)}^2 + \left(\abs \beta +\abs \delta\right)
\abs{\langle\phi_0|\psi_i\rangle+ \rmi \langle\phi_1|\psi_i\rangle\left[1 - t(\abs \beta+\abs\delta)\right]}^2  \Bigr\}.
\end{equation}
\end{subequations}

As it appears, the positivity condition \eqref{posEP} is necessary to guarantee the positivity of these probability densities.
To illustrate this point, we consider two simple examples in the case of a pure initial state $\rho_0=|\psi\rangle \langle \psi|$. By taking $\psi=\frac{1}{\sqrt 2}\left(\phi_0+\rmi\phi_1\right)$ we get
\[
p_T(t)= \rme^{-\tilde c t}\left\{\left(\tilde c -\abs \beta-\abs\delta\right)\left[1+\left(\abs \beta+\abs\delta\right)t\right]+\frac 12\, \tilde c \left(\abs \beta+\abs\delta\right)^2t^2\right\};
\]
the positivity condition \eqref{posEP} is needed to have $p_T(t)\geq 0$ for small times. Instead, for $\psi=\phi_1$ we get
\[
p_T(t)= \rme^{-\tilde c t}  \Big[ \tilde c\left(1+ t^2(\abs \beta +\abs \delta)^2\right) 
-2t (\abs \beta +\abs \delta)^2  \Big]
\]
and the positivity condition \eqref{posEP} is needed to guarantee $p_T(t)\geq 0$ at intermediate times, as one sees by taking $t=1/(\abs \beta +\abs\delta)$:
\[
p_T(t)\Big|_{t=1/(\abs \beta +\abs\delta)}=\exp\left\{\frac{-\tilde c}{\abs \beta +\abs\delta}\right\}2\left(\tilde c-\abs \beta -\abs\delta\right).
\]

As these examples show, there is a variety of behaviours of the probability densities \eqref{k0dens}, depending on the initial state. This variety of behaviours is apparent also by taking only the mean and  the variance of the survival time. Let us consider a pure initial state and the smallest possible value for the parameter $\tilde c$:
\[
\rho_0=|\psi \rangle \langle \psi |, \qquad \tilde c = \abs\beta +\abs \delta=:\lambda_0.
\]
Then, by direct computations from the density $p_T(t)$, we get some examples:
\begin{itemize}
\item for $\psi= \phi_0$: \qquad \qquad $\Ebb_\Pbb[T]=1/\lambda_0$, \quad $\Var_\Pbb[T]=1/\lambda_0^2$;
\item for $\psi=\frac 1 {\sqrt 2}\left(\phi_0- \rmi \phi_1\right)$: \quad $\Ebb_\Pbb[T]=1/\lambda_0$, \quad $\Var_\Pbb[T]=3/\lambda_0^2$;
\item for $\psi=\frac 1 {\sqrt 2}\left(\phi_0+ \rmi \phi_1\right)$: \quad $\Ebb_\Pbb[T]=3/\lambda_0$, \quad $\Var_\Pbb[T]=3/\lambda_0^2$;
\item for $\psi= \phi_1$: \qquad \qquad $\Ebb_\Pbb[T]=3/\lambda_0$, \quad $\Var_\Pbb[T]=5/\lambda_0^2$.
\end{itemize}

\subsubsection{Out of the exceptional point} \label{sec:nonEP}
Now, we consider the case of a n-H evolution out of the EP:
\begin{equation}\label{noEPcond}
\varkappa\equiv \alpha^2+\beta \delta\neq 0, \qquad k\equiv \sqrt{\abs{\beta-\overline \delta}^2 + (2\IM \alpha)^2} > 0 \qquad \Leftrightarrow \qquad K\neq K^\dagger, \qquad K^2\neq0;
\end{equation}
Also in this case it is possible to get the evolution operator generated by $K$ directly from the property \eqref{K2=}, without introducing its eigenvectors. The computations needed in proving the following statements are given in Appendix \ref{proof:4:1}.
\begin{remark}\label{rem:4:1}
The propagator $\rme^{-\rmi Kt}$ can be expressed as
\begin{equation}\label{evolgpm}
\rme^{-\rmi Kt}=g_+(\varkappa t^2)\openone - g_-(\varkappa t^2)\rmi Kt, \qquad g_+(\xi):= \sum_{n=0}^{+\infty}\frac{(-\xi)^n}{(2n)!}, \qquad g_-(\xi):= \sum_{n=0}^{+\infty}\frac{(-\xi)^n}{(2n+1)!}.
\end{equation}
The radius of convergence of the two series is $+\infty$; $\varkappa=0$ is not a singular point. For $\varkappa\to 0$ we have the right evolution in the EP.

It is possible to rewrite the propagator \eqref{evolgpm} in terms of eigenvalues and right eigenvectors of $K$; this procedure involves the choice the square root of $\varkappa$. 
By defining
\begin{equation}\label{defz}
\chi:= \arg \, \frac\varkappa{\abs{\varkappa}}, 
\qquad z:=\rme^{\rmi \chi/2} \sqrt{\abs\varkappa},
\end{equation}
the two series appearing in \eqref{evolgpm} can be expressed by means of $z$ as
\begin{equation}\label{gpmz}
g_+(\varkappa t^2)=\frac 12 \left(\rme^{\rmi zt}+\rme^{-\rmi zt}\right), \qquad g_-(\varkappa t^2)=\frac 1{2\rmi zt} \left(\rme^{\rmi zt}-\rme^{-\rmi zt}\right).
\end{equation}
\end{remark}

By inserting \eqref{gpmz} into \eqref{evolgpm}, we get  the expression of the propagator in terms of eigenvalues and eigenvectors:
\begin{equation}\label{Vpmprop}
\rme^{-\rmi K t}= \rme^{-\rmi z t} V_+ + \rme^{\rmi z t} V_-  , \qquad V_\pm:=\frac \openone 2 \pm \frac K {2z}  \quad \Rightarrow \quad {V_\pm}^2=V_\pm, \quad V_\mp V_\pm = 0,  \quad V_++V_-=\openone.
\end{equation}
By the definition of $V_\pm$ and its properties, we obtain
\begin{equation}\label{KV}
K= zV_+-zV_-, \qquad KV_\pm\psi= \pm z V_\pm\psi.
\end{equation}
So, the quantities $V_\pm\psi$ are the right eigenvectors of the n-H operator $K$ (if $\psi$ is such that both are different from zero). However, $V_-^\dagger V_+\neq 0$ and the two eigenvectors are not mutually orthogonal. 

\begin{remark}\label{rem:zsign}
In \eqref{defz} the choice of the sign of the square root of $\varkappa $ is arbitrary. If we take a shift of $2\pi$ in the phase $\chi$, we get an exchange of the two eigenvalues:
\begin{equation}
\chi \to \chi + 2 \pi \qquad \Rightarrow \qquad \pm z \to \mp z , \qquad V_\pm \to V_\mp.
\end{equation}
In any case this exchange leaves invariant $K$ and $\rme^{-\rmi Kt}$, see \eqref{evolgpm}, \eqref{gpmz}, \eqref{KV}.
\end{remark}

Let us denote by $u_\pm$ the normalized right eigenvectors of $K$:
\begin{equation}\label{eigeneq}
Ku_\pm=\pm z u_\pm , \qquad \norm{u_\pm}=1.
\end{equation}
By using the  expression of $K$ we can find explicitly these two eigenvectors, but we have to distinguish some subcases. The proof of the following proposition is given in Appendix \ref{sec:proof2}.

\begin{proposition}\label{prop:eigenvec}
We take $\varkappa\neq 0$. 
\begin{itemize}
\item For $\beta\delta\neq 0$ we get
\begin{equation}\label{upm}
u_\pm= \frac 1 {\sqrt{\abs{\alpha\pm z}^2+\abs \delta^2}}\begin{pmatrix} \alpha \pm z \\ \delta \end{pmatrix} = \frac {\Phi} {\sqrt{\abs{\alpha\mp z}^2+\abs \beta^2}}\begin{pmatrix} -\beta \\ \alpha \mp z  \end{pmatrix}, \qquad \Phi:= \frac {\left(\overline \alpha \mp\overline z\right)\delta} {\abs{\alpha\mp z}\abs\delta};
\end{equation}

\item For $\beta\delta=0$ we have $\alpha\neq 0$ and we chose $z=\alpha$. Then, we get $u_+$ from the first expression in \eqref{upm} and $u_-$ from the second one, without the phase factor $\Phi$.
\end{itemize}
Moreover, we have the following equivalencies:
\begin{equation}\label{orthogonal}
\langle u_+|u_-\rangle =0   \qquad \Leftrightarrow \qquad (K/z)^\dagger =K/z \qquad \Leftrightarrow \qquad V_\pm={V_\pm}^\dagger.
\end{equation}
\end{proposition}

In \eqref{upm} $\Phi$ is only a phase factor needed to have the equality. Equation \eqref{orthogonal} says that $u_+$ and $u_-$ are mutually orthogonal only when $K$ is proportional to a self-adjoint operator.

When $\varkappa\to 0$ also $z\to 0$ and the two right eigenvectors $u_\pm$ collapse in a single one, proportional to $\phi_0$ \eqref{phi01}. It is possible to introduce also the left eigenvectors $ {v_\pm}^{\dagger}\propto \left( \alpha \pm z,\; \beta\right)$; for $\varkappa\to 0$ they collapse to $\phi_1$.

By \eqref{kappa2} and \eqref{defz}, the positivity condition \eqref{R+k+tilde} can be written as
\begin{equation}\label{ctilde-k}
\tilde c\geq k\equiv \sqrt{2\left(\abs\alpha^2+\abs{\beta\delta}-\abs\varkappa\right)+\left(\abs\beta-\abs\delta\right)^2+4(\IM z)^2}\geq 2\abs{\IM z}.
\end{equation}

Having the explicit expression of the propagator, by \eqref{C2survprob} and \eqref{C2rhot}, one can compute the conditional state up to the first jump and the survival probability. Let us consider the survival probability in the case of a pure initial state $\psi$:
\begin{equation}\label{coeffupm}
\psi =x_+u_++x_-u_-, \qquad x_\pm\in \Cbb, \qquad \abs{x_+}^2+\abs{x_-}^2 +2 \RE (\overline x_+ x_- \langle u_+|u_-\rangle)\equiv\norm\psi^2=1.
\end{equation}
Then, we have
\begin{equation}\label{survPupm}\begin{split}
&\Pbb[N(t)=0]= \rme^{-\tilde c t} \norm{\psi(t)}^2, \qquad  \psi(t)=\rme^{-\rmi K t}\psi = \rme^{-\rmi zt}x_+u_+ +\rme^{\rmi zt}x_-u_-, 
\\
&\norm{\psi(t)}^2=\abs{x_+}^2 \rme^{-2t\IM z}+ \abs{x_-}^2 \rme^{2t\IM z}+ 2 \RE (\rme^{2\rmi t \RE z}\;\overline x_+ x_- \langle u_+|u_-\rangle).
\end{split}\end{equation}

By the condition \eqref{ctilde-k} the survival probability cannot increase. However, for certain choices of the parameters, we could have a not vanishing survival probability for long times: $\lim_{t\to +\infty}\Pbb[N(t)=0]>0$. To have this case (but with $K\neq K^\dagger$), we need $\tilde c=k$ and $k= 2\abs{\IM z}>0$, which means $\abs\beta=\abs\delta$, \ \ $\abs\alpha^2 +\abs{\beta\delta}=\abs\varkappa$.
So, we are in the class of models parameterized by \eqref{refex}, and we have to take
\[
\abs\beta=\abs\delta, \qquad \rme^{2\rmi \arg \alpha}=\rme^{\rmi \theta}, \qquad \rme^{\rmi \theta}\neq 1, \qquad \tilde c=2\abs{\IM z} >0.
\]
This gives
\[
\varkappa=\rme^{\rmi \theta} \left(\abs\alpha^2+\abs\beta^2\right), \qquad \IM z \neq 0,
\]
\[
\Pbb[N(t)=0]= \rme^{-2 t\abs{\IM z}} \left[\abs{x_+}^2 \rme^{-2t\IM z}+ \abs{x_-}^2 \rme^{2t\IM z}+ 2 \RE (\rme^{2\rmi t \RE z}\;\overline x_+ x_- \langle u_+|u_-\rangle)\right] ;
\]
then, one of the first two contributions survives for long times.

A special case of interest is when $K\neq K^\dagger$, but its eigenvalues are real, see for instance \cite{SCr15}. This case is equivalent to $k>0$ and $\varkappa> 0$, i.e.
\[
\varkappa\equiv \alpha^2 + \beta \delta >0,\qquad k^2\equiv 2\left(\abs\alpha^2+\abs{\beta\delta}-\varkappa \right)+\left(\abs \beta-\abs \delta\right)^2>0.
\]
By taking $z:=\sqrt\varkappa >0$ and using \eqref{coeffupm}, \eqref{survPupm}, \eqref{Twtj}, we have
\[
\Pbb[N(t)=0]=\rme^{-\tilde c t}\norm{\psi(t)}^2, \qquad \norm{\psi(t)}^2=1+ 2 \RE \left[\left(\rme^{2\rmi z t }-1\right)\overline x_+ x_- \langle u_+|u_-\rangle\right],
\]
\[
p_T(t)=\rme^{-\tilde c t}\left\{\tilde c \norm{\psi(t)}^2+4 z \IM (\rme^{2\rmi z t } \,\overline x_+ x_- \langle u_+|u_-\rangle) \right\}.
\]
By $k>0$ we have that $\tilde c>0$ and, so, $p_T(t)$ is a true probability density. If we exclude the case of $K$ proportional to a self-adjoint operator, this probability density has a very peculiar oscillating behaviour. 

An example with real eigenvalues is given \cite{SCr15}; the proposed model is
\[
K=s\begin{pmatrix}\rmi\sin \phi & 1 \\ 1 &-\rmi \sin \phi\end{pmatrix}, \quad s\in \Rbb \qquad \Rightarrow \qquad \varkappa=s^2(\cos\phi)^2, \quad z=\abs{ s \cos\phi}, \quad k= 2\abs{s\sin\phi};
\]
$K$ is Hermitian only for $s\sin\phi=0$. As example we take 
\[
s>0, \qquad \phi=\pi/4, \qquad \tilde c=k, \qquad x_+=\frac{1+\rmi}2,\qquad x_-=\frac\rmi {\sqrt 2};
\]
then, we have
\[
\varkappa=\frac{s^2}2 \, \qquad z=\frac s {\sqrt 2}, \qquad k=\sqrt 2\, s=2z, \qquad u_\pm=\frac 12 \begin{pmatrix}\rmi \pm 1 \\ \sqrt 2\end{pmatrix}, \qquad \langle u_+|u_-\rangle=\frac{1+\rmi}2 ,
\]
\[
\overline x_+ x_- \langle u_+|u_-\rangle=\frac \rmi{2\sqrt 2}, \qquad \psi=x_+u_++x_-u_-, \qquad \norm{\psi}^2=1,
\]
\[
\norm{\psi(t)}^2=1-\frac{\sin \sqrt 2 st}{\sqrt 2}, \qquad p_T(t)=\rme^{-\sqrt 2 s t} s\left\{\sqrt 2-\sin\sqrt 2 s t + \cos \sqrt 2 st \right\}.
\]
With this choice of the parameters, we have not only a probability density with oscillations, but also with a sequence of zeros: $p_T(t)=0$ for $\sqrt 2 s t= 2n\pi+ 3\pi/4$.

Sometimes, the case of n-H Hamiltonian with real eigenvalues is considered of interest when it is possible to modify the internal product in such a way that the evolution becomes unitary. However, as stated in \cite{SCr15}, this means to go back to an Hermitian Hamiltonian. Here we have studied the n-H case without changing the internal product  and we have shown that we get very unusual probability densities for the survival time.

Another point of interest is to see what happens when the parameters are modified along a path which goes around an EP. This procedure exchanges the two eigenvectors and the two eigenvalues; as discussed in Remark \ref{rem:zsign}, there is no change in the dynamics. On the contrary, if this continuous modification of the parameters is applied to the preparation, there is an exchange $u_+ \leftrightarrow u_-$ and the properties of the preparation change; about this point there are also experimental results \cite{Gao+15}.

\section{A piecewise dynamics interspersed by jump operators at random times}\label{sec:inters}

Now we consider the case of a (unitary or dissipative) dynamics interrupted at random times by some action (typically given by a quantum channel). It can represent a smooth dynamics modified by interactions with external environments or measuring apparatuses, and it has been used to model dissipative evolutions, mainly in solid state physics and in statistical mechanics \cite{GoHa06,ScLu07,DDG22,NaGu23}; it is also a way to construct a non-Markovian mean dynamics and to test the proposed quantifiers of non-Markovianity \cite{Bud04,Vacc+11,Vacc13,Vacc20,Vacc+21}. 

We start from the counting process $N(t)$ introduced in \eqref{def:N}; the measure $\nu(\rmd u)$ satisfies Assumptions \ref{ass:U}, \ref{ass:Poiss}. Under $\Qbb$,  $N(t)$ is a Poisson process of intensity $\nu(\Uscr)$. Under $\Pbb$, we want   $N(t)$ to acquire a prefixed law, not depending on the quantum dynamics. As the $\Pbb$-law of $N(t)$ is determined by the stochastic intensity $\lambda(t)$, given by \eqref{lambda}, to have independence from the quantum state we have to ask
\begin{equation}\label{R=lambda}
R(t)=\lambda(t)\openone,
\end{equation}
with $\lambda(t)$ assigned once for all. The operator $R(t)$ is defined in \eqref{def:R(t)} and depends on the operators $J_j(u,t)$, introduced in Assumption \ref{ass:HJL}. To have \eqref{R=lambda},  we  need to take
\begin{equation}\label{OtoJ}
J_j(u,t)=\sqrt{\lambda(t)}\,O_j(u,t), \qquad 
\int_\Uscr\nu(\rmd u)\sum_{j=1}^{d_1}O_j(u,t)^\dagger O_j(u,t)=\openone, \qquad  O_j(u,t)\in \Bscr(\Hscr).
\end{equation}

Now, we can introduce all the operators on $\Tscr(\Hscr)$ needed in the SMEs \eqref{linearSME++}, \eqref{nonlinearSME}.
Firstly, we define 
\begin{equation}\label{Ocal}
\Ocal(u,t)[\rho]:=\sum_jO_j(u,t)\rho O_j(u,t)^\dagger ,  \qquad  \Ycal_t(A):=\int_A \Ocal(u,t)\nu(\rmd u) , \quad \forall A\in \Bscr(\Uscr).
\end{equation}
It is possible to check that $\Ycal_t$ is an \emph{instrument} with value space $\big(\Uscr,\; \Bscr(\Uscr)\big)$, as it satisfies the three defining properties of an instrument, given after Eq.\ \eqref{Icalinstr}. By  \eqref{Jcalu}, \eqref{OtoJ}, \eqref{Ocal} we get
\begin{equation}\label{Jlambda}
\Jcal(u,t)[\rho]= \lambda(t)\Ocal(u,t)[\rho] ,\qquad \Gcal(t)[\rho]= \lambda(t)\Ycal_t(\Uscr)[\rho];
\end{equation}
$\Lcal_0(t)$ is defined by \eqref{Lcalt}. Due to the properties \eqref{R=lambda}, \eqref{OtoJ}, $\Lcal(t)$, defined in \eqref{Lcalx}, can be written also as
\begin{equation}\label{lambdaLcal}
\Lcal(t)[\rho]=\Lcal_0(t)[\rho] +\lambda(t)\left(\Ycal_t(\Uscr)[ \rho ]-\rho\right).
\end{equation}
Then, the intensity $I(u,t)$ \eqref{def:pt} takes the form
\begin{equation}\label{fFI}
I(u,t)=\Tr\left\{\Jcal(u,t)\big[\rho(t_-)\big]\right\}= \lambda(t)V(u,t),\qquad V(u,t):=\Tr\{\Ocal(u,t)^*[\openone] \rho(t_-)\};
\end{equation}
note that we have
\begin{equation}\label{normposV}
V(u,t)\geq 0, \qquad \int_\Uscr V(u,t)\nu(\rmd u)= \Tr\{\Ycal_t(\Uscr)^*[\openone] \rho(t_-)\}=1.
\end{equation}

Then, the linear SME \eqref{linearSME++} and the nonlinear one \eqref{nonlinearSME} can be written as
\begin{subequations}\label{fFSMEs}
\begin{equation}\label{fFlin}
\rmd  \sigma(t)= \Lcal_0(t)[\sigma(t_-)]  \rmd t- \big(\lambda(t)-\nu(\Uscr)\big)\sigma(t_-)\rmd t
+ \int_{\Uscr} \Bigl( \lambda(t)\Ocal(u,t)[ \sigma(t_-) ]-\sigma(t_-) \Bigr) \Pi(\rmd u,\rmd t),
\end{equation}
\begin{equation}\label{fFnonlin}
\rmd \rho (t)= \Lcal_0(t)\big[\rho(t_-)\big]  \rmd t
+ \int_{\Eo_t} \Bigl( V(u,t)^{-1}\Ocal(u,t)[ \rho(t_-) ]-\rho(t_-) \Bigr) \Pi(\rmd u,\rmd t).
\end{equation}
\end{subequations}
Let us stress that $\lambda(t)$ does not appear explicitly in the nonlinear SME, but it determines the distribution of the times of the jumps and, so, it has effect on the law of $\Pi(\rmd u,\rmd t)$. Moreover, the nonlinearity appears only in the jump term through $ V(u,t)$, not in the evolution in between two jumps, determined by $\Lcal_0(t)$ alone. The operator $\Lcal_0(t)$ generates an usual Markovian dynamics, eventually it can be a pure unitary dynamics. The dynamics determined by $\Lcal_0(t)$ represents a ``free dynamics", while the jump operators give the interactions at random times.  The meaning of $V(u,t)$ is given by equation \eqref{condprob:uinA}, the conditional probability of having a jump of type $u\in A$, 
$A\subset \Uscr$ knowing that there is a jump at time $t$; by \eqref{fFI}, this conditional probability becomes 
\begin{equation} 
\Pbb\big[\Pi(A,[t_-,t_+])=1\big|N(t_+)-N(t_-)=1,\Fscr_{t_-}\big]=  \int_A V(u,t)\nu(\rmd u).
\end{equation}
Through the quantity $V(u,t)$ (defined in \eqref{fFI}), this probability depends on the conditional state $\rho(t_-)$ before the jump and on the structure of the jump operators 
$O_j(u,t) $ introduced in \eqref{OtoJ}, \eqref{Ocal}. 
 
Obviously, if in \eqref{OtoJ} we take $\sum_{j=1}^{d_1}O_j(u,t)^\dagger O_j(u,t)\propto\openone$, we loose the dependence on the quantum state in $V(u,t)$ too.

\subsection{The counting process}\label{sec:cppre}
Here we fix the structure of $\lambda(t)$ by asking $N(t)$ \eqref{def:N} to be a simple counting process with stopping times $T_m$ (the time at which we have the $m$th count) and intervals between two counts $S_m$ \cite{Prott04}:
\begin{equation}\label{N(t)intersp}
N(t)=\sum_{m\geq 1}\ind_{t\geq T_m}, \qquad S_m=T_m-T_{m-1}, \quad m\geq 1, \quad T_0=0, \qquad T_m=\sum_{i=1}^m S_i.
\end{equation}
Under $\Qbb$,  $N(t)$ is a Poisson process of intensity $\nu(\Uscr)$ and, so, the random variables $S_i$ are independent and identically distributed with exponential distribution. By the choice of $\lambda(t)$ we can make the waiting times to acquire any distribution \cite{DaVJ08}; for simplicity, we ask  that, under $\Pbb$, the intertimes $S_m$ acquire predetermined distributions, but that they remain independent, the case considered in \cite{DDG22} for instance.  

We take a sequence of distributions on $\Rbb^+$ having a density:
\begin{equation}\label{fiFi}
f_i(t)\geq 0, \qquad \int_0^{+\infty}f_i(s)\rmd s=1, \qquad F_i(t):=\int_0^tf_i(s)\rmd s, \qquad i\geq 1.
\end{equation}
Then, we assume the intertimes $S_i$ to be independent and distributed according these densities (under the physical probability $\Pbb$); so, we take
\begin{subequations}
\begin{equation}\label{nojumpp1}
\Pbb[N(t)=0]=\Pbb[T_1>t]= \Pbb[S_1>t]=1-F_1(t),
\end{equation}
\begin{multline}
\Pbb[N(t)=m] = \Pbb[T_m\leq t, \, T_{m+1}>t]= \Pbb\big[ S_1+\cdots+S_m\leq t,\; S_1+\cdots+S_{m+1}>t\big]
\\ {}=\int_0^{t}\rmd t_1\int_{t_1}^t\rmd t_{2} \cdots \int_{t_{m-1}}^t\rmd t_m p_t\big(t_m;\ldots;t_1\big), \qquad m\geq 1,
\end{multline}
\begin{equation}\label{Npdens}
p_t\big(t_m;\ldots;t_1\big)= \big(1-F_{m+1}(t-t_m)\big)\Big(\prod_{i=1}^m f_i(t_i-t_{i-1})\Big), \qquad t_0=0.
\end{equation}
\end{subequations}

To have a regular counting process with
\begin{equation}\label{N(t)to1}
\sum_{m=0}^{+\infty}\Pbb[N(t)=m] = 1,
\end{equation}
we assume that the following condition holds:
\begin{equation}\label{rmt}\begin{split}
&\lim_{m \to +\infty}r_m(t)=0,
\\
&r_m(t):=\int_0^{t}\rmd t_1\int_{t_1}^t\rmd t_{2} \cdots \int_{t_{m-1}}^t\rmd t_m \int_{t_m}^t\rmd t_{m+1}F_{m+2}(t-t_{m+1})\Big(\prod_{i=1}^{m +1} f_i(t_i-t_{i-1})\Big).
\end{split}
\end{equation}
\begin{proof}[{Computations}] By separating $1$ and $-F_{m+2}(t-t_{m+1})$ in the integral giving $\Pbb[N(t)=m+1]$, we get
\begin{equation}\label{partialsum}
\Pbb[N(t)=m+1] = r_{m-1}(t)-r_m(t), \qquad \sum_{k=0}^{m+1}\Pbb[N(t)=k]=1- r_m(t).
\end{equation}
This gives  $r_{m-1}(t)-r_m(t)\geq 0$ and, by  \eqref{rmt}, we have $r_m(t)\geq 0$; this gives
\[
0\leq r_m(t)\leq r_{m-1}(t)\leq 1.
\]
So, the sequence $r_m(t)$ is not increasing and bounded from below; then, it has a limit for $m\to +\infty$. In \eqref{rmt} we assume this limit to be zero and, by \eqref{partialsum}, we get \eqref{N(t)to1}.
\end{proof}

Finally, (under $\Pbb$) the stochastic intensity (or compensator) of $N(t)$ takes the form \cite[Sec.\ 14.1]{DaVJ08}
\begin{equation}\label{lambda+f+F}
\lambda(t)= \sum_{i=1}^{+\infty}\ind_{\{T_{i-1}\leq t <T_i\}}\,h_i(t-T_{i-1}), \qquad h_i(t)=\frac{f_i(t)}{1-F_i(t)}\equiv -\frac{\rmd \ }{\rmd t} \ln \bigl( 1-F_i(t)\bigr).
\end{equation}
The expression \eqref{lambda+f+F} is consistent with the heuristic interpretation of $\lambda(t)\rmd t$ as conditional probability of a jump in the infinitesimal time interval, see eq.\ \eqref{jumpprob}; the function $h_i(t)$ is known as \emph{hazard rate}. As $\lambda(t)$ does not depend on the quantum state, also the distribution of the jump times is independent from the state.

\subsection{Typical trajectories and exclusive probability densities} \label{sec:interspex}
Let us recall that the linear SME \eqref{fFlin} is given under the reference probability $\Qbb$, under which $\Pi(\rmd u,\rmd t) $ is a Poisson point process. As discussed in Sec.\ \ref{sec:nprob}, the physical probability $\Pbb$ is determined by the $\Qbb$-density $p(t)$ given in \eqref{def:pt}. Instead, the nonlinear SME \eqref{fFnonlin} gives the evolution of the normalized conditional state of the quantum system. By using the notion of ``typical trajectories" of Sec.\ \ref{sec:ttrajs}, we can get the structure of probabilities and quantum states in a more explicit form.

According to Assumption \ref{ass:HJL} all the operators involved in the SMEs can be random and can have a dependence on the past. Here we simplify the situation by the follwing assumption.
\begin{assumption} We assume that there is only a dependence on the last value of $N(t)$, without other extra-noise. So, in \eqref{Lcalt} and \eqref{Ocal} we take
\begin{equation}\label{LOY+N}
\Lcal_0(t) \to \Lcal_0\big(t;N(t)\big), \qquad \Ocal(u,t) \to \Ocal\big(u;N(t)\big), \qquad \Ycal_t(A) \to \Ycal_{N(t)}(A).
\end{equation}
When $N(t)$ is fixed, $\Lcal_0(t)$, $\Ocal(u,t)$, $\Ycal_t(A)$ are non random and $\Ocal(u,t)$ and $\Ycal_t(A)$ loose also the explicit time dependence.
\end{assumption}

Let us consider the no jump propagator of Definition \ref{def:Stt0}; due to \eqref{R=lambda} and \eqref{LOY+N}, equations
\eqref{def:propag} give, under the condition $N(t_0)=k$,
\begin{equation*}
\Scal(t,t_0)=\rme^{-\int_{t_0}^t \lambda(s)\rmd s}\Scal_0(t,t_0;k), \qquad \frac {\rmd \ }{\rmd t}\,\Scal_0(t,t_0;k)[\rho]=\Lcal_0(t;k) \circ\Scal_0(t,t_0;k)[\rho], \quad \Scal_0(t_0,t_0)=\id .
\end{equation*}
The propagator $\Scal_0(t,t_0;k)$ is CP and trace preserving. From \eqref{lambda+f+F} we have also
\[
\prod_{i=1}^m \exp\left\{-\int_{t_{i-1}}^{t_i}h_i(s)\rmd s\right\}h_i(t_i)=\prod_{i=1}^m f_i(t_i-t_{i-1}).
\]
By taking into account the trajectories $\tau(t;m)$ \eqref{trajec} and the expressions \eqref{decompsigma} for the solution of the linear SME, we get, from \eqref{lambda+f+F}, \eqref{Jlambda},
\begin{subequations}
\begin{equation}
\sigma(t_-)= \exp\left\{ \nu(\Uscr)t\right\}\big(1-F_1(t)\big) \Scal_0(t,0;0)[ \rho_0], \qquad \text{for} \ m=0,
\end{equation}
\begin{multline}
\sigma(t_-)= \big(1-F_{m+1}(t-t_m)\big)\exp\left\{ \nu(\Uscr)t\right\} \Big(\prod_{i=1}^m f_i(t_i-t_{i-1})\Big)\Scal_0(t,t_m;m)\circ \Ocal(u_m;m)
\\ {}\circ \Scal_0(t_m,t_{m-1};m-1)\circ
\cdots \circ \Ocal(u_2;2)\circ \Scal_0(t_2,t_{1};1)\circ \Ocal(u_1;1)\circ \Scal_0(t_1,0;0)[ \rho_0], \qquad \text{for} \ m\geq 1.
\end{multline}
\end{subequations}
Similarly, from \eqref{decomprho} and \eqref{fFI} we get the structure of the conditional states:
\begin{subequations}\label{rhointersp}
\begin{equation}
\rho (t_-)= \Scal_0(t,0)\big[ \rho_0\big], \qquad  \text{for} \ m=0, \qquad \qquad V(u_j,t_j)=\Tr\{\Ocal(u_j;j)^*[\openone] \rho({t_j}_-)\},
\end{equation}
\begin{multline}
\rho (t_-)= \Scal_0(t,t_m;m)\circ \frac{\Ocal(u_m;m)}{V(u_m,t_m)}\circ \Scal_0(t_m,t_{m-1};m-1)\circ
\\{} \circ\cdots \circ \frac{\Ocal(u_2;2)}{V(u_2,t_2)}\circ \Scal_0(t_2,t_{1};1)\circ \frac{\Ocal(u_1;1)}{V(u_1,t_1)}\circ \Scal_0(t_1,0;0)[ \rho_0], \qquad \text{for} \ m\geq 1.
\end{multline}
\end{subequations}
As already underlined after \eqref{fFnonlin},  the structure of $\rho(t)$ is independent from the probabilities \eqref{fiFi}; only the distribution of the jump times depends on these probabilities.

By using the trace preserving property of $\Scal_0(t,t_0;k)$, from \eqref{N0} we get the no jump probability \eqref{nojumpp1};
note that there is no dependence on the quantum system, as wanted. Moreover, by \eqref{genexclpd}, we get the structure of the exclusive probability densities:
\begin{multline}\label{exclpdintersp}
p_t\big(u_m,t_m;\ldots;u_1,t_1\big)= \big(1-F_{m+1}(t-t_m)\big)\Big(\prod_{i=1}^m f_i(t_i-t_{i-1})\Big)\Tr\{\Ocal(u_m;m)\circ
\\ {}\circ \Scal_0(t_m,t_{m-1};m-1)\circ
\cdots \circ \Ocal(u_2;2)\circ \Scal_0(t_2,t_{1};1)\circ \Ocal(u_1;1)\circ \Scal_0(t_1,0;0)[ \rho_0]\}.
\end{multline}
By integrating \eqref{exclpdintersp} on all possible types of jumps, by using the measure $\nu(\rmd u_1)\cdots \nu(\rmd u_m)$, we get the exclusive probability densities for the jumps of $N(t)$, which turn out to be the quantities \eqref{Npdens}. Indeed, only the operators 
$\int_\Uscr\nu(\rmd u_j)\Ocal(u_j;j)=\Ycal_{j}(\Uscr)$ and $\Scal_0(t_j,t_{j-1};j-1)$ appear in the formula, which are all trace preserving; this means that the term with the trace gives exactly a factor 1 and only the terms not involving the quantum system survive, as it must be by construction,

\subsection{The mean state}\label{sec:mstinter}
The results given in Sec.\ \ref{sec:interspex} allow to obtain a recursive structure for the mean state $\eta(t)$ \eqref{meanst}.
By using the exclusive probability densities \eqref{exclpdintersp} to take the mean of the expressions \eqref{rhointersp} and by summing up over the index $m$, we get the explicit expression of the mean state:
\begin{subequations}\label{meandintrsp}
\begin{equation}\label{eta(tilde)}
\eta(t)=[1-F_1(t)]\Scal_0(t,0;0)\big[ \rho_0\big]+\sum_{m=1}^{+\infty} \int_{0}^{t}\rmd t_{m} [1-F_{m+1}(t-t_m)]\Scal_0(t,t_{m};m)[\tilde \eta_m(t_m)],
\end{equation}
\begin{equation}
\tilde \eta_1(t_1):=f_1(t_1)  \Ycal_{1}(\Uscr)\circ \Scal_0(t_1,0;0)[ \rho_0],
\end{equation}
\begin{multline}
\tilde \eta_m(t_m)=\int_{0}^{t_m}\rmd t_{m-1}\int_0^{t_{m-1}}\rmd t_{m-2}\cdots\int_{0}^{t_2}\rmd t_1
 \Big(\prod_{i=1}^m f_i(t_i-t_{i-1})\Big) \Ycal_{m}(\Uscr)\circ \\ {}\circ \Scal_0(t_m,t_{m-1};m-1)
\circ\cdots \circ \Ycal_{2}(\Uscr)\circ \Scal_0(t_2,t_{1};1)\circ \Ycal_{1}(\Uscr)\circ \Scal_0(t_1,0;0)[ \rho_0]
\\ {} =\int_{0}^{t_m}\rmd t_{m-1} f_m(t_m-t_{m-1})\Ycal_{m}(\Uscr)\circ \Scal_0(t_m,t_{m-1};m-1)[\tilde \eta_{m-1}(t_{m-1})], \qquad m\geq 2.
\end{multline}
\end{subequations}
The map $\rho_0 \mapsto \eta(t)$ is CP and trace preserving by construction; however, these properties can be checked directly. Indeed, the maps $\Scal_0(t,s;k)$ and $\Ycal_{k}(\Uscr)$ are CP and trace preserving and the various real functions representing probabilities are non negative; then, the map is CP. By taking the trace all the operators give a factor 1; only the probability densities and the survival probabilities remain. Then, each term in the sum gives $\Pbb[N(t)=m]$ and the sum is 1 by \eqref{N(t)to1}; this gives the trace preserving property.

When $\Scal_0(t,t_0;k)$ is independent from $k$, the propagator $\Scal_0(t,t_0)$ represents the ``free" dynamics, without the random interactions. This dynamics is interrupted at random times by the quantum channels $\Ycal_k(\Uscr)$ (CP and trace preserving linear maps).

Often the interest for applications is on the mean dynamics \cite{Bud04,Vacc+11,GoHa06,ScLu07,DDG22,NaGu23,Vacc13,Vacc+21,Vacc20,Vacc14,Vacc16,MSV20,Bud13a,Bud13b,ClaB71,MSM18}. The mean state \eqref{meanst} satisfies equation \eqref{masteq} with generator \eqref{lambdaLcal}. However,  apart from the case of exponential distributions, $\lambda(t)$ is random  and, as already underlined in Sec.\ \ref{sec:mean+cond}, \eqref{masteq} is not a closed equation. When $\lambda(t)$ is non random, $\Pi(\bullet,\bullet)$ is a Poisson point process; without extra randomness in $\Lcal(t)$, \eqref{lambdaLcal} implies that the  operation of taking the mean disappears from \eqref{masteq}, which reduces to an usual Markovian quantum master equation. This is the case of \cite{ClaB71}, where the interarrival times are indeed independent and identically distributed with exponential distribution. This proposal was before the GKSL generator \cite{GKS76,L76}, and it opened a way to treat dissipative quantum systems in solid state physics and in statistical mechanics.

When the interest is in the mean dynamics, eqs.\ \eqref{meandintrsp} can be taken as starting point \cite{DDG22}. The non exponential case is of interest because it produces a non Markovian dynamics; for instance, in \cite{ScLu07} the effects of densities with heavy tails are studied.  Then, the computation of the mean state is attached by techniques based on Laplace transform \cite{Bud04,GoHa06,ScLu07,ClaB71,DDG22,Vacc16,Bud13a,Vacc14,Vacc13,MSV20,Vacc20}.
In some applications in quantum statistical mechanics, the jump operator is taken to be a ``reset" to an initial state, which means to have $\Ycal(\Uscr)[\rho]=\Tr\{\rho\}\rho_{\rm reset}$ \cite{MSM18,NaGu23,DDG22,DDG22b}; for instance, the interest is in constructing models with a nonthermal steady-state density matrix.

Another point of research is to find a closed evolution equation for the mean state. With some restrictions, also some cases of non exponential distributions can give rise to a closed evolution equation involving a memory kernel or to an integro-differential equation \cite{Bud04,Vacc+11,Vacc20,Vacc14,Vacc13,Vacc16,MSV20,Bud13a,Bud13b}. 
Structures in the class of \eqref{meandintrsp} have been introduced by analogies with classical stochastic processes, under the names of ``quantum renewal processes" and ``quantum semi-Markov dynamics".
Much effort has been devoted also to quantifying the degree of non-Markovianity of this class of quantum evolutions, mainly by using  quantifiers based on the trace distance \cite{Vacc+11,Vacc20,Vacc+21,Vacc14,Vacc13,Vacc16,MSV20,Bud13a,Bud13b,BreuerLP09}.

Note that in the present approach there are closed SDEs giving the evolution of the ``hybrid system" (the SMEs \eqref{fFSMEs} together with the stochastic intensity \eqref{lambda+f+F}, which determines $N(t)$), but not a closed equation for $\eta(t)$. This gives more freedom in the structure of $\eta(t)$; all the elements in \eqref{meandintrsp} can depend on $m$, the number of jumps up to the considered time: the no-jump propagator $\Scal_0(t,t_{m};m)$, the jump operator $\Ycal_{m}(\Uscr)$, the waiting-time density $f_m(t)$.

\subsection{Probabilities and conditional states}

The general model of equations \eqref{R=lambda}--\eqref{fFSMEs} can include the case of different quantum channels applied at random to the system of interest, or even the case of measurements at random times; this generalizes other proposals, such as the ``projective measurements at random times" introduced in \cite{DDG22}, see also \cite{VKB08}. To put in evidence this point, we consider the case of no extra-randomness in the various operators and the same waiting time distribution at every jump; so, in \eqref{fiFi} and \eqref{LOY+N} we take
\begin{gather*}
f_i(t) \ \to \ f(t), \qquad  \Ocal(u,t) \ \to \ \Ocal(u), \qquad \Ycal_t(A) \ \to \ \Ycal(A),  \qquad  \Scal_0(t,t_0;k) \  \to \  \Scal_0(t,t_0).
\end{gather*}

As an example, we consider the event of having a jump of type $u\in A$ in $(t_0,t)$, no other jumps in $(t_0,t)$,  and any number of jumps in $(0,t_0)$: \quad $t> t_0\geq 0$, \quad $A\in \Bscr(\Uscr)$,
\begin{equation}\label{Ettau}
\Eo(A,t_0,t):= \left\{N(t)-N(t_0)=1, \ \Pi\big(A; (t_0, t)\big)=1\right\} .
\end{equation}
By using the ``instrument" defined in Sec.\ \ref{sec:POVM},  we get from \eqref{Icalinstr}
\begin{subequations}\label{IcalAttau0}
\begin{multline}\label{IcalAttau}
\Ical_t\big(\Eo(A,t_0,t)\big)[\rho_0]=
\int_{t_0}^{t} \rmd s [1- F(t-s)]\Scal_0(t,s)\circ\Ycal(A)\circ\Scal_0(s,t_0)\circ
\\ {}\circ\Big\{f(s)\Scal_0(t_0,0)\big[ \rho_0\big] 
+\sum_{m=1}^{+\infty}\int_0^{t_0}\rmd t_m f(s-t_m)\Scal_0(t_0,t_m)\big[ \tilde \eta_m(t_m)\big]\Big\},
\end{multline}
\begin{equation}
\tilde \eta_1(t_1)=f(t_1)  \Ycal(\Uscr)\circ \Scal_0(t_1,0)[ \rho_0],
\end{equation}
\begin{equation}
\tilde \eta_m(t_m)=\int_{0}^{t_m}\rmd t_{m-1} f(t_m-t_{m-1})\Ycal(\Uscr)\circ \Scal_0(t_m,t_{m-1})[\tilde \eta_{m-1}(t_{m-1})]
, \qquad m\geq 2.
\end{equation}
\end{subequations}
Then, the probability of the event \eqref{Ettau} is
\begin{equation}\label{probEtA}
\Pbb\big[\Eo(A,t_0,t)\big|\rho_0\big]= \Tr\big\{\Ical_t\big(\Eo(A,t_0,t)\big)[\rho_0]\big\},
\end{equation}
and, by \eqref{Econd}, the related conditional state at time $t$ is
\begin{equation*}
\rho\big(t\big|\Eo(A,t_0,t),\,\rho_0\big)=\frac{\Ical_t\big(\Eo(A,t_0,t)\big)[\rho_0]}{\Tr\big\{\Ical_t\big(\Eo(A,t_0,t)\big)[\rho_0]\big\}}.
\end{equation*} 

In taking the trace of the quantity \eqref{IcalAttau}, the propagator $\Scal_0(t,s)$ disappears because it is trace preserving. Instead, $\Ycal(A)$ is not trace preserving and the quantum system does not disappear from the probabilities \eqref{probEtA}. When $A=\Uscr$ all the operators inside the trace are trace preserving and the contribution of the quantum system disappears from the probability \eqref{probEtA} for $A=\Uscr$: in the construction given in this section, the probabilities related to the counting process $N(t)$ alone are predetermined and depend only on the distributions \eqref{fiFi}. Indeed, we have 
$\Eo(\Uscr,t_0,t)=\{N(t)-N(t_0)=1\}$,
\begin{multline}\label{N=1}
\Pbb[N(t)-N(t_0)=1]=\Tr\{\Ical_t\big(\Eo(\Uscr,t_0,t)\big)[\rho_0]\}=\int_{t_0}^t \rmd s \left[1-F(t-s)\right]\biggl\{f(s)
\\ {}+ \sum_{m=1}^{+\infty} \int_0^{t_0}\rmd t_m f(s-t_m)\int_0^{t_m}\rmd t_{m-1} \cdots \int_0^{t_2}\rmd t_1 f(t_m-t_{m-1}) \cdots f(t_2-t_1)f(t_1)\biggr\}.
\end{multline}

The event given in \eqref{Ettau} concerns only the times $(t_0,t)$; the jumps in the time interval $(0,t_0)$ are not registered. However, the instrument \eqref{IcalAttau0} can not be written as some operator applied to the mean state $\eta(t_0)$ \eqref{meandintrsp}. The only exception is the exponential distribution: for 
$f(t)=\lambda \rme^{-\lambda t}$, eqs.\ \eqref{meandintrsp} and \eqref{IcalAttau0} give
\begin{equation}\label{It0texp}
\Ical_t\big(\Eo(A,t_0,t)\big)[\rho_0]=\int_{t_0}^t\rmd s \left[1-F(t-s)\right]f(s-t_0)\Scal_0(t,s)\circ\Ycal(A)\circ \Scal_0(s,t_0)[\eta(t_0)].
\end{equation}
The right hand side can be read as an instrument applied to the mean state at time $t_0$, but this rewriting holds only in the exponential case. This is a trace of the effects of memory in the non-exponential cases.

The choice \eqref{Ettau} is just an example of a possible event; in general, more interactions and different results $A_j$ can be involved. The set $A_j$ can be interpreted as the result of a measurement at random time or only as a label of a type of interaction with a reservoir. 

\subsubsection{Loss of memory and revivals}\label{sec:revivals}
The models discussed in this section have been used to study the loss of memory of the initial state. When the waiting time is exponentially distributed, there  is a monotonic loss of memory in the mean state, while in the other cases this monotonic behaviour can be violated and revivals can happen with information flow from the environment to the system \cite{Vacc+11,Vacc+21}. By defining a suitable quantifier based on the trace distance, a flexible witness of the revivals has been introduced, giving a measure of non-Markovianity of the mean dynamics \cite{Vacc+11,BreuerLP09,Vacc+21}. 
However, in our approach we have not only the mean state $\eta(t)$ \eqref{meandintrsp}, but also probabilities and instruments, which can be affected by the presence of revivals.

In our formulation we have probabilities for classical events, which depend also on the quantum dynamics, and it is of interest to see how these probabilities are influenced by possible back-flows of information from the environment to the quantum system. To compare the probabilities  obtained by starting with two different quantum states at time zero we shall use the Kolmogorov distance between classical probabilities. Indeed, the trace distance was seen as a quantum analogue of the Kolmogorov distance in quantifying the dependence on the initial state \cite{Vacc+11}. 

\begin{subequations}\label{Ecalpmi}
To have an analytically computable model, we take a dissipative dynamics suggested in \cite{Vacc20,Vacc+11,MSV20,Vacc+21}:
\begin{equation}\label{Epm}
\Scal(t,s)=\openone, \qquad \Hscr=\Cbb^2, \qquad \Ycal(\Uscr)=\Ecal_\pm= \Ecal_0+\Ecal_1, \quad \Ecal_0[\rho]:= \sigma_+\rho\sigma_-, \quad \Ecal_1[\rho]:=\sigma_-\rho\sigma_+,
\end{equation}
where $\sigma_-$ and $\sigma_+$ are the usual lowering and rising operators in $\Cbb^2$. 
Then, we decompose $\Ycal(\Uscr)$ in two possible results:
\begin{equation}
\Uscr =\{0,1\},  \qquad A_i=\{ i\} , \qquad  \Ycal(A_i):=\Ecal_i.
\end{equation}
\end{subequations}

Moreover, to have higher probabilities for the results $A_i$ we modify the event \eqref{Ettau} and we consider instead the two events
\begin{equation}\label{modifevent}
\Eo(A_i,t_0,t;\geq 1):= \left\{N(t)-N(t_0)\geq1, \ \text{and the first jump is of type} \ A_i\right\} .
\end{equation}
Similarly to the case of equations \eqref{IcalAttau0}, \eqref{probEtA}, the probabilities of the events \eqref{modifevent} for the model of equations \eqref{Ecalpmi} are given by
\begin{multline}\label{P_iN(t)}
\Pbb_i(t_0,t|\rho_0):=\Pbb[\Eo(A_i,t_0,t;\geq 1)|\rho_0]= 
\int_{t_0}^{t} \rmd s \Big\{f(s)\Tr\{\Ecal_i[ \rho_0] \}
\\ {}+\sum_{m=1}^{+\infty}\int_0^{t_0}\rmd t_m f(s-t_m)\Tr\left\{\Ecal_i\big[ \tilde \eta_m(t_m)\big]\right\}\Big\},
\end{multline}
\begin{equation}\label{eta,tilde}\begin{split}
\tilde \eta_1(t_1)&=f(t_1) \Ecal_\pm  [ \rho_0],
\\
\tilde \eta_{2m+1}(t_{2m+1})&=\int_{0}^{t_{2m+1}}\rmd t_{2m} \cdots \int_0^{t_2}\rmd t_1f(t_{2m+1}-t_{2m})\cdots f(t_1)\Ecal_\pm[\rho_0],\qquad m\geq 1,
\\
\tilde \eta_{2m}(t_{2m})&=\int_{0}^{t_{2m}}\rmd t_{2m-1} \cdots \int_0^{t_2}\rmd t_1f(t_{2m}-t_{2m-1})\cdots f(t_1)\Ecal_\pm^2[\rho_0], \qquad m\geq 1.
\end{split}
\end{equation}
By summing up the two probabilities \eqref{P_iN(t)}, we get
\begin{equation}\label{Ngeq1}
\Pbb_0(t_0,t|\rho_0)+\Pbb_1(t_0,t|\rho_0)=\Pbb[N(t)-N(t_0)\geq 1].
\end{equation}
This probability  does not depend on the quantum state and has the  expression \eqref{N=1} without the factor $\left[1-F(t-s)\right]$ in the integrand. Finally, also the mean state \eqref{eta(tilde)} can be expressed in terms of the quantities \eqref{eta,tilde}:
\begin{equation}\label{eta-t0}
\eta(t_0)=[1-F(t_0)]\rho_0+\sum_{m=1}^{+\infty} \int_{0}^{t_0}\rmd t_{m} [1-F(t_0-t_m)]\tilde \eta_m(t_m).
\end{equation}

Due the dependence of the quantities  \eqref{eta,tilde} on $\Ecal_\pm$, which  diagonalizes $\rho_0$  in the eigen-basis of $\sigma_z$, a choice of initial states which tends to maximize the differences among the probabilities \eqref{P_iN(t)} is to take 
\begin{equation}\label{rho12}
\rho^0=\begin{pmatrix} 0 & 0 \\ 0 & 1\end{pmatrix}, \qquad \rho^1=\begin{pmatrix} 1 & 0 \\ 0 & 0\end{pmatrix}. 
\end{equation}
Then, the Kolmogorov distance for the  probabilities \eqref{P_iN(t)} with the two initial states above is
\begin{equation}\label{Kolmdist}
D_K(t_0,t):=\frac 12 \abs{\Pbb_0(t_0,t|\rho^0)- \Pbb_0(t_0,t|\rho^1)}+\frac 12\abs{\Pbb_1(t_0,t|\rho^0)- \Pbb_1(t_0,t|\rho^1)}.
\end{equation}
Let us note that
for every initial state we have the three probabilities $\Pbb_i(t_0,t|\rho_0)$, $i=0,1$, and $\Pbb[N(t)-N(t_0)=0]$, but the last one does not depend on the initial state by construction and, so, it does not contribute to the Kolmogorov distance. To have a term of comparison with other approaches \cite{BreuerLP09,Vacc+11}, we compute also the trace distance between the mean states:
\begin{equation}\label{Trdist}
D\big(\eta_0(t_0),\eta_1(t_0)\big)=\frac 12\Tr\left\{\abs{\eta_0(t_0)- \eta_1(t_0)}\right\}, \qquad \eta_i(t_0):= \eta(t_0)\big|_{\rho_0=\rho^i}, \quad i=0,1.
\end{equation}

An analytically computable example showing many revivals is the case of an Erlang distribution with shape parameter $k=2$, as in \cite{Vacc+11}. Moreover, we compare it with the exponential case; the exponential and the Erlang distributions are in the class of the Gamma distributions. The results are collected in the two Remarks here below and the computations are given in the Appendix \ref{App:comput}.

In the following  the initial state $\rho_0$ is parameterised in the form: 
\begin{equation}\label{matrixrho0}
\rho_0=\begin{pmatrix} \rho_{11} &\rho_{10} \\ \rho_{01} &\rho_{00}\end{pmatrix}, \qquad \rho_{ii}\geq 0, \qquad \rho_{00} + \rho_{11}=1, \qquad \rho_{10}=\overline{\rho_{01}}, \qquad \rho_{00}\rho_{11}\geq \abs{\rho_{01}}^2.
\end{equation}

\begin{remark}[Exponential]\label{rem:exp} For the exponential distribution (or $\Gamma(1,\lambda)$) we have 
\[
f(t)=\lambda \rme^{-\lambda t}, \qquad 1-F(t)=\rme^{-\lambda t}, \qquad \lambda>0, \qquad t\geq 0. 
\]
Then, for the model defined by equations \eqref{Ecalpmi}, the mean state \eqref{eta-t0} and the probabilities \eqref{P_iN(t)}, \eqref{Ngeq1} take the expressions
\begin{equation}\label{etat0exp}
\eta(t_0)=\rme^{-\lambda t_0} \begin{pmatrix} 0 &\rho_{10} \\ \rho_{01} & 0\end{pmatrix} +\frac 12 \left(1+\rme^{-2\lambda t_0}\right)\begin{pmatrix} \rho_{11} & 0 \\ 0 &\rho_{00}\end{pmatrix} +\frac 12 \left(1-\rme^{-2\lambda t_0}\right)\begin{pmatrix} \rho_{00} & 0 \\ 0 &\rho_{11}\end{pmatrix},
\end{equation}
\begin{equation}\label{Piexp}
\Pbb_i(t_0,t|\rho_0)=\left(1- \rme^{-\lambda \left(t-t_0\right)}\right)\frac 12 
\left\{\left( 1+\rme^{-2\lambda t_0}\right) \rho_{ii}
+\left(1-\rme^{-2\lambda t_0}\right)  \left(1-\rho_{ii}\right)\right\},
\end{equation}
\begin{equation}\label{N(t)exp}
\Pbb[N(t)-N(t_0)\geq 1]=1-\rme^{-\lambda (t-t_0)}=F(t-t_0).
\end{equation}

By taking the initial states \eqref{rho12},  the  probabilities \eqref{Piexp} can be written as
\begin{equation}
\Pbb_i(t_0,t|\rho_j)=\left(1- \rme^{-\lambda \left(t-t_0\right)}\right)\frac{1+\left(2\delta_{ij}-1\right)\rme^{-2\lambda t_0}}2,\end{equation}
and the Kolmogorov distance \eqref{Kolmdist} between them becomes
\begin{equation}\label{Koldisexp}
D_K(t_0,t)=\left(1- \rme^{-\lambda \left(t-t_0\right)}\right)\rme^{-2\lambda t_0}.
\end{equation}
By computing the trace distance \eqref{Trdist}, it is easy to check that we have
\begin{subequations}
\begin{equation}\label{Kolexp+trdis}
D_K(t_0,t)=\Pbb[N(t)-N(t_0)\geq 1] \times D\big(\eta_0(t_0),\eta_1(t_0)\big), \end{equation} 
\begin{equation}\label{DKlarget} \lim_{t\to +\infty}D_K(t_0,t)= D\big(\eta_0(t_0),\eta_1(t_0)\big).
\end{equation}
\end{subequations}
\end{remark}

\begin{remark}[Erlang]\label{rem:Erl} For the Erlang distribution with shape parameter $k=2$ (or $\Gamma(2,\lambda)$) we have 
\[
f(t)=\lambda^2 t \rme^{-\lambda t}, \qquad 1-F(t)=(1+\lambda t)\rme^{-\lambda t},  \qquad \lambda>0, \qquad t\geq 0. 
\]
Then, for the model defined by equations \eqref{Ecalpmi}, the mean state \eqref{eta-t0} and the probabilities \eqref{P_iN(t)}, \eqref{Ngeq1} take the expressions
\begin{multline}\label{etat0Erl}
\eta(t_0)= (1+\lambda t_0)\rme^{-\lambda t_0}\begin{pmatrix} 0 &\rho_{10} \\ \rho_{01} & 0\end{pmatrix} 
+\frac 12 \left(1 +\rme^{-\lambda t_0}\bigl(\cos(\lambda t_0) + \sin(\lambda t_0)\bigr)\right) \begin{pmatrix} \rho_{11} & 0 \\ 0 &\rho_{00}\end{pmatrix}
\\ {} +\frac 12 \left(1 -\rme^{-\lambda t_0}\bigl(\cos(\lambda t_0) + \sin(\lambda t_0)\bigr)\right) \begin{pmatrix} \rho_{00} & 0 \\ 0 &\rho_{11}\end{pmatrix},
\end{multline}
\begin{multline}\label{PiErl}
\Pbb_i(t_0,t|\rho_0)=  \frac 12\left(1-\rme^{-\lambda(t-t_0)}\right) \left( 1+{\rme^{-\lambda t_0}}\left(\cos \lambda t_0+\sin \lambda t_0 \right) \left(2\rho_{ii}-1\right)\right)
\\ {}  -\frac{\lambda(t-t_0)} 2\rme^{-\lambda t} \left(\cosh \lambda t_0 +\left(2\rho_{ii}-1\right)\cos \lambda t_0 \right)  ,
\end{multline}
\begin{multline}\label{N(t)Erl}
\Pbb[N(t)-N(t_0)\geq 1]= 1-\rme^{-\lambda (t-t_0)}- \frac \lambda 2 \left(t-t_0\right)\rme^{-\lambda (t-t_0)}\left(1+ \rme^{-2\lambda t_0}\right) 
\\ {}=F(t-t_0) + \lambda \left(t-t_0\right) \rme^{-\lambda t} \sinh\lambda t_0.
\end{multline}
By taking the initial states \eqref{rho12},  the probabilities \eqref{PiErl} can be written as
\begin{multline}\label{2PiErl}
\Pbb_i(t_0,t|\rho_j)=\frac 12 \,\Pbb[N(t)-N(t_0)\geq 1] \\ {}+\frac {2\delta_{ij} -1}2\Bigl\{\left(1-\rme^{-\lambda(t-t_0)}\right) \rme^{-\lambda t_0}\left(\cos \lambda t_0+\sin \lambda t_0 \right)- \lambda (t-t_0)\rme^{-\lambda t} \cos \lambda t_0\Bigr\} ,
\end{multline}
and the Kolmogorov distance \eqref{Kolmdist} between them becomes
\begin{equation}\label{DKErl}
D_K(t_0,t)=\abs{\left(1-\rme^{-\lambda(t-t_0)}\right) \rme^{-\lambda t_0}\left(\cos \lambda t_0+\sin \lambda t_0 \right)- \lambda (t-t_0)\rme^{-\lambda t} \cos \lambda t_0} .
\end{equation}
Finally, the trace distance \eqref{Trdist} is given by
\begin{equation}\label{DtraceErl}
D\big(\eta_0(t_0),\eta_1(t_0)\big)=\rme^{-\lambda t_0}\abs{\cos \lambda t_0+\sin \lambda t_0 }.
\end{equation}
Note that we have
\begin{equation}\label{limDtraceErl}
\lim_{t\to +\infty}D_K(t_0,t)=D\big(\eta_0(t_0),\eta_1(t_0)\big).
\end{equation}
\end{remark}

The model based on the Erlang distribution has been introduced as a prototype of a system with non-Markovian dynamics showing infinitely many revivals in the trace distance \cite{Vacc+11,Vacc+21}.  By integrating the trace distance over the times of increasing and by optimizing over the initial states an effective index of non-Markovianity has been constructed \cite{BreuerLP09,Vacc+11}. We have reported the trace distance for the mean states in eq.\ \eqref{DtraceErl}. However, in our case, we have not only the mean state, but also the fact that the type of jumps are observed. 

From the results of Remark \ref{rem:Erl} we see that the revivals affect also the classical probabilities of the jumps \eqref{PiErl} and can be highlighted by using the Kolmogorov distance \eqref{DKErl}. As a test we have chosen a very simple event \eqref{modifevent}; however, it is possible to consider more involved situations, as the event \eqref{Ettau}, or even the probabilities for the typical trajectories. The simple choice \eqref{modifevent} is enough to see that the memory manifests itself not only in the revivals, but also in the fact that the Kolmogorov distance \eqref{DKErl} does not factorize in a term involving only the mean state at time $t_0$ and a term involving the observation which is registered in the time interval $(t_0,t)$. On the contrary, the exponential example of Remark \ref{rem:exp} does not present revivals, as expected, and shows lack of memory also in the factorization of the Kolmogorov distance \eqref{Koldisexp}, \eqref{Kolexp+trdis}. In both cases, for large $t$, the Kolmogorov distance reproduces the trace distance for the mean states at time $t_0$ \eqref{DKlarget}, \eqref{limDtraceErl}, but this is due to the simplicity of the model and to a suitable choice of the considered event.

\section{Quantum/classical random walks}\label{sec:walks} 
Jump type quantum trajectories can be related to random walks, in particular to ``open quantum walks" \cite{APS12,CarP15,CGMh22,Pell14,CLRB17}. Here we give a class of random walks which are also an example of a Markovian hybrid system: only the full quantum/classical system has a Markovian dynamics, not the quantum component taken alone, nor the classical component; for general Markovian hybrid systems see \cite{Bar24}. This class of models includes the so called ``non-Markovian generalized Lindblad-type master equation" or ``Lindblad rate equation"
\cite{Bud06,Breuer07,Dio14,Dio23,BarP10,SSP17,Pell14,Bud26,BrGM06}, ``continuous time open quantum walks" (CTOQW) \cite{Pell14,CLRB17,Bri18,BBPP19,Loeb23,Loe24a,Loe24b,Kang19} and the related ``continuous time quantum Markov chains" \cite{IglL24}.

\subsection{The classical component}\label{sec:grcl}
We take a discrete measure for $\nu(\rmd u)$, which means a discrete set of independent Poisson processes (under the reference probability $\Qbb$): $N_1(t)$, \ldots, $N_r(t)$ of intensities $\nu_1>0$, \ldots, $\nu_r>0$; eventually, we can have $r=+\infty$.  According to Assumption \ref{ass:Poiss}, the sum of the intensities is finite, even if $r=+\infty$. So, we have
\[
\Uscr =\{1,\ldots,r\}, \quad \nu(\rmd u)\to \nu_1,\ldots,\nu_r, \quad \nu:=\sum_{u=1}^r \nu_u<+\infty, \quad \Pi(\rmd u,\rmd t)\to N_1(t),\ldots, N_r(t).
\]
The counting process \eqref{def:N} now turns out to be 
\begin{equation}\label{Nsum}
N(t)=\sum_{u=1}^r N_u(t);
\end{equation}
under the reference probability $\Qbb$ it is a Poisson process of intensity $\nu$ .
Then, we take a graph $G$ in $\Rbb^s$ with $n$ vertices $x(j)\in \Rbb^s$ ($n$ can be $+\infty$); the set of the vertices is denoted by $V$. We take also  $r$  subsets of $V$ (not necessarily disjoint):
\begin{equation}\label{defF0}
V=\{x(1),\ldots,x(n)\}, \quad x(j)\in \Rbb^s, \qquad F_u\subset V, \qquad F_0:=V\setminus \bigcup_{u=1}^{r} F_u. 
\end{equation}
In \eqref{Xprocess} we take
\begin{equation*}
c(x)=0, \qquad g(x,u)=\begin{cases} y^u(x)-x & \text{if } \ x\in F_u, \\ 0 & \text{if } \ x\notin F_u ,\end{cases}
\qquad y^u(x)\in V.
\end{equation*}
Then, \eqref{Xprocess} becomes
\begin{equation}\label{Diodiscr}
\rmd X(t)= \sum_{u=1}^{r} \big(y^u\big(X(t_-)\big)-X(t_-)\big)\ind_{F_u}\big(X(t_-)\big)\rmd N_u(t),  \qquad X(0)\in V;
\end{equation}
the classical evolution \eqref{Diodiscr} gives rise to a random walk on $V$. 
The choice of the sets $F_u$ and of $y^u(x)$ determines the features of the walk. When $F_0\neq \emptyset$, the random walk of $X(t)$ stops if it enters in $F_0$. When $X(t_-)$ is contained in $F_u$ and there is a count in $N_u(t)$, the walker goes from $X(t_-)$ to $y^u\big(X(t_-)\big)$. When a certain vertex is contained in more sets $F_u$ the walker has more choices, whose probabilities depend on the intensities of the related processes $N_u(t)$.

\subsubsection{The Pauli rate equation}\label{sec:Pauli}

Under the reference law $\Qbb$, the probabilities at time $t$ of the classical process $X(t)$ follow a \emph{Pauli rate equation} \cite{Dio23}, also known as continuous time Markov chain on a graph \cite{BBPP19}.
Let us introduce the probability of being in the position $x(k)$,  under the reference probability $\Qbb$:
\begin{equation}\label{probsqj}
q_k(t)= \Qbb[X(t)=x(k)]=\Ebb_\Qbb\left[\delta_{X(t),x(k)}\right].
\end{equation}
Then,  the probabilities \eqref{probsqj} satisfy a \emph{Pauli rate equation} \cite[(14)]{Dio23}:
\begin{equation}\label{dqj}
\frac{\rmd q_k(t)}{\rmd t}= \sum_{l\in  V }\Bigl(T(k,l)q_l(t)-T(l,k)q_k(t)\Bigr), \qquad T(k,l)= \sum_{u= 1}^r \nu_u\ind_{F_u}\big(x(l)\big) \delta_{y^u(x(l)),x(k)}.
\end{equation}
Moreover, the following properties holds:
\[
T(k,l)=0 \ \text{ if} \ x(l)\in F_0, \qquad \qquad \sum_{l\in  V }T(l,k)=\sum_{u= 1}^r \nu_u\ind_{F_u}\big(x(k)\big) .
\]

\begin{proof}[{Computations.}] By differentiating the Kronecker delta in \eqref{probsqj} we get
\begin{equation}\label{ddelta}
\rmd\delta_{X(t),x(k)}= \sum_{u= 1}^r \ind_{F_u}\big(X(t_-)\big)\rmd N_u(t)\left(\delta_{y^u(X(t_-)),x(k)}-\delta_{X(t_-),x(k)}\right).
\end{equation}
Then, by taking the $\Qbb$-mean we have
\begin{multline*}
\frac{\rmd q_k(t)}{\rmd t}= \sum_{u= 1}^r\nu_u\Ebb_\Qbb\left[\ind_{F_u}\big(X(t_-)\big) \delta_{y^u(X(t_-)),x(k)}-\ind_{F_u}\big(x(k)\big)\delta_{X(t_-),x(k)}\right]
\\ {}= \sum_{u= 1}^r\nu_u\Ebb_\Qbb\left[\sum_i\ind_{F_u}\big(x(i)\big) \delta_{y^u(x(i)),x(k)}\delta_{x(i),X(t_-)}-\ind_{F_u}\big(x(k)\big)\delta_{X(t_-),x(k)}\right]
\\ {}= \sum_{u= 1}^r\nu_u\left[\sum_i\ind_{F_u}\big(x(i)\big) \delta_{y^u(x(i)),x(k)}q_i(t)-\ind_{F_u}\big(x(k)\big)q_k(t)\right].
\end{multline*}
This gives \eqref{dqj}; the properties of the quantities $T(k,l)$ follow by direct check.
\end{proof}

We can say that we have a directed edge from $l$ to $k$ of the graph $G $ when $T(l,k)\neq 0$. Note that the correspondence from \eqref{dqj} to \eqref{Diodiscr} is not unique.

\subsection{The quantum component}
Consistently with the assumption of discrete counting processes, the linear SME \eqref{linearSME++} becomes
\begin{equation}\label{kllinSME}
\rmd  \sigma(t)= \Lcal_0(t)[\sigma(t_-)]  \rmd t-\frac 12 \left\{R(t)-\nu\openone,\sigma(t_-)\right\}\rmd t
+ \sum_{u=1}^r\Bigl( \Jcal(u,t)[ \sigma(t_-) ]-\sigma(t_-) \Bigr) \rmd N_u(t).
\end{equation}
We take all the operators involved in \eqref{kllinSME} dependent on time only through the last value of the classical component \eqref{Diodiscr}. Moreover, to exclude the case of a jump in the quantum state without a jump in the classical component, we introduce the factor $\ind_{F_u}\big(X(t_-)\big)$ also in the SMEs. So, equations \eqref{def:R(t)} and \eqref{linearSME+} are taken with the following structure:
\begin{subequations}\label{Xdependence}
\begin{equation}\label{JcalX}
\Jcal(u,t)= \Jcal\big(X(t_-),u\big), \qquad R(t)= R\big(X(t_-)\big) , \qquad \Lcal_0(t)=\Lcal_0\big(X(t_-)\big), \qquad
\end{equation}
\begin{equation}\label{feedb}
\Jcal(x,u)[ \rho ]= \ind_{F_u}(x)\sum_{j=1}^{d_1}J_j(x,u) \rho J_j(x,u)^\dagger,
\end{equation}
\begin{equation}\label{klR}
R(x)=\sum_{u=1}^r  \nu_u\, \ind_{F_u}(x)\sum_{j=1}^{d_1}J_j(x,u)^\dagger J_j(x,u) , \qquad 
R_0(x)= \sum_{k=1}^{d_2}L_{k}(x)^\dagger L_{k}(x), 
\end{equation}
\begin{equation}\label{Lcal0discr}
\Lcal_0(x)[\rho]=-\rmi[H(x),\rho]-\frac 12 \left\{R_0(x),\rho\right\}+\sum_{k=1}^{d_2} L_{k}(x) \rho L_{k}(x)^\dagger, 
\end{equation}
\begin{equation}\label{Lcaldiscr}
\Lcal(t)= \Lcal\big(X(t_-)\big),  \qquad \Lcal(x)[\rho] =
\Lcal_0(x) [\rho]-\frac 12 \left\{R(x),a\right\}
+\sum_{u=1}^r \nu_u\Jcal(x,u)[ \rho],
\end{equation}
\end{subequations}
where $H(x), \; L_k(x),\; J_j(x,u), \;R(x), \; R_0(x) \in \Bscr(\Hscr)$.

\subsubsection{The physical probability and the non linear SME}
By introducing the probability density $p(t)$ of $\Pbb$ with respect to $\Qbb$ as in Proposition \ref{theor:sigmaprop}, we get the analogous of \eqref{def:pt}, \eqref{eq:pt}
\[
p(t)=\Tr\{ \sigma(t)\} , \qquad
\rmd p(t)=p(t_-)\sum_{u=1}^r \left(I(u,t) -1\right) \left(\rmd N_u(t)-\nu_u \rmd t\right),  \]
\begin{equation}\label{klIut}
I(u,t)=\Tr\left\{\Jcal(u,t)\big[\rho(t_-)\big]\right\}= \Tr\biggl\{ \ind_{F_u}\big(X(t_-)\big)\sum_{j=1}^{d_1}J_j\big(X(t_-),u\big)^\dagger J_j\big(X(t_-),u\big)\rho(t_-)\biggr\},
\end{equation}
$\nu_uI(u,t)$ is the stochastic intensity of $N_u(t)$ under the physical probability $\Pbb$; in general it depends on the  quantum state. 
In particular, the factor $\ind_{F_u}(x)$ introduced in \eqref{feedb} implies $I(u,t)=0$ if $X(t_-)\notin F_u$: 
\[
x\notin F_u \quad \Rightarrow \quad \Jcal(x,u)=0 \quad \Rightarrow \quad I_t(u)=0 \ \ \text{if} \ \ X(t_-)\notin F_u;
\]
in other terms, when $X(t_-)\notin F_u$, under the physical probability $\Pbb$, there is no jump of ``type $u$" nor in the quantum state, neither in $X(t)$.

Now, the non-linear SME \eqref{nonlinearSME} becomes
\begin{multline}\label{nlSMEwalk}
\rmd \rho (t)= \Lcal_0(t)\big[\rho(t_-)\big]  \rmd t-\frac 12 \left\{R(t)-\lambda(t)\openone,\,\rho (t_-)\right\} \rmd t
\\ {}+ \sum_{u: I(u,t)\neq 0 }\Bigl( I(u,t)^{-1}\Jcal(u,t)[ \rho(t_-) ]-\rho(t_-) \Bigr) \rmd N_u(t).
\end{multline}
Finally, under the law $\Pbb$ the stochastic intensity \eqref{lambda} of $N(t)$ \eqref{Nsum} is
\begin{equation}\label{lambdakl}
\lambda(t)=\sum_{u=1}^r I(u,t)\nu_u \equiv \Tr\left\{R(t) \rho(t_-)\right\}<+\infty.
\end{equation}

\subsection{The Lindblad rate equation}\label{sec:Linrate}

It is possible to construct an evolution equation for ``hybrid states", analogous to the classical evolution \eqref{dqj}.
We take as initial condition $X_0$ and $\rho_0=\rho_0(X_0)$, where
$X_0$ is an $\Fscr_0$-measurable, $V$-valued random variable and $\rho_0$ is a random statistical operator. By setting
\[
p_0^k:=\Qbb[X_0=x(k)]=\Ebb_\Qbb[\delta_{X_0,x(k)}], \qquad \rho_0^k:=\rho_0\big(x(k)\big),
\]
we can say that the initial condition is $\big(x(k),\; \rho_0^k\big)$ with probability $p_0^k$.

Similarly to \eqref{probsqj}, we define
\begin{equation}\label{def:reddes}
\eta_k(t):=\Ebb_\Qbb\big[\sigma(t)\delta_{X(t),x(k)}\big];
\end{equation}
in particular we have
\[
\eta_k(0)=\Ebb_\Qbb\big[\rho_0\delta_{X_0,x(k)}\big]=\Ebb_\Qbb\big[\rho_0^k\delta_{X_0,x(k)}\big]=p_0^k\rho_0^k.
\]
By the definition of the mean state $\eta(t)$ \eqref{meanst} we have immediately from \eqref{def:reddes}
\begin{equation}
\sum_{k=1}^n \eta_k(t)=\Ebb_\Qbb[\sigma(t)]=\eta(t).
\end{equation}

The vector of statistical operators $\eta_k(t)$ \eqref{def:reddes} satisfies a quantum analog of the Pauli master equation \eqref{dqj}. To simplify the notation we set
\begin{equation}\label{x(k)tok}
\Lcal_0^k=\Lcal_0\big(x(k)\big), \qquad R_k=R\big(x(k)\big), \qquad H_k=H\big(x(k)\big).
\end{equation}

\begin{proposition}
The vector of statistical operators $\eta_k(t)$ satisfies the Lindblad rate equation
\begin{equation}\label{Linre}
\frac {\rmd\ }{\rmd t}\,\eta_k(t)= \Lcal_0^k[\eta_k(t)]
+\sum_{l\in V}\left(\Tcal(k,l)[ \eta_l(t) ]- \frac 12 \bigl\{\Tcal(l,k)^*[\openone],\eta_k(t)\bigr\}\right), \qquad k\in V,
\end{equation}
\begin{equation}\label{tildeL0}
\Tcal(k,l)[\rho]=\sum_{u= 1}^r\nu_u  \delta_{y^u(x(l)),x(k)} \Jcal\big(x(l),u\big)[  \rho];
\end{equation}
$\Lcal_0(x)$ is given by \eqref{Lcal0discr}.
\end{proposition}

Note that we have
\begin{equation}\label{TcaltoR}
\sum_{l\in V}\Tcal(l,k)^*[\openone]=R\big(x(k)\big)\equiv R_k,\qquad \qquad \Tcal(k,l)=0 \ \ \text{if} \ \ x(l)\in F_0.
\end{equation}
As for the Pauli rate equation, there are more choices for the SDEs \eqref{Diodiscr} and \eqref{kllinSME} which give the same Lindblad rate equation \eqref{Linre}.

\begin{proof}
We use equations \eqref{ddelta} and \eqref{kllinSME} to compute the differential of the product $\delta_{X(t),x(k)}\sigma(t) $; by the rules of stochastic calculus, we get
\begin{multline*}
\rmd \left(\sigma(t) \delta_{X(t),x(k)}\right)= \delta_{X(t_-),x(k)}\rmd \sigma(t)+\sigma(t_-)\rmd \delta_{X(t),x(k)}+ \left( \rmd \delta_{X(t),x(k)}\right)\rmd \sigma(t)
\\ {}=\delta_{X(t_-),x(k)}\rmd \sigma(t)+\sum_{u= 1}^r \ind_{F_u}\big(X(t_-)\big)\left(\delta_{y^u(X(t_-)),x(k)}-\delta_{X(t_-),x(k)}\right)  \sum_{u=1}^r \Jcal(u,t)[ \sigma(t_-) ] \rmd N_u(t).
\end{multline*}
By taking the $\Qbb$-mean we have $\rmd N_u(t) \to \nu_u\rmd t$ and we get
\begin{multline*}
\Ebb_\Qbb\big[\rmd\left(\delta_{X(t),x(k)} \sigma(t)\right)\big]=\Lcal\big(x(k)\big)[ \eta_k(t_-)]\rmd  t
\\ {}+\sum_{u= 1}^r \nu_u\sum_{l\in V} \ind_{F_u}\big(x(l)\big)\delta_{y^u(x(l)),x(k)} \Jcal\big(x(l),u\big)[\eta_l(t_-) ] \rmd t
-\sum_{u= 1}^r\nu_u \ind_{F_u}\big(x(k)\big)  \Jcal\big(x(k),u)[\eta_k(t_-) ] \rmd t
\\ {}=\Lcal\big(x(k)\big)[ \eta_k(t_-)]\rmd  t
+\sum_{u= 1}^r \nu_u\sum_{l\in V} \delta_{y^u(x(l)),x(k)} \Jcal\big(x(l),u\big)[\eta_l(t_-) ] \rmd t
-\sum_{u= 1}^r \nu_u \Jcal\big(x(k),u)[\eta_k(t_-) ] \rmd t;
\end{multline*}
the indicator functions disappear because they are already contained in $\Jcal\big(x,u\big)$ \eqref{feedb}. The last term in the expression above cancels the last term in $\Lcal(x)$ \eqref{Lcaldiscr} and we get
\[
\frac {\rmd\ }{\rmd t}\,\eta_k(t)=\Lcal_0\big(x(k)\big)[ \eta_k(t_-)]
+\sum_{u= 1}^r\nu_u \sum_{l\in V} \delta_{y^u(x(l)),x(k)} \Jcal\big(x(l),u\big)[ \eta_l(t_-) ]
-\frac 12 \left\{R\big(x(k)\big),\eta_k(t)\right\}.
\]
By using \eqref{klR} we get the final result.
\end{proof}

The solution of \eqref{Linre} can be written in a way similar to the Dyson expansion \eqref{vareta}:
\begin{subequations}\label{etakD}
\begin{equation}
\eta_k(t)= \sum_{m=0}^{+\infty}\Dcal_m^k(t)[\vec \eta(0)], \qquad \vec \eta(0)= \{\rho^l_0 p^l_0,\; l=1,\ldots,n\}, 
\end{equation}
\begin{equation*}
\begin{split}
&\Dcal_0^k(t)[\vec \eta(0)]=\Scal_k(t)[\rho_0^k]p_0^k,
\\ 
&\Dcal_1^k(t)[\vec \eta(0)]=\sum_{l_1=1}^n\int_0^t \rmd t_1\Scal_k(t-t_1)\circ \Tcal(k,l_1)\circ \Scal_{l_1}(t_1)[\rho_0^{l_1}]p_0^{l_1},
\\ 
&\Dcal_2^k(t)[\vec \eta(0)]=\sum_{l_1,l_2=1}^n\int_0^t \rmd t_2\int_0^{t_2} \rmd t_1\Scal_k(t-t_2)\circ \Tcal(k,l_2)\circ \Scal_{l_2}(t_2-t_1)\circ \Tcal(l_2,l_1) \circ \Scal_{l_1}(t_1)[\rho_0^{l_1}]p_0^{l_1},
\\
& \Dcal_m^k(t)[\vec \eta(0)]=\cdots, \end{split}
\end{equation*}
\begin{equation}\label{Scalk}
\Scal_k(t)=\exp\left\{\left(\Lcal_0^k -\frac 12 \left\{R_k, \bullet\right\}\right)t\right\}.
\end{equation}
\end{subequations}

Master equations with the structure \eqref{Linre} have been introduced and developed in
\cite{Bud06,Breuer07,BrGM06} under the names of 
``non-Markovian generalized Lindblad-type master equations" or ``Lindblad rate equations". The original aim was to give non-Markovian evolutions for quantum systems, mainly tailored for the case of structured baths. The first unravellings of these equations were developed in \cite{BarP10,Pell14,SSP17,Bri18}. While the dynamics of the quantum system alone is non-Markovian, the construction we followed in this section shows that the composed classical/quantum system follows a dynamics without memory. The dynamics of the type of \eqref{Linre} was explicitly recognized as an hybrid dynamics \cite{Dio14,Dio23,Bud26,BarP10}: the discrete index in \eqref{Linre} behaves essentially in a  classical way.

Also the CTOQWs are included in the models of this section. The notion of ``open quantum walks" is used to mean a walk on a graph, but with a  vector in a Hilbert space (the ``position space")  associated to every vertex \cite{APS12}: $x(i) \leftrightarrow |i\rangle\langle i|$. However, the walk is essentially classical and in the continuous time version it has been connected to the Lindblad rate equation \cite{Pell14} and to a walk on graph as in  \eqref{Diodiscr} \cite{BBPP19}. 
Many properties of CTOQWs have been studied, such as steady state \cite{CLRB17}, recurrence and transience \cite{BBPP19,Loeb23}, large deviations \cite{CGMh22,Bri18}.

\subsection{Typical trajectories}\label{sec:jumps}
Once again, we can give the recursive structure of the solutions of the evolution equations along the typical trajectories \eqref{trajec}. Let us start by the classical random walk \eqref{Diodiscr}, which is constant in time up to the first jump and between two jumps. At the jump $(u_j,t_j)$ we have 
\begin{equation}\label{Xjump}
X(t_{j-1})\quad \to \quad X(t_j)=\begin{cases} y^{u_j}\big( X(t_{j-1})\big) \quad &\text{if} \ X(t_{j-1})\in F_{u_j}
\\ X(t_{j-1}) & \text{if} \ X(t_{j-1})\notin F_{u_j},\end{cases} \qquad t_0=0,
\end{equation}
where we have taken into account that $X({t_j}_-)= X(t_{j-1})$ by the absence of evolution between two jumps.

Regarding the SMEs, recall that all the operators involved are independent of time in between two jumps, while at a jump all the operators can change together with $X(t_j)$ according to the dependence \eqref{Xdependence}. 

Firstly, we introduce the no jump propagator as in Definition \ref{def:Stt0}: in between two jumps
\begin{equation} \label{nojpropakl}
\Scal(t,t_{0})=\exp\left\{ (t-t_0)\Acal(t_{0})\right\}, \qquad \Acal(t_{0})[\rho]=\Lcal_0\big(X(t_0)\big)[\rho]-\frac 12 \left\{ R\big(X(t_0)\big),\,\rho\right\}.
\end{equation}
Then, the solution of the linear SME \eqref{kllinSME} has again the structure \eqref{decompsigma} with $\nu(\Uscr)\equiv \nu$; if the jump $(u_j,t_j)$ is such that $X(t_{j-1})\notin F_{u_j}$ the solution goes to zero because of the indicator function inside the jump operator \eqref{feedb}.
Analogously, the solution of the non-linear SME \eqref{nlSMEwalk} maintains the expression \eqref{decomprho}; the trajectories for which a term $I(u_j,t_j)$ vanishes have to be eliminated, because they have probability zero.  Then, we have
\begin{equation}\label{rhoklNOJ}
\rho(t) =\frac{\Scal(t,t_j)[\rho(t_j)]}{\Tr\{\Scal(t,t_j)[\rho(t_j)]\}}
\end{equation}
in between two jumps, while  at a jump (with $X(t_{j-1})\in F_{u_j}$) we have the transfomation
\begin{equation}\label{rhojump}
\rho({t_j}_-) \quad \to \quad \rho(t_j)= I(u_j,t_j)^{-1}\sum_{i=1}^{d_1}J_i\big(X(t_{j-1}),u_j\big)\rho({t_j}_-) J_i\big(X(t_{j-1}),u_j\big)^\dagger;
\end{equation}
$I(u_j,t_j)$ is the normalization factor and it coincides with the quantity \eqref{klIut}.
Finally, the exclusive probability densities have the structure \eqref{genexclpd}.

Summarizing, we can say that the jumps happen at random times determined by the stochastic intensities $\nu_uI(u,t)$ \eqref{klIut}. The process $N(t)$ \eqref{Nsum}, of intensity $\lambda(t)$ \eqref{lambdakl}, determines the times when a random jump happens; then, the jump is of type $u$ with probability $I(u,t)\nu_u/\lambda(t)$. In between two jumps $X(t)$ stays constant, while $\rho (t)$ evolves according to equation \eqref{rhoklNOJ}. At a jump, $X(t)$ changes according to equation \eqref{Xjump} and $\rho(t)$ according to equation \eqref{rhojump}.

The whole construction gives explicitly the hybrid dynamics as a classical random walk on a graph \eqref{Diodiscr} associated with a stochastic evolution of the quantum state \eqref{nlSMEwalk}. The dynamics of the classical component influences the dynamical behaviour of the quantum component by feedback \eqref{feedb} and the quantum dynamics determines the intensities \eqref{klIut} of the counting processes involved. The dynamics is without memory for the composed system: if $\big(X(t_0), \rho(t_0)\big)$ is known, then the process $\big(X(t), \rho(t)\big)$, $t\geq t_0$, is detemined by the SDEs \eqref{Diodiscr}, \eqref{nlSMEwalk}.

\subsubsection{Waiting times}\label{sec:WTwalks}
Due to the Markov property of the hybrid prcess, to have the no jump probability \eqref{N0|} we do not need the whole past, but only $\big(X(t_0), \rho(t_0)\big)$:
\begin{equation}\label{nojprobkl}
\Pbb[N(t)-N(t_0)=0|\Fscr_{t_0}]=  \Tr\left\{\Scal(t,t_0)[ \rho(t_0)]\right\} =\Tr\left\{\exp\left\{ (t-t_0)\Acal(t_{0})\right\}[ \rho(t_0)]\right\} ;
\end{equation}
$\Acal(t_{0})$ is defined in \eqref{nojpropakl}. The distribution of the waiting time of the next jump can be introduced as in \eqref{T=N(t)}.

By the introduction of the set $F_0$ in \eqref{defF0}, we have opened to the possibility of stopping the jumps: if $X(t_0)\in F_0$, then, by \eqref{Diodiscr}, $X(t)=X(t_0)$, $\forall t\geq t_0$. From \eqref{Xdependence} we get $R\big(X(t)\big)=0$, $t\geq t_0$,  and, so, also the intensity \eqref{lambdakl} of  $N(t)$ vanishes. Moreover, 
the no jump propagator becomes trace preserving  by \eqref{nojpropakl} and, from \eqref{nojprobkl}, we have $\Pbb[N(t)-N(t_0)=0|\Fscr_{t_0}]=1$, $\forall t\geq t_0$.

In the class of processes of Sec.\ \ref{sec:inters} the distribution of the waiting times was fixed by the classical component through eqs.\ \eqref{N(t)intersp}--\eqref{lambda+f+F}. On the contrary, in the case of the processes of this section, the permanence probabilities \eqref{nojprobkl} of the walker in a given vertex are determined only by the evolution \eqref{nojpropakl} of the quantum component. As the generator $\Acal(t_0)$ is constant in between two jumps, the survival probability \eqref{nojprobkl} is given by a trace decreasing evolution;  when $\Lcal_0$ has no dissipative contribution, the waiting time distribution is determined by a non-Hermitian effective Hamiltonian as in the case of Sec.\ \ref{sec:nonH}.

\subsection{Examples for a two-dimensional quantum component}\label{sec:2dex}
In the case of a two-dimensional quantum system, different examples of the evolution \eqref{Linre} have been developed and possible derivations from system/environment interactions have been discussed \cite{Breuer07,BrGM06,Bud06}. Here we give a very simple walk model for one of that examples.

We consider a single counting process $N(t)$; under the reference probability $\Qbb$ it is a Poisson process of intensity $\nu>0$. According to the notation of Sec.\ \ref{sec:grcl}, this means to take $r=1$, \ $\Uscr=\{1\}$.

For the classical component $X(t)$ we assume to have a graph with only two vertices:  \ $V=\{x(1),x(0)\}$. \ Moreover, we assume that the process $X(t)$ can go only from $x(1)$ to $x(0)$ and viceversa; this means to take $y(x)=\delta_{x,x(1)}x(0) + \delta_{x,x(0)}x(1)$ in \eqref{Diodiscr}.
It is easy to check that in the associated Pauli rate equation \eqref{dqj} we have $T(1,0)=T(0,1)=\nu$, \quad $T(1,1)=T(0,0)=0$.

For what concerns the quantum component, in \eqref{Xdependence} we take
\begin{equation}\label{Jcalexample}
\Jcal(x(0))[\rho]=g_0\Ecal_0[\rho]= g_0\sigma_+ \rho\sigma_-, \qquad \Jcal\big(x(1)\big)[\rho]= g_1\Ecal_1[\rho] = g_1\sigma_- \rho\sigma_+ , \qquad g_k>0,
\end{equation}
\begin{equation}\label{HkLk}
H_k\equiv H\big(x(k)\big)=H\big(x(k)\big)^\dagger, \qquad \Lcal_0^k\equiv \Lcal_0\big(x(k)\big)=-\rmi [H_k,\bullet].
\end{equation}
The operators $\Ecal_k$ have been introduced also in Sec.\ \ref{sec:revivals}, Eq.\ \eqref{Epm}; we need also the  related projections:
\begin{equation}\label{projections}
P_1= \sigma_+\sigma_-=\Ecal_1^*[\openone], \qquad P_0=\sigma_-\sigma_+=\Ecal_0^*[\openone].
\end{equation}
Then, from \eqref{tildeL0} and \eqref{TcaltoR} we get
\[
\Tcal(k,k)=0, \qquad \Tcal(1,0)=\nu \Jcal(x(0))=\nu_0\Ecal_0, \qquad \Tcal(0,1)=\nu \Jcal\big(x(1)\big)=\nu_1 \Ecal_1, \qquad  \nu_k:=\nu g_k ,
\]
\[
R_1\equiv R\big(x(1)\big)=\nu_1P_1, \qquad R_0\equiv R\big(x(0)\big)=\nu_0P_0.
\]

By setting 
\begin{equation}\label{Acalk}
\Acal_k[\rho]=-\rmi[H_k,\rho] -\frac 12 \{R_k,\rho\}=-\rmi H_{\rm eff}^{(k)}\rho +\rmi \rho H_{\rm eff}^{(k)\dagger}, \qquad H_{\rm eff}^{(k)}:=H_k -\frac \rmi 2 \,R_k,
\end{equation}
we get the Lindblad rate equation \eqref{Linre} associated to the present model:
\begin{equation}\label{Lre2}
\frac {\rmd\ }{\rmd t}\,\eta_1(t)= \Acal_1[\eta_1(t)]
+\nu_0\Ecal_0[ \eta_0(t) ], \qquad \frac {\rmd\ }{\rmd t}\,\eta_0(t)= \Acal_0[\eta_0(t)]
+\nu_1\Ecal_1[ \eta_1(t) ].
\end{equation}
In the case $H_k\propto \sigma_z$, this couple of equations has been proposed and studied in \cite{BrGM06,Breuer07} as a prototype of non-Markovian evolution.

\subsubsection{Typical trajectories and waiting times}\label{sec:ttr+wt}
Let us consider the typical trajectories of Secs.\ \ref{sec:ttrajs} and \ref{sec:jumps}, here restrictyed to the particular case of a single counting process $N(t)$. At a count $\rmd N(t)=1$, according to the choice made at the beginning of Sec.\ \ref{sec:2dex}, the classical walk jumps in the following way:
\begin{subequations}\label{atajump2}
\begin{equation}
X(t_-) \mapsto X(t)=\begin{cases} x(1) \ & \text{if } X(t_-)= x(0),
\\ x(0) \ & \text{if } X(t_-)= x(1).\end{cases}
\end{equation}
Moreover, by \eqref{rhojump}, \eqref{Jcalexample}, \eqref{projections}, the behaviour of the conditional state at a jump is given by
\begin{equation}
\rho({t}_-) \quad \mapsto \quad \rho(t)= \begin{cases} \frac{\Ecal_0[\rho(t_-)]}{\Tr\{P_0\rho(t_-)\}} =P_1 \ & \text{if } X(t_-)= x(0),
\\  \frac{\Ecal_1[\rho(t_-)]}{\Tr\{P_1\rho(t_-)\}}=P_0 \ & \text{if } X(t_-)= x(1).\end{cases}
\end{equation}
\end{subequations}

Instead, if at time $t_0$ the state of the hybrid system is $\big( X(t_0),\,\rho(t_0)\big)$, with $X(t_0)=x(k)$, and there is no jump up to time $t$, \eqref{nojpropakl}, \eqref{rhoklNOJ}, \eqref{Acalk}, we have
\begin{equation}\label{njevol2}
\rho(t)=\frac{\Scal_k(t-t_0)[\rho(t_0)]}{\Tr\{\Scal_k(t-t_0)[\rho(t_0)]\}},
\qquad X(t)=x(k),
\end{equation}
\begin{equation}\label{nojpropag2}
\Scal_k(t-t_0):=\exp\left\{\Acal_k(t-t_0)\right\}.
\end{equation}

Now, the stochastic intensity \eqref{lambdakl} of $N(t)$ becomes
\begin{equation}\label{lambdak2}
\lambda(t)\equiv \lambda_k(t)=\nu\Tr\left\{\Jcal(x(k))\big[\rho(t_-)\big]\right\}=\nu_k\Tr\left\{P_k\rho(t_-)\right\}, \qquad \text{if} \ X(t_-)=x(k).
\end{equation}
Finally, the no jump probability \eqref{nojprobkl} is given by
\begin{equation}\label{wtime2}
\Pbb[N(t)-N(t_0)=0|\rho(t_0), x(k)]=  \Tr\left\{\Scal_k(t-t_0)[ \rho(t_0)]\right\} .
\end{equation}

\subsubsection{Dependence on the Hamiltonian contribution}\label{sec:Hdepend}
For the model introduced in this section, the behaviour at a jump is completely fixed \eqref{atajump2}, but the waiting time distribution \eqref{wtime2} and the quantum evolution \eqref{njevol2} depend on the no-jump propagator \eqref{nojpropag2} and, so, they depend on the choice of the Hamiltonian \eqref{HkLk}. As in Sec.\ \ref{sec:C2ex}, we have again a non-Hermitian evolution in $\Cbb^2$ and the choice of the Hamiltonian \eqref{HkLk} has a strong influence on the behaviour of the system.

\paragraph{(a) $H_k\propto \sigma_z$.} The first choice is to take the Hamiltonian proportional to $\sigma_z$:
\begin{equation}
H_k=\frac{\omega_k} 2 \, \sigma_z, \qquad \omega_k\in \Rbb, \quad k=0,\,1.
\end{equation}
Then, the selfadjoint part of the effective Hamiltonian \eqref{Acalk} commutes with the anti-selfadjoint contribution, and the no-jump propagator \eqref{nojpropag2} has the expression
\[
\Scal_0(t)[\rho]=\left(\rme^{-\rmi \omega_0 t/2}P_1+\rme^{(\rmi \omega_0-\nu_0) t/2}P_0\right)\rho\left(\rme^{\rmi \omega_0 t/2}P_1+\rme^{(-\rmi \omega_0-\nu_0) t/2}P_0\right),
\]
\[
\Scal_1(t)[\rho]=\left(\rme^{(-\rmi \omega_1-\nu_1) t/2}P_1+\rme^{\rmi \omega_1t/2}P_0\right)\rho\left(\rme^{(\rmi \omega_1-\nu_1) t/2}P_1+\rme^{-\rmi \omega_1 t/2}P_0\right).
\]
So, the no-jump probability \eqref{wtime2} turns out to be
\[
\Pbb[N(t)-N(t_0)=0|\rho(t_0), x(k)]=\Tr\{P_k\rho(t_0)\}\rme^{-\nu_kt} +\Tr\{(\openone -P_k)\rho(t_0)\}.
\]
This is a second example in which the waiting time of the next jump can be infinity with a non vanishing probability, as discussed after eqs.\ \eqref{Twtj}. Note that the initial state $X(0)=x(0)$ and $\rho(0)=P_1$ turns out to be a stationary state, because it is not possible to have jumps; similarly, also $X(0)=x(1)$ and $\rho(0)=P_0$ is a stationary state. 
Analogously, the Lindblad rate equations \eqref{Lre2} have not a unique stationary state, as was found in \cite{BrGM06,Breuer07}.

\paragraph{(b) $H_k\propto \sigma_x$.}  The second choice for the Hamiltonian is to take it proportional to a spin component orthogonal to $\sigma_z$, say
\begin{equation}
H_k=\frac{\omega_k} 2 \, \sigma_x \qquad \omega_k\in \Rbb, \quad \omega_k\neq 0, \quad k=0,\,1.
\end{equation}
Then, the effective n-H Hamiltonian defined in \eqref{Acalk} becomes
\begin{equation}\label{HeffK2}
H_{\rm eff}^{(k)}=-\frac{\rmi \nu_k}4\,\openone +K_k, \qquad K_k=\frac{\omega_k}2\,\sigma_x-\frac{\rmi \tilde \nu_k}4\,\sigma_z, \qquad \tilde \nu_1=\nu_1, \quad \tilde \nu_0=-\nu_0.
\end{equation}
The evolution generated by this effective Hamiltonian  can be written as
\[
\rme^{-\rmi H_{\rm eff}^{(k)}t}=\rme^{-\nu_k t/4} \rme^{-\rmi K_kt}
\]
and we can get an explicit expression for it by the techniques developed in Sec.\ \ref{sec:C2ex}. Firstly, we have
\begin{equation}
K_k^2=\varkappa_k\openone, \qquad \varkappa_k=\frac 1 4 \left(\omega_k^2-\frac{\nu_k^2}4\right).
\end{equation}

\begin{remark}\label{remS6}
According to the value of $\varkappa_k$ we have three cases:
\begin{enumerate}
\item $\abs{\omega_k}> \nu_k/2$: $K_k$ has two real eigenvalues \ $\pm z_k$, \ $z_k=\sqrt{\varkappa_k}>0$, and from \eqref{evolgpm} we have
\begin{equation}\label{K1)}
\rme^{-\rmi K_kt}= \openone \cos(z_k t)-\frac{\rmi K_k}{z_k}\,\sin(z_k t);
\end{equation}
\item $\abs{\omega_k}= \nu_k/2$: there is an EP in $K_k$ as in Sec.\ \ref{sec:EP} and from \eqref{EPevolution} we have
\[
\rme^{-\rmi K_kt}= \openone -\rmi K_k t;
\]
\item $0< \abs{\omega_k}< \nu_k/2$: $K_k$ has two imaginary eigenvalues \ $\pm \rmi a_k$, \ $0<a_k=\sqrt{-\varkappa_k}< \nu_k/4$, and from \eqref{evolgpm} we have
\[
\rme^{-\rmi K_kt}= \openone \cosh(a_k t)-\frac{\rmi K_k}{a_k}\,\sinh(a_k t).
\]
\end{enumerate}
\end{remark}

In all these three cases we have exponential decay and, so, 
\[
\lim_{t\to +\infty}\Pbb[N(t)-N(t_0)=0|\rho(t_0), x(k)]=0,
\]
which means that for every choice of the initial state $(X(0),\,\rho_0)$ we have infinitely many jumps with probability 1.
Having the explicit expression of the evolution generated by $H_{\rm eff}^{(k)}$ it is immediate to compute the propagator $\Scal_k(t)$ \eqref{nojpropakl}, the waiting time distribution \eqref{Twtj}, \eqref{nojprobkl}, the intensity $\lambda(t)$ \eqref{lambdakl}, \eqref{lambdak2} of the counting process, the evolution \eqref{njevol2} of the conditional state in between two jumps. A large variety of behaviours can be obtained, depending on the choice of $H_k$. 

Let assume that at time  $t_0$ there is a jump of the type $x(0)\mapsto x(1)$;  by \eqref{atajump2}, the state at time $t_0$, after the jump is $X(t_0)=x(1)$ and $ \rho(t_0)=P_1$. Before the next jump, the evolution is determined by the no-jump propagator $\Scal_1(t)$ \eqref{nojpropakl}: 
\begin{equation}\label{S1psi}
\Scal_1(s)[P_1]=\rme^{-\nu_1s/2}|\psi(s)\rangle \langle \psi(s)|,\qquad \psi(s)=\rme^{-\rmi K_1s}\begin{pmatrix} 1 \\ 0\end{pmatrix}.
\end{equation}
Then, the conditional state $\rho(t)$ \eqref{njevol2} up to the next jump, the intensity of counts $\lambda_1(t)$ \eqref{lambdak2} and the no-jump probability \eqref{wtime2} are  given by
\begin{subequations}\label{proptonorm}
\begin{equation}
\rho(t)=\frac{|\psi(t-t_0)\rangle \langle \psi(t-t_0)|}{\norm{\psi(t-t_0)}^2} , 
\qquad \lambda_1(t)=   \nu_1\Tr\{P_1\rho(t)\}=\nu_1\, \frac{\langle \psi(t-t_0)|P_1\psi(t-t_0)\rangle}{\norm{\psi(t-t_0)}^2},
\end{equation}
\begin{equation}
\Pbb[N(t)-N(t_0)=0|P_1,\, x(1)]=\rme^{-\nu_1(t-t_0)/2}\norm{\psi(t-t_0)}^2.
\end{equation}
We can give also the probability distribution of the waiting time $T_1$ of the next jump; by taking eqs.\ \eqref{Twtj} with initial state $P_1$ we get, for $s\geq 0$,
\begin{equation}\label{probdens2}
\Pbb[T_1\leq s]=1-\rme^{-\nu_1s/2}\norm{\psi(s)}^2, \qquad p_{T_1}(s)=\frac{\rmd \ }{\rmd s}\,\Pbb[T_1\leq s]=\nu_1 \rme^{-\nu_1s/2}\langle \psi(s)|P_1\psi(s)\rangle.
\end{equation}
\end{subequations}

The expressions above can be immediately computed in all the three cases of Remark \ref{remS6}. To have a significant example, showing revivals similar to the case of Remark \ref{rem:Erl}, 
we assume the parameters $\omega_1$ and $\nu_1$ to satisfy the inequality of point 1 above: $\abs{\omega_1}> \nu_1/2$. Then, from \eqref{HeffK2}, \eqref{K1)}, \eqref{S1psi} we get
\[
\psi(s)=\begin{pmatrix} \cos z_1s -\frac {\nu_1}{4z_1}\,\sin z_1 s \\ -\frac{\rmi \omega_1 }{2z_1}\, \sin z_1s\end{pmatrix}, \qquad \langle \psi(t-t_0)|P_1\psi(t-t_0)\rangle=\left(\cos z_1(t-t_0) -\frac {\nu_1}{4z_1}\,\sin z_1 (t-t_0) \right)^2,
\]
\[
\norm{\psi(t-t_0)}^2=\left(\cos z_1(t-t_0) -\frac {\nu_1}{4z_1}\,\sin z_1 (t-t_0) \right)^2+ \left(\frac{\omega_1}{2z_1}\,\sin z_1(t-t_0)\right)^2,
\]
and the explicit expressions for all the quantities in \eqref{proptonorm} follow. All these quantities have a strongly oscillating behaviour and, in particular, the intensity $\lambda_1(t)$ has infinitely many zeros at the times $t^*$ satisfying
\[
\tan \big(z_1(t^*-t_0)\big)=\frac{4z_1}{\nu_1}\equiv \sqrt{\frac{4{\omega_1}^2}{\nu_1}-1};
\]
correspondingly, the probability density \eqref {probdens2} vanishes in the points $s^*=t^*-t_0$. Note that at the same times we have $\rho(t^*)=P_0$. 

In this simple model, after a of type $x(0)\mapsto x(1)$, a jump of type $x(1)\mapsto x(0)$ follows; now, by \eqref{atajump2}, the quantum state just after the jump is $P_0$. The behaviour of the system up to the first jump can be obtained as in the previous case. So, we have an alternate sequence of two types of jumps with two waiting times distributions with densities $p_{T_k}(s)$, $k=0,1$; the situation is very reminiscent of the models of Sec.\ \ref{sec:inters}, but now the waiting time distributions are not predetermined, but depend on the dynamics of the quantum component.

\section{Conclusions}\label{sec:end} 
By using the idea of quantum/classical hybrid systems, we have given a general formulation of quantum stochastic master equations of jump type (Sec.\ \ref{sec:sigmat}). This formulation allows to include not only detection in continuous time, feedback, some forms of memory, but also other fields of research: non-Hermitian evolutions (Sec.\ \ref{sec:nonH}), piecewise evolutions interspersed by jump operators (Sec.\ \ref{sec:inters}), Lindblad rate equation and continuous time open quantum walks (Sec.\ \ref{sec:walks}). The key ingredients are the notions of typical trajectories and of waiting time of the next jump (Sec.\ \ref{sec:ttrajs}). 

The notion of typical trajectory allows to introduce the ``exclusive probability densities" (Sec.\ \ref{sec:exclpd}), from which all the probability distributions involved in the theory can be obtained. Moreover, typical trajectories are involved also in the recursive construction of the solution of the quantum dynamics (Secs.\ \ref{sec:solvSME}, \ref{sec:interspex}, \ref{sec:mstinter}).

The notion of waiting time of the next jump is typical of a formulation based on counts (Sec.\ \ref{sec:nojprob}) and we have used it as a unifying feature in all the applications. As for all probabilities involved in the theory, the waiting time distribution can be obtained by using positive operator-valued measures, as required in the general formulation of a quantum theory (Sec.\ \ref{sec:POVM}). These distributions appear as survival probabilities in the theory of non-Hermitian evolutions (Sec.\ \ref{sec:survpr}) and as starting point in the theory of a piecewise dynamics (Sec.\ \ref{sec:cppre}). In the final application to quantum/classical walks, we show how very general waiting time distributions can be obtained (Secs.\ \ref{sec:WTwalks}, \ref{sec:2dex}), despite the Markovian character of the composed hybrid system.

\appendix  

\section{Computations of the probabilities of Remarks \ref{rem:exp} and \ref{rem:Erl}}\label{App:comput}

For the model of Section \ref{sec:revivals}, the following relations hold:
\begin{equation*}
\Ecal_\pm^{2m}=\Ecal_\pm^2, \quad \Ecal_\pm^{2m+1}=\Ecal_\pm^2, \quad m\geq 1, \qquad \Ecal_0\circ \Ecal_\pm= \Ecal_0\circ \Ecal_1, \quad \Ecal_1\circ \Ecal_\pm= \Ecal_1\circ \Ecal_0, \quad \Ecal_i\circ \Ecal_\pm^2= \Ecal_i.
\end{equation*}
By using \eqref{matrixrho0}, we have
\[
\Ecal_0[\rho_0]=\Ecal_0\circ \Ecal_\pm^2[\rho_0]=\begin{pmatrix} \rho_{00} &0 \\ 0 &0\end{pmatrix}, \qquad \Ecal_1[\rho_0]=\Ecal_1\circ \Ecal_\pm^2[\rho_0]=\begin{pmatrix}0&0 \\ 0 &  \rho_{11} \end{pmatrix},
\]
\[
\Ecal_\pm[\rho_0]=\begin{pmatrix} \rho_{00}&0 \\ 0 &  \rho_{11} \end{pmatrix}, \qquad \Ecal_0\circ \Ecal_\pm[\rho_0]=\begin{pmatrix}\rho_{11}&0 \\ 0 & 0 \end{pmatrix},  \qquad \Ecal_1\circ \Ecal_\pm[\rho_0]=\begin{pmatrix}0&0 \\ 0 & \rho_{00} \end{pmatrix},
\]
\[
\Tr\{\Ecal_i[\rho_0]\}=\Tr\{\Ecal_i\circ \Ecal_\pm^2[\rho_0]\}= \rho_{ii} , \qquad \Tr\{\Ecal_i\circ \Ecal_\pm[\rho_0]\}= 1-\rho_{ii}.
\]

\subsection{The exponential case}

From \eqref{eta,tilde} we have
\[
\tilde \eta_{2m+1}(t_{2m+1})=\lambda\rme^{-\lambda t_{2m+1}}\, \frac{(\lambda t_{2m+1})^{2m}}{(2m)!}\,\Ecal_\pm[\rho_0],\qquad m\geq 0, \]
\[
\tilde \eta_{2m}(t_{2m})=\lambda\rme^{-\lambda t_{2m}}\, \frac{(\lambda t_{2m})^{2m-1}}{(2m-1)!}\,\Ecal_\pm^2[\rho_0],\qquad m\geq 1;
\]
then, from \eqref{eta-t0} we get \eqref{etat0exp}. 
By this result and equations \eqref{P_iN(t)}, we get the probabilities \eqref{Piexp} and \eqref {N(t)exp}. Equations \eqref{Koldisexp} follow by direct computations from \eqref{Piexp} and \eqref{Kolmdist}.

\subsection{The Erlang case}
From \eqref{eta,tilde} and \eqref{eta-t0} we have
\[
\tilde \eta_{2m+1}(t_{2m+1})=\lambda^{4m+2}\rme^{-\lambda t_{2m+1}}\, \frac{t_{2m+1}^{4m+1}}{(4m+1)!}\,\Ecal_\pm[\rho_0],\qquad m\geq 0,
\]
\[
\tilde \eta_{2m}(t_{2m})=\lambda^{4m}\rme^{-\lambda t_{2m}}\, \frac{t_{2m}^{4m-1}}{(4m-1)!}\,\Ecal_\pm^2[\rho_0],\qquad m\geq 1,
\]
\begin{multline*}
\eta(t_0)=(1+\lambda t_0)\rme^{-\lambda t_0}(\rho_0-\Ecal_\pm^2[\rho_0])+\sum_{m=0}^{+\infty} \rme^{-\lambda t_0}\left(\frac{(\lambda t_0)^{4m}}{(4m)!}+\frac{(\lambda t_0)^{4m+1}} {(4m+1)!}\right)\Ecal_\pm^2[\rho_0]
\\ {}+ \sum_{m=0}^{+\infty} \rme^{-\lambda t_0}\left(\frac{(\lambda t_0)^{4m+2}}{(4m+2)!}+\frac{(\lambda t_0)^{4m+3}} {(4m+3)!}\right)\Ecal_\pm[\rho_0].
\end{multline*}
By using 
\begin{equation}\label{series}\begin{split}
\frac 12 \left(\cosh x +\cos x \right)= \sum_{m=0}^{+\infty} \frac {x^{4m}}{(4m)!}, \qquad &\frac 12 \left(\cosh x -\cos x \right)= \sum_{m=0}^{+\infty} \frac {x^{4m+2}}{(4m+2)!},
\\
\frac 12 \left(\sinh x +\sin x \right)= \sum_{m=0}^{+\infty} \frac {x^{4m+1}}{(4m+1)!}, \qquad &\frac 12 \left(\sinh x -\sin x \right)= \sum_{m=0}^{+\infty} \frac {x^{4m+3}}{(4m+3)!},
\end{split}\end{equation}
we get \eqref{etat0Erl}.

From \eqref{P_iN(t)} we obtain
\begin{multline*}
\Pbb_i(t_0,t|\rho_0)= 
\int_{t_0}^{t} \rmd s \lambda^2 \rme^{-\lambda s}\Big\{ s \rho_{ii}+\sum_{m=0}^{+\infty}\int_0^{t_0}\rmd t_{2m+1}(s-t_{2m+1}) \lambda^{4m+2}\, \frac{t_{2m+1}^{4m+1}}{(4m+1)!}\left(1-\rho_{ii}\right)
\\ {}+\sum_{m=1}^{+\infty}\int_0^{t_0}\rmd t_{2m} \left(s-t_{2m}\right)\lambda^{4m}\, \frac{t_{2m}^{4m-1}}{(4m-1)!}\rho_{ii}\Big\} 
\\ {}= 
\int_{t_0}^{t} \rmd s \lambda \rme^{-\lambda s}\Big\{ \sum_{m=0}^{+\infty}\left(\lambda(s-t_0)\frac{(\lambda t_{0})^{4m+2}}{(4m+2)!} + \frac{(\lambda t_{0})^{4m+3}}{(4m+3)!}\right)  \left(1-\rho_{ii}\right)
\\ {}+\sum_{m=0}^{+\infty} \left(\lambda(s-t_0)\frac{(\lambda t_{0})^{4m}}{(4m)!} + \frac{(\lambda t_{0})^{4m+1}}{(4m+1)!}\right) \rho_{ii}\Big\} 
\\ {}=\rme^{-\lambda t_0} \sum_{m=0}^{+\infty}\biggl\{\biggl(\left(1-\rme^{-\lambda(t-t_0)}-\lambda(t-t_0) \rme^{-\lambda(t-t_0)}\right)\frac{(\lambda t_{0})^{4m+2}}{(4m+2)!} + \left(1-\rme^{-\lambda(t-t_0)}\right)\frac{(\lambda t_{0})^{4m+3}}{(4m+3)!}\biggr)  
\\ {}\times \left(1-\rho_{ii}\right) +\biggl(\left(1-\rme^{-\lambda(t-t_0)}-\lambda(t-t_0) \rme^{-\lambda(t-t_0)}\right)\frac{(\lambda t_{0})^{4m}}{(4m)!} + \left(1-\rme^{-\lambda(t-t_0)}\right)\frac{(\lambda t_{0})^{4m+1}}{(4m+1)!}\biggr)\rho_{ii}\biggr\},
\end{multline*}
\begin{multline*}
\Pbb_i(t_0,t|\rho_0)=\frac {\rme^{-\lambda t_0}} 2\biggl\{\biggl(\left(1-\rme^{-\lambda(t-t_0)}-\lambda(t-t_0) \rme^{-\lambda(t-t_0)}\right)\left(\cosh \lambda t_0 -\cos \lambda t_0 \right) 
\\ {}+ \left(1-\rme^{-\lambda(t-t_0)}\right) \left(\sinh \lambda t_0 -\sin \lambda t_0 \right)\biggr)  \left(1-\rho_{ii}\right)
  +\biggl(\left(1-\rme^{-\lambda(t-t_0)}-\lambda(t-t_0) \rme^{-\lambda(t-t_0)}\right)\\ {}\times \left(\cosh \lambda t_0 +\cos \lambda t_0 \right) + \left(1-\rme^{-\lambda(t-t_0)}\right) \left(\sinh \lambda t_0 +\sin \lambda t_0 \right) \biggr)\rho_{ii}\biggr\} 
\\ {}=\frac {\rme^{-\lambda t_0}} 2\biggl\{\left(1-\rme^{-\lambda(t-t_0)}\right) \left(\rme^{ \lambda t_0}+\left(\cos \lambda t_0+\sin \lambda t_0 \right) \left(2\rho_{ii}-1\right)\right)
\\ {}  -\lambda(t-t_0) \rme^{-\lambda(t-t_0)} \left(\cosh \lambda t_0 +\cos \lambda t_0 \left(2\rho_{ii}-1\right)\right) \biggr\} ;
\end{multline*}
this gives \eqref{PiErl}. By applying \eqref{Ngeq1}, we get \eqref{N(t)Erl}. Finally, eqs.\ \eqref{2PiErl}--\eqref{limDtraceErl} follow by direct computations.

\section{Computations for Section \ref{sec:nonEP}}\label{sec:computC4}
\subsection{Proof of Remark \ref{rem:4:1}}\label{proof:4:1}
By expanding the exponential we have immediately the result \eqref{evolgpm}:
\[
\rme^{-\rmi Kt}=\sum_{n=0}^{+\infty}\frac{(-\rmi Kt)^{2n}}{(2n)!}+ \sum_{n=0}^{+\infty}\frac{(-\rmi Kt)^{2n+1}}{(2n+1)!} =\sum_{n=0}^{+\infty}\frac{(-\varkappa t^2)^n}{(2n)!}\,\openone+ \sum_{n=0}^{+\infty}\frac{(-\varkappa t^2)^{n}}{(2n+1)!}\left(-\rmi K t\right).
\]
By using the fact that the modulus of a sum is less than the sum of the moduli, we get easily the convergence result:
\[
\abs{\sum_{n=0}^{+\infty}\frac{(-\xi)^n}{(2n+1)!} }\leq \sum_{n=0}^{+\infty}\abs{\frac{(-\xi)^n}{(2n+1)!}} \leq \sum_{n=0}^{+\infty}\frac{\abs\xi^n}{(2n)!} = \cosh\big(\sqrt{\abs \xi}\big) .
\]

By using $-\varkappa t^2= (-\rmi zt)^2$, we get 
\[
g_+(\varkappa t^2)=\sum_{n=0}^{+\infty}\frac{(-\varkappa t^2)^n}{(2n)!} =\sum_{n=0}^{+\infty}\frac{(-\rmi zt)^{2n}}{(2n)!} =\frac 12 \left(\rme^{\rmi zt}+\rme^{-\rmi zt}\right),
\]
\[
g_-(\varkappa t^2)=  \sum_{n=0}^{+\infty}\frac{(-\varkappa t^2)^{n}}{(2n+1)!} =\frac 1 {\rmi zt} \sum_{n=0}^{+\infty}\frac{(\rmi zt)^{2n+1}}{(2n+1)!} =\frac 1{2\rmi zt} \left(\rme^{\rmi zt}-\rme^{-\rmi zt}\right).
\]
This proves \eqref{gpmz}.

\subsection{Proof of Proposition \ref{prop:eigenvec}}\label{sec:proof2}
It is easy to check that both the expressions in \eqref{upm} satisfy \eqref{eigeneq}; we have to show that the two expressions are equal. By using $\alpha^2-z^2=-\beta\delta$, we get
\[
\frac\Phi{\sqrt{\abs{\alpha\mp z}^2+\abs \beta^2}}=\frac {\left(\overline \alpha \mp\overline z\right)\delta} {\abs{\alpha\mp z}\sqrt{\abs\delta^2\abs{\alpha\mp z}^2+\abs{\alpha\pm z}^2\abs{\alpha\mp z}^2}}=\frac {\left(\overline \alpha \mp\overline z\right)\delta} {\abs{\alpha\mp z}^2\sqrt{\abs\delta^2+\abs{\alpha\pm z}^2}},
\]
\[
\frac{ \delta\left(\overline \alpha \mp \overline z \right)}{\abs{\alpha\mp z}^2}(-\beta)= \frac{\left(\alpha\pm z\right)\left(\alpha\mp z\right)\left(\overline \alpha \mp \overline z \right)}{\abs{\alpha\mp z}^2}=\alpha\mp z,
\qquad
\frac{ \delta\left(\overline \alpha \mp \overline z \right)}{\abs{\alpha\mp z}^2}\left(\alpha\mp z\right)=\delta;
\]
this ends the proof of the equality of the two expressions. In the case $\beta\delta=0$ we can check directly that the two proposed expression satisfy \eqref{eigeneq}.

By using the definition given in \eqref{Vpmprop}, we get immediately the second equivalence in \eqref{orthogonal}. The last condition, together with the other properties in \eqref{Vpmprop}, says that the operators $V_\pm$ are orthogonal projections, which project on orthogonal spaces; this implies $\langle u_+|u_-\rangle =0 $. So, we have to prove that the first condition is sufficient to get the other ones. In the case $\beta\delta\neq 0$, we take $u_\pm$ from the first expression in \eqref{upm} and $\langle u_+|u_-\rangle =0$ gives
\[
\left(\overline \alpha +\overline z \right)\left( \alpha - z\right)+\abs \delta^2=0 \qquad \Leftrightarrow \qquad \begin{cases} \overline z\, \alpha= z\overline \alpha,\\ \abs\alpha^2 +\delta^2=\abs\varkappa;\end{cases}
\]
we have used $z^2=\varkappa$ from the definition \eqref{defz}. Now we set $a:=\overline z\, \alpha$, \ $b:=\overline z\, \beta$, \ $d:=\overline z\, \delta$; we have $a\in\Rbb$ and
\[
\varkappa=\frac{a^2 + bd}{{\overline z }^2} =\frac{a^2 + bd}{\overline \varkappa }, \qquad \abs{\varkappa}^2= a^2+bd, \qquad \abs\varkappa= \abs\alpha^2 +\delta^2= \frac{a^2 +\abs d^2}{\abs{\overline z}^2}=\frac{a^2 +\abs d^2}{\abs \varkappa}.
\]
Finally we have $a^2+bd=a^2 +\abs d^2$, which gives $\overline d =b$. This gives 
\[
K=\frac 1 {\overline z}\begin{pmatrix}a & b \\ \overline b & -a\end{pmatrix};
\]
being $a\in\Rbb$ we get that $K/z$ is Hermitian and this ends the proof.


\begin{thebibliography}{999}
\bibitem{Bar86} A.~Barchielli, \textsl{Measurement
    theory and stochastic differential equations in quantum mechanics}, \href{http://pra.aps.org/abstract/PRA/v34/i3/p1642_1}{Phys.\ Rev.\ A {\bf 34} (1986) 1642--1649}.
\bibitem{Bel88} V.P. Belavkin, \textsl{Nondemolition measurements, nonlinear filtering and dynamic programming of quantum stochastic processes}. In A. Blaqui\`ere (ed.), \textit{Modelling and Control of Systems}, \href{https://link.springer.com/chapter/10.1007/BFb0041197#citeas}{Lecture Notes in Control and Information Sciences, vol.\ 121 (Springer, Berlin, 1988) pp.\ 245--265}.

\bibitem{Bel89} V.P. Belavkin, \textsl{A new wave equation for a continuous nondemolition measurement}, \href{https://doi.org/10.1016/0375-9601(89)90066-2}{Phys. Lett. A \textbf{140} (1989) 355-358}.

\bibitem{BarB91} A. Barchielli, V.P. Belavkin,    \textsl{Measurements    continuous in time and   a posteriori states in quantum mechanics},    \href{http://iopscience.iop.org/0305-4470/24/7/022}{J.\ Phys.\ A: Math.\ Gen. {\bf24} (1991) 1495--1514.}

\bibitem{Bar93} A. Barchielli,    \textsl{On the    quantum theory of measurements continuous in time}, \href{http://dx.doi.org/10.1016/0034-4877(93)90037-F}{Rep.\ Math.\    Phys.\ {\bf 33} (1993) 21--34.}
\bibitem{Car93} H. Carmichael, \textit{An Open System Approach to Quantum Optics}, Lect. Notes in Physics \textbf{m 81} (Springer, Berlin, 1993).
\bibitem{BarP96} A.~Barchielli, A.~M.~Paganoni, \textsl{Detection    theory in quantum optics: Stochastic representation}, \href{http://dx.doi.org/10.1088/1355-5111/8/1/011}{Quantum  Semiclass.\ Opt.\ {\bf8} (1996) 133--156}. 
\bibitem{ZolG97}  P. Zoller and C.W. Gardiner, \textsl{Quantum noise in quantum optics: the stochastic Schr\"odinger equation}. In S. Reynaud, E. Giacobino \& J. Zinn-Justin eds., \textit{Fluctuations quantiques, (Les Houches 1995)} (North-Holland, Amsterdam, 1997) pp. 79--136.
\bibitem{Hol01} A.S. Holevo, \textit{Statistical Structure of Quantum Theory}, Lecture Notes in Physics m 67 (Springer, Berlin, 2001).

\bibitem{KumM03}  B. K\"ummerer and H. Maassen, \textsl{An ergodic theorem for quantum counting processes}, \href{https://doi.org/10.1088/0305-4470/36/8/312}{J. Phys. A 36 (2003) 2155.}

\bibitem{Bar06} A. Barchielli, \textsl{Continual Measurements in Quantum Mechanics and  Quantum Stochastic Calculus}.    In S. Attal, A. Joye, C.-A. Pillet (eds.),   \textit{Open Quantum Systems III}, \href{http://www.springerlink.com/content/2623j57n1t7w0131/}{Lect.\ Notes Math.\ \textbf{1882}    (Springer, Berlin, 2006), pp. 207--291}.
\bibitem{Car08} H. Carmichael, \textit{Statistical Methods in Quantum Optics 2 --- Non-Classical Fields}, (Springer, Berlin 2008).
\bibitem{BarG09} A. Barchielli, M. Gregoratti,     \textit{Quantum Trajectories and Measurements in Continuous Time: The Diffusive Case}, \href{http://www.springer.com/physics/quantum+physics/book/978-3-642-01297-6}{Lect.\ Notes Phys.\ \textbf{782} (Springer, Berlin \& Heidelberg,  2009)}.

\bibitem{WisM10} H.M. Wiseman and G.J. Milburn, \textit{Quantum Measurement and  Control} (Cambridge University Press, Cambridge, 2010).

\bibitem{Jac14} K. Jacobs, \textit{Quantum Measurement Theory and its Applications} \href{https://doi.org/10.1017/CBO9781139179027}{(Cambridge University Press, 2014)}.
\bibitem{Caves+18} J.A. Gross, C.M. Caves, G.J. Milburn, J. Combes, \textsl{Qubit models of weak continuous measurements: Markovian
conditional and open-system dynamics}, \href{https://doi.org/10.1088/2058-9565/aaa39f}{Quantum Sci. Technol. \textbf{3} (2018) 024005}.
\bibitem{CavJ23} C.S. Jackson, C.M. Caves, \textsl{Simultaneous measurements of noncommuting observables:
Positive transformations and instrumental Lie groups}, \href{https://doi.org/10.3390/e25091254}{Entropy \textbf{25} (2023) 1254}. 

\bibitem{Maa24} H. Maassen,  \textsl{Continuous observation of quantum systems}, \href{https://doi.org/10.1142/S0219749924400112}{Int. J. Quantum Inf. 22 (2024) 2440011.} Also in \href{https://doi.org/10.1142/14142}{B.-G. Englert and R. Werner (eds.), \textsl{Communicating the Quantum Way --- Contributions in Honor of Alexander S Holevo's 80th Birthday} (World Scientific, 2025)}, \href{https://doi.org/10.1142/9789819806607_0012}{pp. 261--299.}
  
\bibitem{Potts+24}  G.T. Landi, M.J. Kewming, M.T. Mitchison, P.P. Potts, \textsl{Current fluctuations in open quantum systems: Bridging the gap between quantum continuous measurements and full counting statistics}, \href{https://doi.org/10.1103/PRXQuantum.5.020201}{PRX QUANTUM \textbf{5} (2024) 020201}.
\bibitem{Jak25}  C.S. Jackson, \textsl{Sequential quantum measurements and the instrumental group algebra}, \href{https://doi.org/10.3390/quantum7040057}{Quantum Rep. \textbf{7} (2025) 57}.

\bibitem{DGS00} L. Di\'osi, N. Gisin,  W.T. Strunz, \textsl{Quantum approach to coupling classical and quantum dynamics}, \href{https://journals.aps.org/pra/abstract/10.1103/PhysRevA.61.022108}{Phys. Rev. A \textbf{61} (2000) 022108.}

\bibitem{Dio14} L. Di\'osi, \textsl{Hybrid quantum-classical master equations}, \href{10.1088/0031-8949/2014/T163/014004}{Phys. Scr. \textbf{T163} (2014) 014004}.

\bibitem{Dio23} L. Di\'osi, \textsl{Hybrid completely positive Markovian quantum-classical dynamics},  \href{https://journals.aps.org/pra/abstract/10.1103/PhysRevA.107.062206}{Phys. Rev. A \textbf{107} (2023)  062206}; \textsl{Erratum}, \href{https://doi.org/10.1103/PhysRevA.108.059902}{Phys. Rev. A \textbf{108} (2023) 059902(E)}.

\bibitem{ManRT23} G. Manfredi, A. Rittaud, C. Tronci, \textsl{Hybrid quantum-classical dynamics of pure-dephasing systems}, \href{https://doi.org/10.1088/1751-8121/acc21e}{J. Phys. A: Math. Theor. \textbf{56} (2023) 154002}. 

\bibitem{Pomar+23} J.L. Alonso, C. Bouthelier-Madre, J. Clemente-Gallardo, D. Mart\'inez-Crespo, J. Pomar, \textsl{Effective nonlinear Ehrenfest hybrid quantum-classical dynamics}, \href{https://doi.org/10.1140/epjp/s13360-023-04266-w}{Eur. Phys. J. Plus \textbf{138} (2023) 649}.

\bibitem{Sergi+23} A. Sergi, D. Lamberto, A. Migliore, A. Messina, \textsl{Quantum–classical hybrid systems and Ehrenfest's Theorem}, \href{https://doi.org/10.3390/e25040602}{Entropy \textbf{25} (2023) 602}. 
\bibitem{Opp+23} J. Oppenheim, C. Sparaciari, B. Šoda, Z. Weller-Davies, \textsl{Objective trajectories in hybrid classical-quantum dynamics},
\href{https://doi.org/10.22331/q-2023-01-03-891}{Quantum \textbf{7} (2023)  891}. 

\bibitem{DamWer23} L. Dammeier, R.F. Werner, \textsl{Quantum-classical hybrid systems and their quasifree transformations}, \href{https://doi.org/10.22331/q-2023-07-26-1068}{Quantum \textbf{7} (2023) 1068}.

\bibitem{Bar23} A. Barchielli, \textsl{Markovian master equations for quantum-classical hybrid systems}, \href{https://doi.org/10.1016/j.physleta.2023.129230}{Phys. Lett. A \textbf{492} (2023) 129230}.

\bibitem{LOW24} I. Layton, J. Oppenheim, Z. Weller-Davies, \textsl{A healthier semi-classical dynamics}, 
\href{https://doi.org/10.22331/q-2024-12-16-1565}{Quantum 8 (2024) 1565.}

\bibitem{BarW24}  A. Barchielli, R. Werner, \textsl{Hybrid quantum-classical systems: Quasi-free Markovian dynamics}, \href{https://doi.org/10.1142/S0219749924400021}{Int. J. Quantum Inf. \textbf{22} (2024) 2440002}. Also in \href{https://doi.org/10.1142/14142}{B.-G. Englert and R. Werner (eds.), \textsl{Communicating the Quantum Way --- Contributions in Honor of Alexander S Holevo's 80th Birthday} (World Scientific, 2025)}, \href{https://doi.org/10.1142/9789819806607_0011}{pp. 211-259.}
\bibitem{Bar24} A. Barchielli, \textsl{Markovian dynamics for a quantum/classical system and quantum trajectories}, \href{https://doi.org/10.1088/1751-8121/ad5fd2}{J. Phys. A: Math. Theor. \textbf{57} (2024) 315301.}

\bibitem{TicNA13} F. Ticozzi, K. Nishio, C. Altafini, \textsl{Stabilization of stochastic quantum dynamics via
 open- and closed-loop control},  \href{https://doi.org/10.1109/TAC.2012.2206713}{IEEE Trans. on Automatic Control \textbf{58} (2013) 74--85}.
\bibitem{BarG12}  A. Barchielli, M. Gregoratti, \textsl{Quantum measurements in continuous time, non-Markovian evolutions and feedback}, \href{http://rsta.royalsocietypublishing.org/content/370/1979/5364.abstract}{Phil. Trans.   R. Soc. A \textbf{370} no. 1979 (2012) 5364--5385}.
\bibitem{Gough20} J.E. Gough, \textsl{A quantum Kalman filter-based PID controller}, \href{https://doi.org/10.1142/S1230161220500146}{Open Systems and Information Dynamics \textbf{27} (2020) 2050014.}
\bibitem{Till24} A. Tilloy, \textsl{General quantum-classical dynamics as measurement based feedback}, \href{https://scipost.org/SciPostPhys.17.3.083}{SciPost Phys. \textbf{17} (2024) 083}.
\bibitem{TicAlt12} C. Altafini, F. Ticozzi, \textsl{Modeling and control of quantum systems: An introduction},  \href{https://doi.org/10.1109/TAC.2012.2195830}{IEEE Trans. on Automatic Control \textbf{57} (2012) 1898--1917}.
\bibitem{Potts+22}
B. Annby-Andersson, F. Bakhshinezhad, D. Bhattacharyya, G. De Sousa, C. Jarzynski, P. Samuelsson, P.P. Potts, \textsl{Quantum Fokker-Planck Master Equation for Continuous Feedback Control}, \href{https://doi.org/10.1103/PhysRevLett.129.050401}{Phys. Rev. Lett. 129 (2022) 050401}. 

\bibitem{AlbGen23}  F. Albarelli, M.G. Genoni, \textsl{A pedagogical introduction to continuously monitored quantum systems
and measurement-based feedback}, \href{https://doi.org/10.1016/j.physleta.2023.129260}{Phys. Lett. A \textbf{494} (2024) 129260}.

\bibitem{Potts+25} A.J.B. Rosal, G. Fiusa, P.P. Potts, G.T. Landi, \textsl{Deterministic equations for feedback control of open quantum systems III:
Full counting statistics for jump-based feedback}, \href{https://doi.org/10.48550/arXiv.2512.11078}{arXiv:2512.11078 (2025)}.


\bibitem{Dav69} E.B. Davies, \textit{Quantum stochastic processes}, \href{https://link.springer.com/article/10.1007/BF01645529}{Commun. Math. Phys. \textbf{15} (1969) 277--304}.
\bibitem{Dav76} E.B. Davies, \textit{Quantum Theory of Open Systems} (Academic, London, 1976).
\bibitem{SriD81} M.D. Srinivas, E.B. Davies, \textsl{Photon counting probabilities in quantum optics}  \href{https://doi.org/10.1080/713820643}{Opt. Acta \textbf{28} (1981) 981-996}.
\bibitem{SriD82} M.D. Srinivas, E.B. Davies, \textsl{What are the photon counting probabilities for open systems --- A reply to Mandel's comments}, \href{https://doi.org/10.1080/713820838}{Opt. Acta \textbf{29} (1982)  235-238}.
\bibitem{Bar90} A.~Barchielli,    \textsl{Direct and
    heterodyne detection and other applications of quantum stochastic    calculus to quantum optics}, \href{http://iopscience.iop.org/0954-8998/2/6/002}{Quantum Opt.\ {\bf2} (1990) 423--441}.
\bibitem{Bar97} A. Barchielli,  \textsl{On
    the quantum theory of direct detection}. In O. Hirota, A.S.
    Holevo, C.M. Caves (eds.),    \href{http://link.springer.com/book/10.1007/978-1-4615-5923-8}
    {\textit{Quantum Communication, Computing, and Measurement}} (Plenum    Press, New York, 1997) \href{http://link.springer.com/chapter/10.1007/978-1-4615-5923-8_26}{pp.\ 243--252}.

\bibitem{KW20} D. Keys, J. Wehr, \textsl{Poisson stochastic master equation unravelings and the measurement problem: A quantum stochastic calculus perspective}, \href{https://doi.org/10.1063/1.5133974}{J. Math. Phys. \textbf{61} (2020) 032101}; \textsl{Erratum}, \href{https://doi.org/10.1063/5.0124270}{J. Math. Phys. \textbf{63} (2022) 109901.}
\bibitem{DM-G22} B. Donvil, P. Muratore-Ginanneschi, \textsl{Quantum trajectory framework for general time-local master equations}, \href{https://doi.org/10.1038/s41467-022-31533-8}{Nature Commun. \textbf{13} (2022) 4140}.
\bibitem{SLCSPVac24} F. Settimo, K. Luoma, D. Chruściński, B. Vacchini, A. Smirne, J. Piilo, \textsl{Generalized-rate-operator quantum jumps via realization-dependent transformations}, \href{10.1103/PhysRevA.109.062201}{Phys. Rev. A \textbf{109} (2024) 062201}.
\bibitem{CLPSm25} D. Chruściński, K. Luoma, J. Piilo, A. Smirne, \textsl{How to design quantum-jump trajectories via distinct master equation representations}, \href{https://doi.org/10.22331/q-2022-10-13-835}{Quantum \textbf{6} (2022) 835}.
\bibitem{MCD93} K. Mølmer, Y. Castin, J. Dalibard,  \textsl{A Monte Carlo wave function method in quantum optics}, \href{https://doi.org/10.1364/JOSAB.10.000524}{J. Opt. Soc. Am. B \textbf{10} (1993) 524–538}.

\bibitem{RadLB24}  M. Radaelli, G.T. Landi, F.C. Binder, \textsl{Gillespie algorithm for quantum jump trajectories}, \href{https://doi.org/10.1103/PhysRevA.110.062212}{Phys. Rev. A \textbf{110} (2024) 062212}.
\bibitem{DaVJ08} D.J. Daley, D. Vere-Jones, \textit{An Introduction to the Theory of Point Processes --- Volume II: General Theory and Structure}, (Springer, 2008)
\bibitem{LMMR25}  S. Lloyd, L. Maccone, L. Martellini, S. Roncallo, \textsl{Quantum stroboscopy for time measurements}, \href{https://doi.org/10.1103/58vf-m1yx}{Phys. Rev. Lett. (2026)}; \href{https://doi.org/10.48550/arXiv.2507.17740}{arXiv:2507.17740 (2025)}.

\bibitem{RGMor25}  S. Roy, S. Gupta,  G. Morigi, \textsl{Causality, localisation, and universality
 of monitored quantum walks with long-range hopping},  \href{https://doi.org/10.1103/rbtb-8d27}{Phys. Rev. E \textbf{112} (2025) 044146}.

\bibitem{HBZT25} T. Heine,  E. Barkai, K. Ziegler, S. Tornow, \textsl{Quantum walks: First hitting times with weak measurements}, \href{https://doi.org/10.48550/arXiv.2506.21168}{arXiv:2506.21168 [quant-ph]}.
\bibitem{SerZ13} A. Sergi, K.G. Zloshchastiev, \textsl{Non-Hermitian quantum dynamics of a two-level system and models of dissipative environments}, \href{https://doi.org/10.1142/S0217979213501634}{Int. J.  Mod. Phys. B \textbf{27} (2013) 1350163}.
\bibitem{Gao+15} T. Gao, E. Estrecho, K.Y. Bliokh, T.C.H. Liew, M.D. Fraser, S. Brodbeck, M. Kamp, C. Schneider, S. H\"ofling,
Y. Yamamoto, F. Nori, Y.S. Kivshar, A.G. Truscott, R.G. Dall,  E.A. Ostrovskaya, \textsl{Observation of non-Hermitian degeneracies
in a chaotic exciton-polariton billiard}, \href{https://doi.org/10.1038/nature15522}{Nature \textbf{526} (2015) 554–558.}

\bibitem{SCr15} S. Croke, \textsl{$\Pcal\Tcal$-symmetric Hamiltonians and their application in quantum information}, \href{https://doi.org/10.1103/PhysRevA.91.052113}{Phys. Rev. A \textbf{91} (2015) 052113.}

\bibitem{MMCN19} F. Minganti, A. Miranowicz, R.W. Chhajlany, F. Nori, \textsl{Quantum exceptional points of non-Hermitian Hamiltonians and Liouvillians: The effects of quantum jumps}, \href{https://doi.org/10.1103/PhysRevA.100.062131}{Phys. Rev. A \textbf{100} (2019) 062131.}
    
\bibitem{NAJM19} M. Naghiloo, M. Abbasi1, Y.N. Joglekar, K.W. Murch, \textsl{Quantum state tomography across the 
exceptional point in a single dissipative qubit}, \href{https://doi.org/10.1038/s41567-019-0652-z}{Nature Phys. \textbf{15} (2019) 1232–1236.}

\bibitem{AGU20} Y. Ashida, Z. Gong, M. Ueda, \textsl{Non-Hermitian physics}, \href{https://doi.org/10.1080/00018732.2021.1876991}{Advances in Physics \textbf{69} (2020) 249–435.}
\bibitem{DBernardinD21} V. Dubey, C. Bernardin, A. Dhar, \textsl{Quantum dynamics under continuous projective measurements:
Non-Hermitian description and the continuum-space limit}, \href{https://doi.org/10.1103/PhysRevA.103.032221}{Phys. Rev. A \textbf{103} (2021) 032221}.

\bibitem{NonS25} R. Nongthombam, A.K. Sarma, \textsl{Homodyne measurement of a non-Hermitian qubit undergoing fluorescence}, \href{https://doi.org/10.48550/arXiv.2510.13345}{arXiv:2510.13345 [quant-ph]  (2025)}

\bibitem{NonVs25} R. Nongthombam, A. Verma, A.K. Sarma, \textsl{Role of inefficient measurement in realizing post-selection-based non-Hermitian qubits}, \href{https://doi.org/10.48550/arXiv.2510.21159}{arXiv:2510.21159 (2025).}
\bibitem{Longhi25} S. Longhi, \textsl{Phase transitions and virtual exceptional points in quantum emitters coupled to dissipative baths}, \href{https://doi.org/10.1063/5.0299681}{J. Appl. Phys. \textbf{138} (2025) 184401}.
\bibitem{Ho+26} K. Ho, S. Perna, S. Wittrock, Nhat-Tan Phan, S. Tsunegi, H. Kubota, S. Yuasa, P. Bortolotti, M. d’Aquino,  C. Serpico, V. Cros, R. Lebrun, \textsl{Encirclement of an exceptional point and eigenvalue switch in non-Hermitian coupled spintronic nano-oscillators}, \href{https://doi.org/10.1016/j.newton.2025.100333}{Newton \textbf{2} (2026) 100333.}
    
\bibitem{DMD26} Duttatreya, Ipsika Mohanty, Sanjib Dey, \textsl{Improved coherence time of a non-Hermitian qubit in a PT-symmetric environment}, \href{https://doi.org/10.1016/j.aop.2025.170298}{Ann. Phys. \textbf{484} (2026) 170298}.

\bibitem{GoHa06}  I. Goychuk, P. H\"anggi, \textsl{Quantum two-state dynamics driven by stationary non-Markovian discrete noise:
 Exact results}, \href{https://doi.org/10.1016/j.chemphys.2005.11.026}{Chemical Physics \textbf{324} (2006)  160-171.}
\bibitem{ScLu07} H. Schomerus, E. Lutz, \textsl{ Nonexponential decoherence and momentum subdiffusion in a quantum Lévy kicked rotator},  \href{https://doi.org/10.1103/PhysRevLett.98.260401}{PRL \textbf{98} (2007) 260401}.

\bibitem{DDG22} D. Das, S. Dattagupta,  S. Gupta, \textsl{Quantum unitary evolution interspersed with repeated non-unitary interactions at random
 times: the method of stochastic Liouville equation, and two examples of interactions in the context of a tight-binding chain}, \href{https://doi.org/10.1088/1742-5468/ac6256}{J. Stat. Mech. (2022) 053101}.

\bibitem{NaGu23} A. Nagar, S. Gupta, \textsl{Stochastic resetting in interacting particle systems: a review}, \href{https://doi.org/10.1088/1751-8121/acda6c}{J. Phys. A: Math. Theor. 56 (2023) 283001}.
\bibitem{Bud04} A.A. Budini, \textsl{Stochastic representation of a class of non-Markovian completely positive evolutions}, \href{https://doi.org/10.1103/PhysRevA.69.042107}{Phys. Rev. A \textbf{69} (2004) 042107}.
\bibitem{Vacc+11} B. Vacchini, A. Smirne, E.-M. Laine, J. Piilo, H.-P. Breuer, \textsl{Markovianity and non-Markovianity in quantum and
classical systems}, \href{https://doi.org/10.1088/1367-2630/13/9/093004}{New J. Phys. \textbf{13} (2011) 093004}.
\bibitem{Vacc13} B. Vacchini, \textsl{Non-Markovian master equations from piecewise dynamics}, \href{https://doi.org/10.1103/PhysRevA.87.030101}{Phys. Rev. A \textbf{87} (2013) 030101(R)}.

\bibitem{Vacc20} B. Vacchini, \textsl{Quantum renewal processes}, \href{https://doi.org/10.1038/s41598-020-62260-z}{Scientific Reports \textbf{10} (2020) 5592}.
\bibitem{Vacc+21} N. Megier, M. Ponzi, A. Smirne, B. Vacchini, \textsl{Memory effects in quantum dynamics modelled by quantum
renewal processes}, \href{https://doi.org/10.3390/e23070905}{Entropy \textbf{23} (2021) 905}. 

\bibitem{Bud06} A.A. Budini, \textsl{Lindblad rate equations}, \href{ https://doi.org/10.1103/PhysRevA.74.053815}{Phys. Rev. A \textbf{74} (2006)  053815}.
\bibitem{BrGM06} H.-P. Breuer, J. Gemmer, M. Michel, \textsl{Non-Markovian quantum dynamics: Correlated projection superoperators and Hilbert space averaging}, \href{https://doi.org/10.1103/PhysRevE.73.016139}{Phys. Rev. E \textbf{73} (2006) 016139}.
\bibitem{Breuer07} H.-P. Breuer, \textsl{Non-Markovian generalization of the Lindblad theory of open quantum system}, \href{https://doi.org/10.1103/PhysRevA.75.022103}{Phys. Rev. A \textbf{75} (2007) 022103}.

\bibitem{Pell14}  C. Pellegrini,  \textsl{Continuous time open quantum random walks and non-Markovian Lindblad
 master equations}, \href{https://doi.org/10.1007/s10955-013-0910-x}{Journal of Statistical Physics \textbf{154} (2014) 838-865}.
\bibitem{CLRB17} Chaobin Liu, Radhakrishnan Balu, \textsl{Steady states of continuous-time open quantum walks}, \href{https://doi.org/10.1007/s11128-017-1625-8}{Quantum Inf. Process 16 (2017) 173}.
\bibitem{Bri18}  H. Bringuier, \textsl{Central limit theorem and large deviation principle for continuous time open quantum walks}, \href{https://doi.org/10.1007/s00023-017-0597-7}{Annales Henri Poincaré 18(10) (2018) 3167-3192}.
\bibitem{BBPP19} I. Bardet, H. Bringuier, Y. Pautrat, C. Pellegrini, \textsl{Recurrence and transience of continuous-time open quantum walks}. In: Donati-Martin, C., Lejay, A., Rouault, A. (eds) Séminaire de Probabilités L. Lecture Notes in Mathematics, vol 2252 (2019). Springer, Cham. 
\bibitem{Kang19} Y.B. Kang, \textsl{Quantum Markov semigroups for continuous-time open quantum random walk}, \href{https://doi.org/10.1007/s11128-019-2294-6}{Quantum Information Processing (2019) 18:196}.

\bibitem{Loeb23} N.  Loebens, \textsl{Site recurrence for continuous-time open quantum walks on the line}, \href{https://doi.org/10.26421/QIC23.7-8-3}
{Quantum Information and Computation \textbf{23} (2023) 0577-0602}.
\bibitem{Loe24a} N. Loebens, \textsl{Continuous-time open quantum walks in one dimension: matrix-valued orthogonal polynomials and Lindblad
 generators}, \href{https://doi.org/10.1007/s11128-024-04303-2}{Quantum Information Processing (2024) 23:96}.

\bibitem{Loe24b} N. Loebens, \textsl{Open quantum jump chain for a class of continuous-time open quantum walks}, \href{https://doi.org/10.1007/s40509-024-00331-w}{Quantum Stud.: Math. Found. (2024) 11:459–476}. 
   
\bibitem{BarH95} A. Barchielli, A.S. Holevo, \textsl{Constructing
    quantum measurement processes via classical stochastic calculus}, \href{http://dx.doi.org/10.1016/0304-4149(95)00011-U}{Stoch. Proc. Appl. {\bf 58} (1995) 293--317}.
\bibitem{BarPZ98} A. Barchielli, A.M. Paganoni, F. Zucca,   \textsl{On stochastic differential equations and semigroups of probability
    operators in quantum probability},  \href{http://dx.doi.org/10.1016/S0304-4149(97)00093-8}{Stoch. Proc.\ Appl.\ {\bf73} (1998)    69--86.}
\bibitem{Met82} M. M\'etivier, \textit{Semimartingales, a Course on Stochastic Processes} (W. de Gruyter, Berlin, 1982).
\bibitem{Prott04} P.E. Protter, \textit{Stochastic Integration and Differential Equations}, II Edition (Springer, Berlin, 2004)

\bibitem{IWat89} N. Ikeda, S. Watanabe, \textit{Stochastic Differential Equations and Diffusion Processes}, Second Edition (North-Holland Publishing Company, Amsterdam, 1989).

\bibitem{LipS86}
R.Sh. Liptser, A.N. Shiryayev, \textit{Theory of Martingales}, (Kluwer Academic Publishers, Dordrecht, 1986).
\bibitem{GKS76} V. Gorini, A. Kossakowski, E.C.G. Sudarshan, \textsl{Completely positive dynamical semigroups of N-level systems}, \href{https://doi.org/10.1063/1.522979}{J. Math. Phys. \textbf{17} (1976) 821--825}. 
\bibitem{L76} G. Lindblad, \textsl{On the Generators of Quantum Dynamical Semigroups}, \href{https://link.springer.com/article/10.1007/BF01608499}{Commun. Math. Phys. \textbf{48} (1976) 119--130}. 
\bibitem{BarPP12} A.\ Barchielli, C.\ Pellegrini, F.\ Petruccione, \textsl{Quantum trajectories: Memory and continuous observation},
    \href{http://link.aps.org/doi/10.1103/PhysRevA.86.063814}{Phys. Rev. A \textbf{86}, 063814 (2012)}.
\bibitem{AttP06} S. Attal, Y. Pautrat, \textsl{From repeated to continuous quantum interactions}, \href{https://doi.org/10.1007/s00023-005-0242-8}{Ann.  Henri Poincar\'e \textbf{7} (2006) 59--104}.
\bibitem{CLGP22} F. Ciccarello, S. Lorenzo, V. Giovannetti, G.M. Palma, \textsl{Quantum collision models: open system dynamics from repeated interactions}, \href{https://doi.org/10.1016/j.physrep.2022.01.001}{Phys. Rep. \textbf{954} (2022)}.

\bibitem{Hol86} A.S. Holevo,  \textsl{Conditionally positive definite functions in quantum probability}, Proc. Int. Cong. Math. 1011-1020 (1986) Berkeley.
\bibitem{BarHL93} A. Barchielli, A.S. Holevo, G. Lupieri,  \textsl{An  analogue of Hunt's representation theorem in quantum probability}, \href{https://link.springer.com/article/10.1007/BF01047573} {J.\  Theor.\ Probab.\ {\bf 6} (1993) 231--265}.     

\bibitem{P+23} C. Bernardin, R. Chetrite, R. Chhaibi,  J. Najnude, C. Pellegrini, \textsl{Spiking and collapsing in large noise limits of SDEs}, \href{https://doi.org/10.1214/22-AAP1819}{Ann. Appl. Probab. \textbf{33} (2023) 417-446}.

\bibitem{BP+21} T. Benoist, C. Bernardin, R. Chetrite, R. Chhaibi,  J. Najnude, C. Pellegrini, \textsl{Emergence of jumps in quantum trajectories via homogenization}, \href{https://doi.org/10.1007/s00220-021-04179-8}{Commun. Math. Phys. \textbf{387} (2021) 1821–1867}.

\bibitem{SBDKC25}  A. Sherry, C. Bernardin, A. Dhar, A. Kundu, R. Chetrite, \textsl{Spikes in Poissonian quantum trajectories}, \href{https://doi.org/10.1103/PhysRevA.111.042215}{Phys. Rev. A 111 (2025) 042215}.

\bibitem{Fazio26} L. Lumia, E. Tirrito, M. Collura, F.H.L. Essler, R. Fazio, \textsl{Complexity of quantum trajectories}, \href{https://doi.org/10.48550/arXiv.2602.00232}{arXiv:2602.00232} (2026).

\bibitem{Nori26} P. Menczel, C. Flindt, F. Brange, F. Nori, C. Gneiting, \textsl{Full counting statistics and first-passage times in quantum markovian
processes: Ensemble relations, metastability, and fluctuation theorems}, \href{http://dx.doi.org/10.1103/4hc1-l8w4}{PRX Quantum \textbf{7} (2026) 010304.}
\bibitem{Vacc14} B. Vacchini, \textsl{General structure of quantum collisional models}, \href{https://doi.org/10.1142/S0219749914610115}{Int. J. Quantum Inf.
\textbf{12} (2014) 1461011}.

\bibitem{Vacc16} B. Vacchini, \textsl{Generalized master equations leading to completely positive dynamics}, \href{https://doi.org/10.1103/PhysRevLett.117.230401}{PRL \textbf{117} (2016) 230401}.

\bibitem{MSV20} N. Megier, A. Smirne, B. Vacchini, \textsl{Evolution equations for quantum semi-Markov dynamics}, \href{https://doi.org/10.3390/e22070796}{Entropy \textbf{22} (2020) 796}.
\bibitem{Bud13a} A.A. Budini, \textsl{Embedding non-Markovian quantum collisional models into bipartite Markovian dynamics}, \href{https://doi.org/10.1103/PhysRevA.88.032115}{Phys. Rev A \textbf{88} (2013) 032115}.
\bibitem{Bud13b} A.A. Budini, \textsl{Non-Markovian quantum jumps from measurements in bipartite Markovian dynamics}, \href{https://doi.org/10.1103/PhysRevA.88.012124}{Phys. Rev A \textbf{88} (2013) 012124}.
    
\bibitem{ClaB71} M.J. Clauser, M. Blume, \textsl{Stochastic theory of line shape: Off-diagonal effects in fine and hyperfine structures}, \href{https://doi.org/10.1103/PhysRevB.3.583}{Phys. Rev. B \textbf{3} (1971) 583--591}.
 
\bibitem{MSM18} B. Mukherjee, K. Sengupta, S.N. Majumdar, \textsl{Quantum dynamics with stochastic reset}, \href{https://doi.org/10.1103/PhysRevB.98.104309}{Phys. Rev. B \textbf{98} (2018) 104309}.
\bibitem{DDG22b} S. Dattagupta,  D. Das, S. Gupta, \textsl{Stochastic resets in the context of a tight-binding chain driven by an oscillating field}, \href{https://doi.org/10.1088/1742-5468/ac98c0}{J. Stat. Mech. (2022) 103210}.

\bibitem{VKB08}  M. Varbanov, H. Krovi, T.A. Brun, \textsl{Hitting time for the continuous quantum walk}, \href{https://doi.org/10.1103/PhysRevA.78.022324}{Phys. Rev. A \textbf{78} (2008) 022324}.

\bibitem{BreuerLP09} H.-P. Breuer, E.-M. Laine, J. Piilo, \textsl{Measure for the degree of non-Markovian behavior of quantum processes in open systems}, \href{https://doi.org/10.1103/PhysRevLett.103.210401}{PRL \textbf{103} (2009) 210401}.

\bibitem{APS12}  S. Attal, F. Petruccione, I. Sinayskiy, \textsl{Open quantum walks on graphs}, \href{https://doi.org/10.1016/j.physleta.2012.03.040}{Phys. Lett. A \textbf{376} (2012) 1545-1548}.

\bibitem{CarP15} R. Carbone, Y. Pautrat, \textsl{Homogeneous open quantum random walks on a lattice}, \href{https://doi.org/10.1007/s10955-015-1261-6}{J. Stat. Phys. \textbf{160} (2015) 1125-1153}.

\bibitem{CGMh22} R. Carbone, F. Girotti, A. Melchor Hernandez, \textsl{On a Generalized Central Limit theorem and large deviations for homogeneous open quantum walks}, \href{https://doi.org/10.1007/s10955-022-02938-y}{Journal of Statistical Physics (2022) 188:8}.

\bibitem{BarP10} A.~Barchielli, C.~Pellegrini,    \textsl{Jump-diffusion    unravelling of a non Markovian generalized Lindblad master equation},
    \href{http://link.aip.org/link/JMAPAQ/v51/i11/p112104/s1}{J.\ Math.\ Phys.\ \textbf{51} (2010) 112104}.

\bibitem{SSP17} V. Semin, I. Semina, F. Petruccione, \textsl{Stochastic wave-function unravelling of the generalized Lindblad equation}, \href{https://doi.org/10.1103/PhysRevE.96.063313}{Phys. Rev. E \textbf{96} (2017) 063313}.
\bibitem{Bud26} A.A. Budini, \textsl{Hybrid quantum-classical dynamics with stationary thermal states}, \href{https://doi.org/10.48550/arXiv.2604.02484}{arXiv:2604.02484 [quant-ph]}.

\bibitem{IglL24}  M.D. de la Iglesia, and C.F. Lardizabal, \textsl{One-dimensional continuous-time quantum Markov chains: qubit probabilities and measures}, \href{https://doi.org/10.1088/1751-8121/ad5bcb}{J. Phys. A: Math. Theor. \textbf{57} (2024)  295301}.

\end{thebibliography}
\end{document}